\documentclass[a4paper,reqno]{article}

\usepackage{latexsym}
\usepackage[english]{babel}
\usepackage{fancyhdr}
\usepackage[mathscr]{eucal}
\usepackage{amsmath}
\usepackage{mathrsfs}
\usepackage{amsthm}
\usepackage[left=2.5cm, right=2.5cm, top=2.5cm, bottom=2cm]{geometry}
\usepackage{amsfonts}
\usepackage{amssymb}
\usepackage{amscd}
\usepackage{bbm}
\usepackage{graphicx}
\usepackage{graphics}
\usepackage{latexsym}
\usepackage{color}
\usepackage{calc}
\usepackage{cancel}
\usepackage{marvosym}
\usepackage{tikz}
\usetikzlibrary{patterns, decorations.markings,decorations.pathmorphing,arrows, calc}
\usepackage[normalem]{ulem}

\newcommand{\ud}{\mathrm{d}}

\newcommand{\ii}{\mathrm{i}}
\newcommand{\cH}{\mathcal{H}}

\newcommand{\tb}{\mathtt{b}}

\newcommand{\N}{\mathbb N}
\newcommand{\R}{\mathbb R}
\newcommand{\Z}{\mathbb Z}
\newcommand{\T}{\mathbb T}
\newcommand{\strong}{\mathtt{s}}
\newcommand{\uv}{\mathtt{uv}}
\newcommand{\ir}{\mathtt{ir}}
\newcommand{\llan}{\langle\! \langle}
\newcommand{\rran}{\rangle\! \rangle}
\newcommand{\inv}{\mathrm{inv}}
\newcommand{\obs}{\mathrm{obs}}
\newcommand{\eff}{\mathrm{eff}}
\newcommand{\ccccc}{\mathtt{c}}

\newcommand{\odd}{\mathrm{odd}}
\newcommand{\even}{\mathrm{even}}

\definecolor{azzurro}{rgb}{0.19, 0.55, 0.91}
\definecolor{giallino}{rgb}{0.97,	0.78,	0.05}
\definecolor{verdino}{rgb}{0.49, 0.74, 0.54}

\definecolor{verdissimo}{rgb}{0.0, 0.5, 0.0}

\theoremstyle{plain}
\newtheorem{theorem}{Theorem}[section]
\newtheorem{lemma}[theorem]{Lemma}
\newtheorem{corollary}[theorem]{Corollary}
\newtheorem{proposition}[theorem]{Proposition}

\usepackage{tocloft}

\renewenvironment{thebibliography}[1]{
  \begin{oldthebibliography}{#1}
    \setlength{\itemsep}{0.1em}
    \setlength{\parskip}{0em}
}
{
  \end{oldthebibliography}
}

\theoremstyle{definition}

\newtheorem{remark}[theorem]{Remark}
\newtheorem*{remark*}{Remark}

\numberwithin{equation}{section}

\addtocontents{toc}{\protect\setcounter{tocdepth}{2}}

\begin{document}

%

\title{\textbf{{Prethermalization and Conservation Laws in Quasi-Periodically Driven Quantum Systems}}}

\author{Matteo Gallone\footnote{\emph{Email:} \texttt{matteo.gallone@sissa.it}} $\,$ and Beatrice Langella\footnote{\emph{Email:} \texttt{beatrice.langella@sissa.it}} \\ \vspace{-10pt} \\
\small \textit{International School for Advanced Studies -- SISSA }\\ \vspace{-18pt} \\ \small \textit{via Bonomea 265 -- 34136 Trieste (Italy)}.}

%
%


\date{\today}




\maketitle

\vspace{-10pt}

\begin{abstract}
We study conservation laws of a general class of quantum many-body systems subjected to an external time dependent quasi-periodic driving. {When the frequency of the driving is large enough or the strength of the driving is small enough, we prove a Nekhoroshev-type stability result: we show that the system exhibits a prethermal state for stretched exponentially long times in the perturbative parameter}. Moreover, we prove the quasi-conservation of the constants of motion of the unperturbed Hamiltonian and we analyze their physical meaning in examples of relevance to condensed matter and statistical physics.
\end{abstract}


\tableofcontents



\section{Introduction}
The investigation of physics of many-body systems out of thermal equilibrium has recently attained a considerable amount of attention. Due to new developments in manipulating quantum systems, fundamental questions related to the thermalization process can be investigated both theoretically and experimentally. Moreover, in the last few years, the establishment of the paradigm of \emph{prethermalization} opened an avenue in the experimental realization of new, exotic out-of-equilibrium phases of matter. Mathematically rigorous investigations on this phenomenon have also been initiated, especially in the context of quantum many-body systems in presence of external drivings \cite{Abanin2017,DeRoeck-Verreet,Else2020}.

With this motivation, in the present paper we consider a class of quantum many-body systems on a lattice whose dynamics is generated by a time dependent Hamiltonian $H(t)$ of the form
\begin{equation}\label{H0}
H(t)\;=\;H_0+P( \omega t)\,,
\end{equation} 
where $\omega \in \R^m$ and $P(\omega t)$ is regarded as a perturbation of $H_0$. We consider perturbations that are quasi-periodic in time and either small in size, i.e. $P \sim \epsilon$, or having high frequency, i.e. $\omega \sim \epsilon^{-1}$, $\epsilon \ll 1$. We focus on the case where the unperturbed operator admits the structure
\begin{equation}\label{JdotN}
H_0\;=\;\sum_{\alpha=1}^r J_\alpha N^{(\alpha)} \;=:\; J \cdot N\,,
\end{equation}
where $J_1, \dots, J_r \in \mathbb{R}$, and $N^{(1)}, \dots,  N^{(r)}$ are local\footnote{With a slight abuse of terminology, we say that an extensive operator is local when it can be written as the sum of local terms, the support of which does not grow with volume.}, mutually commuting, extensive self-adjoint operators with integer spectrum. In this setting, in absence of the perturbation, the dynamics conserves each of the $N^{(\alpha)}$'s separately. Under suitable non resonance and strong locality hypotheses, we prove that, for the two classes of perturbations delined above, each of the $N^{(\alpha)}$'s is \emph{quasi-conserved} for a \emph{stretched exponentially long time} in the small perturbative parameter $\epsilon$. All our results hold uniformly in the size of the lattice. Moreover, we construct an effective \emph{time independent} \emph{local} Hamiltonian $H_\eff$ whose flow approximates very well the dynamics of local observables in the following sense. If $O$ is a local observable and $U$ denotes the unitary operator representing the time evolution generated by $H(t)$, then
\begin{itemize}
\vspace{-6pt}
	\item[(i)] for all $t \lesssim e^{\epsilon^{-c}}$
	\begin{equation}
		\Vert U^*(t) O U(t)-e^{-\ii H_\eff t} O e^{\ii H_\eff t} \Vert_{\mathrm{op}} \;\lesssim\; \epsilon^a \, ,
	\end{equation}
\vspace{-12pt}
	\item[(ii)] there exists an infinite sequence of times $\{t_j\}_{j \in \N}$, $t_j \sim j e^{f \epsilon^{-c}}$ for which this vicinity of dynamics is actually stretched exponentially small, namely
	\begin{equation}
		\Vert U^*(t_j) O U(t_j)-e^{-\ii H_\eff t_j} O e^{\ii H_\eff t_j} \Vert_{\mathrm{op}} \;\lesssim\; j^{d+2} e^{-\epsilon^{-c}} \, ,
	\end{equation}
\end{itemize}
\vspace{-6pt}
	where $d$ is the dimension of the lattice and $a, f, c>0$ are independent of $\epsilon$ and of the size of the lattice.
	
The conservation of each of the $N^{(\alpha)}$'s encodes information on the \emph{collective dynamical behavior} of the model during the prethermal phase. Practically, this suggests that in the absence of other conserved quantities, a good candidate state for predictions in the prethermal phase is given by the Generalized Gibbs Ensemble $\rho_{GGE} = e^{-\sum_{j=1}^{r} \lambda_j N^{(j)}}/Z_{GGE}$ with $Z_{GGE}=\mathrm{Tr}(\rho_{GGE})$ \cite{Kollar2011}.

Paradigmatic models that fit the assumptions of our analysis are the quantum Ising chain in $d$ spatial dimensions (with or without external magnetic field) and generalizations of the Fermi-Hubbard Hamiltonian where $H_0$ is given by the interaction and $P(\omega t)$ by the kinetic energy. Leaving a more detailed discussion on these two models to Subsection \ref{subsec:GranchiettiCarini}, here we briefly discuss the consequences of our result on the dynamics of the one-dimensional quantum Ising chain with an external magnetic field. The magnetic field has two components: a time independent one along the $(3)$ axis, space-periodic with period $2$ lattice sites, and a quasi-periodically time dependent one along the $(1)$ axis. The Hamiltonian of the model is
\begin{equation}
	H(t)\;=\; - J \sum_{x \in \Lambda} \sigma_x^{(3)} \sigma_{x+1}^{(3)} - J_{\odd} \sum_{x \in \Lambda_{\odd}} \sigma_x^{(3)} -J_{\even} \sum_{x \in \Lambda_{\even}} \sigma_x^{(3)} -\epsilon B(\omega_1 t,\omega_2 t) \sum_{x \in \Lambda} \sigma_x^{(1)}  \, .
\end{equation}
When the perturbation is turned off ($\epsilon=0$) the number operators
\begin{equation}
	N_{\odd}:=\sum_{x \in \Lambda_\odd} \sigma_x^{(3)} \, , \qquad N_{\even}:=\sum_{x \in \Lambda_\even} \sigma_x^{(3)}\,, 
\end{equation}
which are the total magnetizations on even and odd sites, are conserved. 

When the perturbation is turned on, $N_\odd$ and $N_\even$ are not conserved anymore. 
A naive analysis suggests that $N_\odd$ and $N_\even$ are quasi-conserved for times $t \lesssim \epsilon^{-1}$.  Our result improves significantly this bound: if the vector $(J,J_\odd,J_\even,\omega_1,\omega_2)$ is non resonant, then  $N_\odd$ and $N_\even$	are quasi-conserved for stretched exponentially long times in $\epsilon$. In practice, this means that if in the initial state of the system all sites with spin up are in the odd sublattice, the number of sites with spin up on the even sublattice is $\sim \epsilon^a |\Lambda|$ for times that are stretched exponentially long in $\epsilon$.

Perturbations of Hamiltonians $H_0$ of the form \eqref{JdotN} were first studied by De Roeck and Verreet \cite{DeRoeck-Verreet} in the small size, time periodic case, proving quasi-conservation of the $N^{(\alpha)}$ operators and the existence of an effective \emph{time dependent} Hamiltonian $H_\eff(t)$ whose flow approximates the evolution of local observables for stretched exponentially long time.

\smallskip

In our work we prove that the quasi-conservation of the $N^{(\alpha)}$'s persists under quasi-periodic perturbations. This  extends what is found in \cite{DeRoeck-Verreet} for the small size time periodic case and is completely new in the case of high frequency perturbations. Moreover, we conjugate the dynamics, up to exponentially small errors, to the one of a \emph{time independent} Hamiltonian $H_\eff$ quasi-commuting with all the operators $N^{(\alpha)}$'s, in contrast with the general expectations expressed in \cite{DeRoeck-Verreet}. To obtain this result, we implement a new normal form procedure with respect to the ones developed in the previous works in the field \cite{Abanin2017,DeRoeck-Verreet,Else2020}, deeply based on the spectral properties of the operators $N^{(\alpha)}$. In addition, we prove that the phenomenon of stroboscopic times originally found by Abanin et al.\ \cite{Abanin2017} in the time periodic setting can also be observed in systems with quasi-periodic driving. Such a result is new both for fast forcing and for small size perturbations. It requires the combination of Lieb-Robinson bounds \cite{Lieb1972}, in order to control the time evolution of local observables, with a refined analysis of the ergodization time on the torus, for which we apply the results of \cite{Berti2003}.

\vspace{5pt}

\noindent
{\textbf{Prethermalization in classical and quantum systems.}
In general, prethermalization refers to the fact that certain physical systems, when initialized far from equilibrium, quickly relax to long-living non-thermal states, before eventually reaching the thermal state on much longer timescales \cite{Berges2004,Moeckel2008,Gring2012,Bertini2015,Lindner2017,Mori2016,Mallayya2019,Howell2019,Huveneers2020,RubioAbadal2020,Collura2022}. Often, during the prethermal dynamics, the time-evolution of the system is well-approximated by an effective Hamiltonian with some additional symmetries that are broken in the Hamiltonian of the system. This is the case, for example, of the first model on which prethermalization was observed: the Fermi-Pasta-Ulam-Tsingou chain \cite{FPU-Original}. For this nonintegrable model, the dynamics in the prethermal phase (in the past, this has often been referred to as the ``metastable state'') has been related to the dynamics of two integrable systems: the Toda chain \cite{Ferguson1982,BambusiMaspero2016,Benettin2023-vw}, and the Korteweg-de Vries hierarchy of equations  \cite{Zabusky1965,Bambusi2006,Gallone2021,Gallone2022PRL}. Rigorous prethermalization results for this model holding in the thermodynamic limit rely either on estimates of autocorrelation of the energies of Fourier modes \cite{Maiocchi2014} or on the quasi-conservation of bunches of Toda actions as adiabatic invariants \cite{Grava2020}. In general, the dynamics of a model in its prethermal phase is not necessarily described by an integrable Hamiltonian. This is the case, for example, of certain systems subjected to external drivings, mathematically encoded in a time dependent Hamiltonian: in the cases when the system prethermalizes, during the prethermal state the dynamics is often well-approximated by the dynamics generated by a time independent Hamiltonian, but this in general guarantees an approximate conservation of energy only. 

For many-body quantum systems, an effective protocol to produce prethermal states with certain controlled features is given by external drivings. It has been predicted theoretically that in absence of Many-Body Localization, the time dependent perturbation increases the energy of the system, which eventually reaches thermal equilibrium in a featureless, infinite-temperature Gibbs state \cite{Ponte2015,DAlessio2014,Lazarides2014}. In the particular case of a fast periodic driving, the heating rate of the system is suppressed and for a very long time, which is exponential in the frequency \cite{Abanin2017}, the system remains in a prethermal state. This prethermal state often presents interesting, non trivial features, and it is now customary to refer to the phases arising in this way as \emph{Floquet phases of matter} \cite{Potter2016,Ye2021,Eckhardt2022,Zhang2022,Zhang2022e}. 

In the case of quasi-periodic drivings, recent theoretical works analyzed the possibility of creating prethermal states which exhibit non-trivial thermodynamic features that are different from the Floquet ones \cite{Martin2017PRX,Lapierre2020,Long2021,Qi2021,Zhao2021,Zhao2022bis,Martin2022,Long2022PRB}.
Experimental works confirmed this possibility  \cite{Boyers2020-exp,Malz2021,Guanghui}.

\vspace{5pt}

\noindent
\textbf{Main novelties and related literature.}
As mentioned above, our work stems from a recent series of rigorous results on prethermalization for quantum many-body systems subjected to time periodic or time quasi-periodic driving.  Abanin et al.\ in \cite{Abanin2017} exhibit one observable $H_\eff$ which is quasi-conserved for stretched exponentially long times in systems with time periodic driving, and prove the existence of stroboscopic times $t_j \sim j$, $j \in \Z$, at which the dynamics generated by $H_\eff$ approximates the time evolution of local observables up to exponentially small errors. Else, Ho and Dumitrescu in \cite{Else2020} extend the existence result of one quasi-conserved quantity $H_\eff$ proven in \cite{Abanin2017} to the case of Hamiltonians of the form \eqref{H0}, with $ P(\omega t)$ a fast quasi-periodic forcing and without any further assumption on the unperturbed Hamiltonian $H_0$. De Roeck and Verreet in \cite{DeRoeck-Verreet} treat models of the form $H(t) = H_0 + \epsilon P(t)$, with $\epsilon P(t)$ a small size time periodic driving and $H_0$ as in \eqref{JdotN}. They prove quasi-conservation of the $N^{(\alpha)}$'s separately and the existence of an effective time dependent Hamiltonian $H_\eff(t)$ approximating the evolution of local observables with polynomial errors for exponentially long times.


In our work we extend both the results of \cite{Abanin2017} and of \cite{DeRoeck-Verreet}, {which hold for the time periodic case,} showing that quasi-conservation of the $N^{(\alpha)}$'s, the existence of an effective Hamiltonian whose flow approximates time evolution of local observables, and the phenomenon of recurrence times all persist under time quasi-periodic perturbations, either small in size or fast forcing. {We also extend the results in \cite{Else2020}, in the sense that, assuming the more specific structure \eqref{JdotN} for the unperturbed Hamiltonian, we exhibit the existence of several conserved quantities, instead of only one. Furthermore, we also treat the case of small size time quasi-periodic perturbations, which was not covered in the above papers. Up to our knowledge, ours is the first existence result of a prethermal Hamiltonian with perturbations in such class. We remark that none of the improvements that we present could be obtained with the normal form techniques of \cite{Abanin2017,DeRoeck-Verreet,Else2020}.}


\smallskip

In finite dimensional quasi-integrable classical systems, non convergent normal forms have been largely used to prove confinement of the dynamics for exponentially long times. Among all, we only mention here the milestone result given by Nekhoroshev Theorem \cite{Nek77, Nek79}, see also \cite{Poschel_nek, Lochak}, under assumptions of convexity or, more generally, steepness, of the unperturbed Hamiltonian $H_0$ in the actions (for a simplified exposition in the quadratic case see also \cite{GioPisa,NekCInfinito}), and the result in \cite{Benettin_Gallavotti},
which deals with the linear, non steep case $H_0 = J \cdot n$, where $n_1, \dots, n_r$ are classical action variables and $J \in \R^r$ is a non resonant vector.

We point out that, from a mathematical point of view, the main interest of the normal form we perform is that we require the smallness conditions on the perturbative parameter to be \emph{uniform} with respect to the volume of the lattice $\Lambda$. The typical situation is that the Hamiltonian $H(t)$ of a many-body system of particles on a lattice of volume $|\Lambda|$ has operator norm growing linearly with $|\Lambda|$; for such a reason, one cannot implement a normal form reducing at every step the size of the time dependent perturbation in operator norm, since this would yield empty results in the limit $|\Lambda| \rightarrow \infty$. For a prethermalization result on a quantum system \emph{without uniformity} in the volume $|\Lambda|$, see the work \cite{Monti_Jauslin}.
Instead, in order to obtain uniformity in the parameter $|\Lambda|$, following the ideas in \cite{Abanin2017}, we exploit the fact that $H(t)$ can be written as the sum of local operators and all its local parts have operator norm independent of $|\Lambda|$. Then one defines a suitable norm, which takes into account these properties and turns out to be uniform in $|\Lambda|$ in all our applications. Similar attempts in this direction were previously made in \cite{PonPon} (see also \cite{Wayne1, Wayne3, Benettin_Fro_Giorgilli}), where a KAM scheme is implemented for several classes of classical systems, among which arbitrary long chains of weakly coupled harmonic oscillators. The work \cite{PonPon} actually presents several choices of norms, suitably tuned to  optimize the quantitative aspects of the theory with respect to the number of degrees of freedom of the system, with the same idea that here we pursue of ``considering perturbations not as a single chunk but rather as
composites of smaller pieces reflecting an underlying spatial structure''. However, none of the norms presented therein suites to our analysis, as they would not allow us to assume smallness of the perturbation uniformly in $|\Lambda|$. Again in the classical case, a very similar idea based on the decomposition of functions on a lattice in terms of their local parts is implemented in \cite{Carati_Maiocchi}.

Finally, for finite and infinite dimensional KAM schemes dealing with fast quasi-periodic frequencies, we mention the works \cite{Baldi_Berti, Corsi_Genovese, Franzoi_Maspero, Franzoi}.

Our results hold under suitable non resonance hypotheses on the energy levels $J$ of the operators $N^{(\alpha)}$ and on the frequencies $\omega$ of the time forcing. Such conditions differ according to the class of perturbations we deal with. In the fast forced case, due to the fact that $|\omega| \gg |J|$, it is sufficient to require separately that the vectors $J$ and $\omega$ are non resonant, and in particular Diophantine (see the definitions in \eqref{eq:siamo.dei.cani.ff.1} and \eqref{eq:siamo.dei.cani.ff.2}). Recall that this is not a restrictive assumption, since Diophantine vectors are a full measure set. Instead, in the case of small size perturbations we impose that the frequency of the forcing term does not resonate with the energy levels of the unperturbed Hamiltonian $H_0$, which amounts to require that $(J, \omega) \in \R^{r+m}$  is Diophantine. For an example of a prethermalization result obtained in a specific case where the vector $(J, \omega)$ is instead resonant (and without the presence of an effective time independent Hamiltonian), see the work \cite{Ho-DeRoeck-2020}. For a classical result in absence of non resonance conditions, we mention the work \cite{Dario1}, see also \cite{Dario2}, which deals with systems of high frequency harmonic oscillators coupled with a slow system. In these works the frequencies $J$ of oscillation may actually be resonant; as a counterpart, differently from our case, nothing prevents energy exchanges among the harmonic oscillators in the fast system.

\vspace{5pt}

\noindent
\textbf{Structure of the work.} In Section \ref{sec:MainRes} we present the functional setting of the problem and the main results of the paper: conservation of the number operators $N^{(\alpha)}$'s and the properties of the effective Hamiltonian. We then proceed to explicitly study the consequences of our results on a couple of physically relevant models: the Fermi-Hubbard model with next-to nearest neighbor interactions and the quantum Ising chain in a quasi-periodic transverse field.

In Section \ref{sec:NormalFormResults} we provide the statements of the normal form results, together with a description of the main ideas of the proof. In Section \ref{sec:Preliminaries} we discuss the properties of strongly local operators and provide the solution(s) to the homological equations that is the key point in the proof of normal form results, presented in Section \ref{sec:norma.normale}. Finally, in Section \ref{sec:physical.cons} we combine the normal form results, suitable Lieb-Robinson bounds and results from dynamical systems to deduce the main results of the paper.

\vspace{5pt}

\noindent
\textbf{Notation.} All the notation we use is standard. Nevertheless, for clarity of the reader, we would like to specify the following

\begin{tabular}{ll}
	$|k|_1$ & is the $\ell^1$ norm of the vector $k \in \mathbb{R}^q$ \\
	$|k|_\infty$ & is the $\ell^{\infty}$ norm of the vector $k \in \mathbb{R}^q$  \\
	$\langle \omega \rangle$ & if $\omega \in \mathbb{R}^m$, denotes the Japanese bracket: $\langle \omega \rangle=\max\{1,|\omega|_\infty\}$. \\
	$\langle V \rangle$ & instead, if $V$ is an operator, is defined in \eqref{zeta.esplicita} \\
	$\llan V \rran$ & if $V$ is an operator, is defined in \eqref{cacio.e.pepe} \\
	$\lfloor \lambda \rfloor$ & is the integer part of $\lambda \in \mathbb{R}$. 
\end{tabular}


\vspace{5pt}

\noindent
\textbf{Acknowledgments.} The authors would like to thank  Dario Bambusi, Massimiliano Berti, Federico Bonetto, Alberto Maspero, Vieri Mastropietro, Gianluca Panati, Antonio Ponno and Marcello Porta for interesting discussions and suggestions on the present work, Alessio Lerose for pointing to our attention part of the literature on prethermalization and for his precious advice, and Stefano Marcantoni for his stimulating remarks on the preliminary version of this work. \newline
\indent M.G.\ acknowledges financial support by the European Research Council (ERC) under the European Union’s Horizon 2020 research and innovation program ERC StG MaMBoQ, n.80290 {and by the MIUR-PRIN 2017 project MaQuMa cod.\ 2017ASFLJR}. B.L.\ acknowledges financial support by PRIN 2020XB3EFL, Hamiltonian and dispersive PDEs. This work was partially supported by GNFM (INdAM) the Italian National Group for Mathematical Physics.

\section{Setting and Main Results}\label{sec:MainRes}
We consider a system on a finite lattice $\Lambda=\mathbb{Z}^d \cap [-L,L]^d$, $L \gg 1$ in $d$ spatial dimensions. The dynamics is generated by a quasi-periodic time dependent self-adjoint Hamiltonian $H(\omega t)$ acting on the Hilbert space $\mathcal{H}_\Lambda:=\bigotimes_{x \in \Lambda} \mathbb{C}^q \simeq (\mathbb{C}^q)^{\otimes |\Lambda|}$ for some $q \in \N$. 

For our purposes, we need to recall {the} standard notion of \emph{locality}. Let $\mathscr{B}(\mathcal{H}_\Lambda)$ be the algebra of bounded operators on $\mathcal{H}_\Lambda$ with the usual operator norm, namely
\begin{equation}
\Vert A \Vert_{\mathrm{op}}\; := \; \sup_{\substack{\psi \in \mathcal{H}_\Lambda \\ \Vert \psi \Vert_{\mathcal{H}_\Lambda}=1}} \Vert A \psi \Vert_{\mathcal{H}_\Lambda} \, 
\end{equation}
for any $A \in \mathscr{B}(\mathcal{H}_\Lambda)$.
If $S \subset \Lambda$, we denote by $\mathcal{H}_S$ the Hilbert subspace obtained as $\mathcal{H}_S:=\bigotimes_{x \in S} \mathbb{C}^q$. An operator $A \in \mathscr{B}(\mathcal{H}_\Lambda)$ is said to be a \emph{local operator acting within  $S \subset \Lambda$} if there exists an operator $A_S \in \mathscr{B}(\mathcal{H}_S)$ such that $A=\mathbbm{1}_{\Lambda \setminus S} \otimes A_S$. We denote by $\mathcal{P}_c(\Lambda)$ the collection of all the connected subsets $S \subset \Lambda$. We define the set of quasi-local operators as the closure (with respect to the operator norm) of the set of finite linear combinations of local operators.
A collection of operators $\{A_S\}_{S \in \mathcal{P}_c(\Lambda)}$ defines an operator $A=\sum_{S \in \mathcal{P}_c(\Lambda)} A_S {\in \mathscr{B}(\mathcal{H}_\Lambda)}$.

Following \cite{Abanin2017}, given $\kappa\geq 0$, we define the norm 
\begin{equation}
\Vert A \Vert_{\kappa} \;:=\; \sup_{x \in \Lambda} \sum_{\substack{S \in \mathcal{P}_c(\Lambda) \\ x \in S}} \Vert A_S \Vert_{\mathrm{op}} \, e^{\kappa |S|} \, .
\end{equation}
{If there exists a constant $C>0$, independent of $|\Lambda|$, such that $\Vert A \Vert_\kappa < C$, we say that the collection $\{A_S\}_S \in \mathcal{O}_\kappa$. By abuse of notation, we also simply say that $A \in \mathcal{O}_\kappa$.}

Given two local operators $A_S$ and $B_{S'}$, we notice that if $S \cap S' = \varnothing$, then $[A_S,B_{S'}]=0$; otherwise if $S \cap S' \neq \varnothing$, there exists an operator $C_{S \cup S'} \in \mathscr{B}(\mathcal{H}_{S \cup S'})$ such that $[A_S,B_{S'}]=C_{S\cup S'} \otimes \mathbbm{1}_{\Lambda \setminus (S \cup S')}$. Then, given two operators of the form $A=\sum_{S' \in \mathcal{P}_c(\Lambda)} A_{S'}$, $B=\sum_{S'' \in \mathcal{P}_c(\Lambda)} B_{S''}$ one can consider for their commutator $[A,B]$ the following decomposition: 
\begin{equation}\label{eq:LocalCommutator}
[A, B]\;=\;\sum_{S \in \mathcal{P}_c(\Lambda)} [A,B]_S\,, \quad	[A,B]_S \;:=\; \sum_{\substack{S',S'' \in \mathcal{P}_c(\Lambda) \\ S' \cup S'' = S, \, S' \cap S'' \neq \varnothing}} [A_{S'},B_{S''}] \, .
\end{equation}

%

A time quasi-periodic family of quasi-local operators is a map $\R \ni t \mapsto \{A_S(\omega t)\}_{S \in \mathcal{P}_c(\Lambda)}$ such that $\omega \in \R^m$ for some $m\in \R$ and for any $S \in \mathcal{P}_c(\Lambda)$ the map $\T^m \ni \varphi \mapsto A_S(\varphi)$ is analytic. Then, one defines
\begin{equation}
	A(\omega t) \;:=\; \sum_{S \in \mathcal{P}_c(\Lambda)} A_S(\omega t) \, .
\end{equation} 
In the normal form procedure, we need to take into account the amount of regularity of the map $\varphi \mapsto A(\varphi)$. Thus, for $\kappa \geq 0$, $\rho \geq 0$, we introduce the norm
\begin{equation}
	\Vert A \Vert_{\kappa,\rho} \;:=\; \sup_{x \in \Lambda} \sum_{\substack{S \in \mathcal{P}_c(\Lambda) \\ x \in S}} \sum_{l \in \mathbb{Z}^m} \big\Vert (\widehat{A_S})_{l} \big\Vert_{\mathrm{op}} \, e^{\kappa |S|} e^{\rho |l|} \, ,
\end{equation}
where, for $l \in \mathbb{Z}^m$,  $(\widehat{A_S})_{l} := (2 \pi)^{-m} \int_{\mathbb{T}^m} e^{-\ii l \cdot \varphi} A_S(\varphi) \, \ud \varphi$. Again, by the same abuse of notation, we say that $A \in \mathcal{O}_{\kappa,\rho}$ {if there exists $C>0$, independent of $|\Lambda|$, such that $\Vert A \Vert_{\kappa,\rho} < C$. 

We point out that the classes of operators $\mathcal{O}_\kappa$ and $\mathcal{O}_{\kappa,\rho}$ contain a quite general class of physically relevant elements. Among them, translationally invariant Hamiltonians with finite range interaction, whose strength and range are independent of $|\Lambda|$. More generally, all operators $A \equiv \{ A_S\}_{S \in \mathcal{P}_c(\Lambda)}$ with
\begin{equation}\label{norma.alvise}
	\sup_{S \in \mathcal{P}_c(\Lambda)} \| A_S\|_{\mathrm{op}} \leq C_A \quad \mbox{and} \quad A_{S} \neq 0  \ \Rightarrow \  |S| \leq \overline{s}\,
\end{equation}
for some positive $C_A, \overline{s}$ independent of $|\Lambda|$, are such that $A \in \mathcal{O}_{\kappa}$. Indeed, one has
\begin{equation}
\|A\|_{\kappa} = \sup_{x \in \Lambda} \sum_{S \in \mathcal{P}_{c}(\Lambda) \atop S \ni x\,,\ |S| \leq \overline{s}} \|A_{S}\|_{\mathrm{op}} e^{\kappa|S|} \leq C_A \sup_{x \in \Lambda}  \sum_{S \in \mathcal{P}_{c}(\Lambda) \atop S \ni x\,,\ |S| \leq \overline{s}} e^{\kappa |S|} =  C_A \sum_{S \in \mathcal{P}_{c}(\Lambda) \atop S \ni 0\,,\ |S| \leq \overline{s}} e^{\kappa |S|} \leq C_A C_{\overline{s} ,\kappa, d}\,
\end{equation}
for some $C_{\overline{s}, \kappa, d}>0$. Note that the operators treated in Section \ref{subsec:GranchiettiCarini} are in $\mathcal{O}_{\kappa}$, since they satisfy \eqref{norma.alvise}, while they have ${\|A\|_{\mathrm{op}} \simeq |\Lambda|.}$}

\smallskip
\noindent\textbf{Hamiltonian of the model. }
As stated in the Introduction, we consider perturbations of the following unperturbed Hamiltonian:
\begin{equation}\label{Hzero}
	H_0 \;:=\; J \cdot N\; :=\; \sum_{\alpha=1}^r J_\alpha N^{(\alpha)}\,,
\end{equation}
where $J \in \mathbb{R}^r$ and $N=(N^{(1)},\dots,N^{(r)})$. Each of the $N^{(\alpha)}$ has the following properties:
\begin{itemize}
	\item[(N.i)] for any $\alpha\in\{1,\dots,r\}$, $N^{(\alpha)}$ is self-adjoint on $\mathcal{H}_\Lambda$;
	\item[(N.ii)] there exists $\kappa>0$ such that $N^{(\alpha)} \in \mathcal{O}_{\kappa}$ for any $\alpha \in \{1,\dots,r\}$;
 	\item[(N.iii)] $\left[N^{(\alpha)}_S, N^{(\alpha')}_{S'}\right] = 0 \quad \forall \alpha, \alpha' = 1, \dots, r,$ $\forall S, S' \in \mathcal{P}_c(\Lambda)$;
 	\item[(N.iv)] $N^{(\alpha)}$ has integer spectrum $\forall \alpha = 1, \dots, r$.
\end{itemize}
Due to property (N.iv), the operators $N^{(\alpha)}$ are called \emph{number operators}.
Following \cite{DeRoeck-Verreet}, we say that $A=\sum_{S \in \mathcal{P}_c(\Lambda)} A_S$ is a \emph{strongly {local} operator} (in symbols $A\in\mathcal{O}^\strong$) if {for any $\alpha=1,\dots,r$}
 \begin{equation}\label{strong.loc}
\forall S' \in \mathcal{P}_c(\Lambda) \quad \textrm{s.t.} \quad S' \nsubseteq S, \quad [A_{S}, N^{(\alpha)}_{S'}] = 0\,.
\end{equation}
Moreover, we write $A\in\mathcal{O}^\strong_\kappa$  if it is strongly {local} and there exists $\kappa >0$ such that $\|A\|_{\kappa}< \infty$.  We write $A\in\mathcal{O}^\strong_{\kappa,\rho}$ if there exists $\omega \in \R^m$ such that the map $t \mapsto A(\omega t)$ is quasi-periodic, for any $\varphi \in \mathbb{T}^m$, $A(\varphi)$ is strongly {local}, and there exist $\kappa, \rho>0$ such that $\|A\|_{\kappa, \rho}< \infty$ and \eqref{strong.loc} holds.

Here and in the following, given a family of self-adjoint operators $H(t)$ with quasi-periodic dependence on time, we will denote by $U_H(t)$ the (unique) solution of the equation
\begin{equation}
	U_H(t) \;=\; -\ii \int_0^t H(s) U_H(s) \, \ud s  {\, + U_{H}(0)}\, , \qquad U_H(0)=\mathbbm{1} \, .
\end{equation}

\subsection{Small-Size Perturbations} 
The first class of perturbations we consider is composed by small in size and quasi-periodic operators; that is, we consider the self-adjoint time dependent Hamiltonian
\begin{equation}\label{eq:our.H}
	H(\omega t)\;=\; J \cdot N + \varepsilon V(\omega t)\, , 
	\qquad \varepsilon \ll 1\,,
\end{equation}
with $V \in \mathcal{O}^\strong_{\kappa,\rho} $ for some $\kappa, \rho >0$,  $J \in \R^r$ and $\omega \in \R^m$. We consider models for which the vector $(J,\omega) \in \R^{r+m}$ is Diophantine, that is there exist $\gamma>0$, $\tau>r+m-1$ such that $(J, \omega)$ belongs to the set
\begin{equation}\label{eq:siamo.dei.cani}
	 \mathcal{D}_{\gamma, \tau}:= \left \lbrace (J, \omega) \in \R^{m+r}\ \left|\ |J\cdot k + \omega \cdot l | \; > \; \frac{\gamma}{(|k|+|l|)^\tau} \, , \quad \forall  (k,l) \in \mathbb{Z}^{m+r} \setminus \{ (0,0) \} \right. \right \rbrace\,.
\end{equation}
Moreover, we fix $\Omega>0$ and $\mathcal{J}>0$ and we shall assume that $|\omega|\leq \Omega$, $|J| \leq \mathcal{J}$.

We will say that a quantity \emph{depends on the parameters of the system associated to \eqref{eq:our.H}} if it depends only on $d, r, m, \kappa, \rho, \gamma, \tau, \mathcal{J}, \Omega$, $\{\|N^{(\alpha)}\|_{0}\}_{\alpha = 1}^r$, $\{\|N^{(\alpha)}\|_{\kappa}\}_{\alpha = 1}^r$, $\|V\|_{\kappa, \rho}$. {In particular, a quantity depending on the parameters of the system associated to \eqref{eq:our.H} \emph{does not depend} on {$|\Lambda|$} or on $\varepsilon$.

	\begin{theorem}[Quasi-conservation laws for small perturbations]\label{teo:SlowHeating}
Let $H(\omega t)$ be as in \eqref{eq:our.H} and let $\kappa, \rho, \gamma, \Omega >0$ and $\tau > r + m -1$ be such that the operators $\{N^{(\alpha)}\}_{\alpha=1}^r$ satisfy assumptions \textnormal{(N.i)--(N.iv)}, $V \in \mathcal{O}_{\kappa, \rho}^\strong$, and $(J, \omega) \in \mathcal{D}_{\gamma, \tau}$ with $|J| \leq \mathcal{J}$ and $|\omega| \leq \Omega$.
Then there exists $\varepsilon_0>0$ depending on the parameters of the system associated to \eqref{eq:our.H}
such that, if $0< \varepsilon <\varepsilon_0$, the following holds. For any $\alpha= 1, \dots, r$ one has
		\begin{equation}
		\frac{1}{|\Lambda|}\|U^*_{H}(t)  N^{(\alpha)}U_{H}(t)  -N^{(\alpha)}\|_{\mathrm{op}} \leq  K \left( e^{-{\varepsilon^{-2\tb}}} t + \varepsilon^{\frac 1 2} \right) \, \qquad \forall t > 0\, , 
	\end{equation}
	where
	\begin{equation}\label{parametri.vitali}
	\begin{gathered}
	\tb := \frac{1}{4(\tau + r +2)}\,, \quad K:=C(\kappa, \rho) \Vert V \Vert_{\kappa,\rho} \max_{\alpha=1, \dots, r}\{\Vert N^{(\alpha)} \Vert_{\kappa} \}\,,
	\end{gathered}
	\end{equation}
	for some $C(\kappa,\rho)>0$. {As a consequence,
	\begin{equation}\label{vicini.small}
	\frac{1}{|\Lambda|}\|U^*_{H_{}}(t)  N^{(\alpha)}U_{H_{}}(t) -N^{(\alpha)}\|_{\mathrm{op}} \leq 2 K \varepsilon^{\frac 1 2} \quad \forall 0< t < e^{{\varepsilon^{-2\tb}}} \varepsilon^{\frac 12}\,.
	\end{equation}
	}
	\end{theorem}

\begin{theorem}[Evolution of local observables {for small perturbations}]\label{thm:EvLocObs}
Let $H(\omega t)$ be as in \eqref{eq:our.H} and let $\kappa, \rho, \gamma, \Omega >0$ and $\tau > r + m -1$ be such that the operators $\{N^{(\alpha)}\}_{\alpha=1}^r$ satisfy assumptions \textnormal{(N.i)--(N.iv)}, $V \in \mathcal{O}_{\kappa, \rho}^\strong$, and $(J, \omega) \in \mathcal{D}_{\gamma, \tau}$ with $|J| \leq \mathcal{J}$ and $|\omega| \leq \Omega$.
Then there exists $\varepsilon_0>0$ depending on the parameters of the system associated to \eqref{eq:our.H}
such that, if $0<\varepsilon< \varepsilon_0$ and $\kappa_* := \frac{\min\{\kappa, \rho\}}{\sqrt{2}}$, the following holds. There exists a quasi-local, time independent effective Hamiltonian
\begin{equation}\label{vaccataprossimavoltastozitto}
H_{\mathrm{eff}} \;:=\;J \cdot N + Z_{\mathrm{eff}}\,, \qquad \| Z_{\mathrm{eff}}\|_{\kappa_*} \leq \frac{e}{e-1} \varepsilon^{\frac12} \|V\|_{\kappa, \rho}\,,
\end{equation}
such that for any $S \in \mathcal{P}_c(\Lambda)$ and any local operator $O$ acting only within $S$, there exist $C_1(O)>0$ and $C_2(O)>0$, depending on the parameters of the system associated to \eqref{eq:our.H} and on $O$ only,
such that
\begin{itemize}
		\item[(i)] for any $t < \left(e^{{\varepsilon^{-\tb}}} \varepsilon^{\frac 1 2}\right)^{\frac{1}{d+2}}$ 
		\begin{equation}\label{eq:prosciutto.crudo}
			\Vert U_H^*(t) O U_H(t) - e^{-\ii H_{\mathrm{eff}} t} O e^{\ii H_{\mathrm{eff}} t} \Vert_{\mathrm{op}} \leq C_1(O) \varepsilon^{\frac 1 2} \, ,
		\end{equation}
		with $\tb$ defined in \eqref{parametri.vitali};
		\item[(ii)] there exists a collection of recurrence times $\{t_j\}_{j\in \mathbb{N}}$ such that {for any $j \in \N$}
		\begin{equation}\label{eq:melone}
			\Vert U_H^*(t_j) O U_H (t_j)- e^{-\ii H_{\mathrm{eff}} t_j} O e^{\ii H_{\mathrm{eff}} t_j} \Vert_{\mathrm{op}} \leq C_2(O) {j^{d+2}} \varepsilon^{\frac{1-\mathsf{f}}{2}} e^{-{\mathsf{f}}\varepsilon^{-\tb}}\,,
		\end{equation}
	where
	\begin{equation}\label{eq:EquazioneCarina}
		\mathsf{f}:=\frac{1}{\tau(d+2)+1} \, .
	\end{equation}
	More precisely, there exists a constant $a_{\gamma, m, \tau}>0$ such that $\forall j \in \mathbb{N}$
	\begin{equation}\label{baucco}
		t_j\in[jT(\varepsilon),(j+1)T(\varepsilon)] \, , \qquad T(\varepsilon):={a_{\gamma, m, \tau}e^{{\mathsf{f}}\tau \varepsilon^{-\tb}}\varepsilon^{\frac{\mathsf{f}}{2}}}\,.
	\end{equation}
\end{itemize}
\end{theorem}

The proofs of Theorems \ref{teo:SlowHeating} and \ref{thm:EvLocObs} are in Section \ref{sec:physical.cons}. We make the following comments:
\begin{itemize}
	\item Theorem \ref{teo:SlowHeating} and Item (i) of Theorem \ref{thm:EvLocObs} generalize the result in \cite{DeRoeck-Verreet}, which holds for time periodic small size perturbations, to the time quasi-periodic case. Moreover, in \cite{DeRoeck-Verreet} De Roeck and Verreet show the existence of a time dependent effective Hamiltonian $H_\eff(t)$ such that \eqref{eq:prosciutto.crudo} holds, whereas the Hamiltonian $H_\eff$ in Theorem \ref{thm:EvLocObs} is time independent.
	\item In \cite{Abanin2017} Abanin et al.  study the case $H(t) = H_0 + P(\nu t)$, with $P$ time periodic and $\nu \gg 1$ a fast frequency. They prove the existence of  an effective Hamiltonian $H_\eff$ and of \emph{stroboscopic times} $t_j := \frac{2\pi}{j}$, $j \in \N$, such that $U_H^*(t) O U_H(t)$ and $e^{-\ii H_\eff t} O e^{\ii H_\eff t}$ become exponentially close at times $t=t_j$. In Item (ii) of Theorem \ref{thm:EvLocObs} we show that such a phenomenon is true also for quasi-periodic systems, replacing stroboscopic times with recurrence times $\{t_j\}_j$ \eqref{baucco}.
	\item Theorem \ref{teo:SlowHeating} is based on the non convergent normal form result of Proposition \ref{teo:main}. Theorem \ref{thm:EvLocObs} requires the combination of the normal form results of Proposition \ref{teo:main} with Lieb-Robinson bounds \cite{Lieb1972} and results from classical dynamical systems \cite{Berti2003}. The same holds for the proof of Theorems \ref{thm:SlowHeatingFF} and \ref{thm:EvLocObs.ff} below, concerning the fast forced case.
\end{itemize}

\subsection{Fast-Forcing Perturbations}
The second class of perturbations we consider are non-small in size but high frequency and time quasi-periodic. We consider the self-adjoint time dependent Hamiltonian
\begin{equation}\label{eq:our.H.lambda}
	H(\lambda \omega t)\;=\; J \cdot N + V(\lambda \omega t)\, , 
	\qquad \lambda \gg 1\,,
\end{equation}
with $ V \in \mathcal{O}^\strong_{\kappa,\rho}$ for some $\kappa,\rho>0$. For these models, we {require} $J \in \R^r$ and $\omega \in \R^m$ to be Diophantine separately:
there exist $\gamma_\omega>0$ and $\gamma_J>0$, $\tau_\omega> m-1$ and $\tau_J > r-1$ such that $J \in \mathcal{D}_{\gamma_J, \tau_J}$ and $\omega \in \mathcal{D}_{\gamma_\omega, \tau_\omega}$, with
\begin{subequations}\label{eq:siamo.dei.cani.ff}
	\begin{equation}\label{eq:siamo.dei.cani.ff.1}
	\mathcal{D}_{\gamma_J, \tau_J} := \left \lbrace J \in \mathbb{R}^r \ \Big|\ |J\cdot k| \geq \frac{\gamma_J}{|k|^{\tau_J}} \, , \qquad \forall k \in \mathbb{Z}^r \setminus \{0\} \right\rbrace\, ,
	\end{equation}
	\begin{equation}\label{eq:siamo.dei.cani.ff.2}
	\mathcal{D}_{\gamma_\omega, \tau_\omega} := \left \lbrace \omega \in \mathbb{R}^m \ \Big|\ |\omega\cdot l| \geq \frac{\gamma_\omega}{|l|^{\tau_\omega}} \, , \qquad \forall l \in \mathbb{Z}^m \setminus \{0\} \right\rbrace\, .
	\end{equation}
{Furthermore, we shall fix $\mathcal{J}>0$ and require that $|J| \leq \mathcal{J}$.}
\end{subequations}

	We will say that a quantity depends on the parameters of the system associated to \eqref{eq:our.H.lambda} if it depends only on $d, r, m, \kappa, \rho,  \tau_J, \tau_\omega, \gamma_J, \gamma_\omega, {\mathcal{J}}, \{\| N^{(\alpha)} \|_{0}\}_\alpha, \{\| N^{(\alpha)} \|_{\kappa}\}_\alpha, \|V\|_{\kappa, \rho}$. In particular, a quantity depending on the parameters of the system associated to \eqref{eq:our.H.lambda} \emph{does not depend} on $|\Lambda|$, or on $\lambda$.
	\begin{theorem}[Quasi-conservation laws {for fast forcing perturbations}]\label{thm:SlowHeatingFF}
	Let $H(\lambda \omega t)$ be as in \eqref{eq:our.H.lambda} and let $\kappa, \rho, \gamma_J, \gamma_\omega, \mathcal{J}>0$, $\tau_J > r-1$  and $\tau_\omega> m-1$ be such that the operators $\{N^{(\alpha)}\}_{\alpha=1}^r$ satisfy assumptions \textnormal{(N.i)--(N.iv)}, $V \in \mathcal{O}^\strong_{\kappa, \rho}$, $J \in \mathcal{D}_{\gamma_J, \tau_J}$, $\omega \in \mathcal{D}_{\gamma_\omega, \tau_\omega}$ and $|J| \leq \mathcal{J}$.
	Then, there exists $\lambda_0>0$, depending on the parameters of the system associated to \eqref{eq:our.H.lambda} only,
	such that for any $\lambda > \lambda_0$ and any $\alpha = 1, \dots, r$ the following holds.
		\begin{equation}
		\frac{1}{|\Lambda|}\|U^*_{H_{}}(t)  N^{(\alpha)}U_{H_{}}(t) -N^{(\alpha)}\|_{\mathrm{op}} \leq  K \left( e^{- \lambda^{2\beta}} t + \lambda^{-\frac{1}{8 \tau_\omega}} \right) \, \qquad \forall t > 0\,,
		\end{equation}
		where
	\begin{equation}\label{parametri.vitali.ff}
		\beta := \frac{1}{16 \tau_\omega( \tau_J +r+2)}\,, \qquad
		K:=C(\kappa,\rho) \Vert V \Vert_{\kappa,\rho} \max_{\alpha = 1, \dots, r}\{ \Vert N^{(\alpha)} \Vert_{\kappa}\}\, 
	\end{equation}
	for some $C(\kappa,\rho)>0$.  {As a consequence,
		\begin{equation}\label{vicini.small.ff}
		\frac{1}{|\Lambda|}\|U^*_{H_{}}(t)  N^{(\alpha)}U_{H_{}}(t) -N^{(\alpha)}\|_{\mathrm{op}} \leq 2 K \lambda^{-\frac{1}{8 \tau_\omega}} \quad \forall 0< t < e^{ \lambda^{2\beta}} \lambda^{-\frac{1}{8\tau_\omega} }\,.
		\end{equation}
	}
	\end{theorem}

\begin{theorem}[Evolution of local observables {for fast forcing perturbations}]\label{thm:EvLocObs.ff}
	Let $H(\lambda \omega t)$ be as in \eqref{eq:our.H.lambda} and let $\kappa, \rho, \gamma_J, \gamma_\omega, \mathcal{J}>0$, $\tau_J > r-1$  and $\tau_\omega> m-1$ be such that the operators $\{N^{(\alpha)}\}_{\alpha=1}^r$ satisfy assumptions \textnormal{(N.i)--(N.iv)}, $V \in \mathcal{O}^\strong_{\kappa, \rho}$, $J \in \mathcal{D}_{\gamma_J, \tau_J}$, $\omega \in \mathcal{D}_{\gamma_\omega, \tau_\omega}$, and $|J| \leq \mathcal{J}$.
	Then, there exists $\lambda_0>0$, depending only on the parameters of the system associated to \eqref{eq:our.H.lambda}
	such that, if $\lambda \geq \lambda_0$ and $\kappa_* := \frac{\min\{\kappa, \rho\}}{\sqrt{2}}$, the following holds. There exists a quasi-local time independent effective Hamiltonian
	\begin{equation}\label{vaccataprossimavoltastozitto.ff}
	H_{\mathrm{eff}} \;:=\;J \cdot N + Z_{\mathrm{eff}}\,, \qquad \| Z_{\mathrm{eff}}\|_{\kappa_*} \leq \frac{e}{e-1} \lambda^{-\frac{1}{8 \tau_\omega}} \|V\|_{\kappa, \rho}\,,
	\end{equation}
	such that for any
	{$S \in \mathcal{P}_c(\Lambda)$ and any local operator $O$ acting only within $S$, there exist} $C_1(O)>0$ and $C_2(O)>0$, depending only on $O$ and on the parameters of the systems associated to \eqref{eq:our.H.lambda},
	such that
	\begin{itemize}
	\item[(i)] for any $t < \left(e^{ \lambda^\beta} \lambda^{-\frac{1}{8 \tau_\omega}}\right)^{\frac{1}{d+2}}$ 
	\begin{equation}
	\Vert U^*(t) O U (t)- e^{-\ii H_{\mathrm{eff}} t} O e^{\ii H_{\mathrm{eff}} t} \Vert_{\mathrm{op}} \leq C_1(O) \lambda^{-\frac{1}{8\tau_\omega}} \,,
	\end{equation}
	where $\beta$ is defined as in \eqref{parametri.vitali.ff};
	\item[(ii)] there exists a collection of recurrence times $\{t_j\}_{j\in \mathbb{N}}$ such that {for any $j \in \N$}
	\begin{equation}
	\Vert U^*({t_j}) O U ({t_j})- e^{-\ii H_{\mathrm{eff}} {t_j}} O e^{\ii H_{\mathrm{eff}} {t_j}} \Vert_{\mathrm{op}} \leq C_2(O) {j^{d+2}} e^{- \mathsf{g} \lambda^\beta} \lambda^{\frac{(8\tau_\omega-1)(1-\mathsf{g})}{8 \tau_\omega}}\,,
	\end{equation}
	where 
	\begin{equation}\label{abbandono.di.dal.maso}
	\mathsf{g} := \frac{1}{\tau_\omega(d+2) + 1}\,.
	\end{equation}
	More precisely, there exists a constant $a_{\gamma, m , \tau_\omega} >0$ such that $\forall j \in \N$
	\begin{equation}\label{tempi.tj.ff}
	t_j \in [j T(\lambda), (j+1) T(\lambda)]\,, \quad T(\lambda) := a_{\gamma_\omega, m , \tau_\omega} e^{ \tau_\omega \mathsf{g} \lambda^{\beta}} \lambda^{\frac{(8 \tau_\omega -1) \mathsf{g}}{8}} \,.
	\end{equation}
	\end{itemize}
\end{theorem}

 We point out the following:
 \begin{itemize}
 \item In \cite{Else2020} Else, Ho and Dumitrescu treat the case of a quasi-periodically fast forced Hamiltonian $H(t)$ as in \eqref{H0} without any structural assumption on $H_0$, and they exhibit \emph{one} quasi-conserved quantity for stretched exponentially long times. Theorem \ref{thm:SlowHeatingFF}, under the stronger assumption that the unperturbed Hamiltonian has the form $H_0 = J \cdot N$, with $N^{(1)}, \dots, N^{(r)}$ number operators, exhibits \emph{$r$ distinct} quasi-conserved quantities over analogous time scales. {We point out that this result would not be reachable with the normal form procedure of \cite{Else2020}, through which  energy exchanges among the number operators $N^{(\alpha)}$ cannot be excluded.}
 	\item Item (i) of Theorem \ref{thm:EvLocObs.ff} is not new, and here we recall it only to draw a complete picture also for the case of fast-forcing perturbation. Its proof is contained in \cite{Else2020}, without assuming that $H_0$ has the structure in \eqref{JdotN} (see Remark \ref{rem:GranchiettiEBussolai} below).
 	\item Item (ii) of Theorem \ref{thm:EvLocObs.ff} is instead original, and it shows for the first time the existence of recurrence times for local observables, previously proved in \cite{Abanin2017} only for time periodic forcing, to the more delicate case of time quasi-periodic perturbations. As for item (i), the assumption on the special form \eqref{JdotN} of $H_0$ is not required.
 \end{itemize}
\subsection{Applications: Fermi-Hubbard Model and Quantum Ising Chain} \label{subsec:GranchiettiCarini}

\subsubsection{Generalized Fermi-Hubbard Model}
We consider a model of $\frac 1 2$-spin fermions on a lattice composed of sites whose distance oscillates quasi-periodically in time. The Hamiltonian of the model is
\begin{equation}
	H\;=\; \varepsilon \sum_{\substack{x \sim y  \\ \sigma \in\{\uparrow,\downarrow\}}} K_{x,y}(\omega t) (a^+_{x,\sigma} a_{y, \sigma} + a^+_{y,\sigma} a_{x,\sigma}) + \sum_{\alpha=1}^r  J_\alpha \sum_{x \in \Lambda}\sum_{\substack{y \in \Lambda \\ |x-y|_1 = \alpha}} n_{x,\uparrow} n_{y,\downarrow}\,,
\end{equation}
where $\forall x \in \Lambda :=\mathbb{Z}^d \cap [-L,L]^d$ and $\sigma \in \{\uparrow, \downarrow\}$, $a_{x,\sigma}$ and $a^+_{x,\sigma}$ are fermionic annihilation and creation operators, $n_{x,\sigma}:=a^+_{x,\sigma} a_{x,\sigma}$, $\forall x, y \in \Lambda$ $K_{x, y} = K_{y, x}: \mathbb{T}^m \rightarrow \R$ are analytic functions {uniformly bounded in $x,y$}, $(J, \omega) \in \R^{r+m}$ satisfies Diophantine condition \eqref{eq:siamo.dei.cani}, $\varepsilon \ll 1$, and $x \sim y$ denotes the sum over the couples $x, y \in \Lambda$ such that $|x-y|_1 = 1$. 
As number operators, we consider
\begin{equation}
	N^{(\alpha)} \;:=\; \sum_{x \in \Lambda} \sum_{\substack{y \in \Lambda \\ |x-y|_1=\alpha}} n_{x,\uparrow} n_{y,\downarrow}\,,
\end{equation}
which for any $\alpha$ count the number of couples of particles with opposite spin and occupying sites at distance $\alpha$. For example, for $\alpha = 1$, $N^{(1)}$ counts the number of particles with opposite spin and occupying nearest-neighbor sites. Each $N^{(\alpha)}$ admits the decomposition in local operators
\begin{equation}
	N^{(\alpha)} = \sum_{x \in \Lambda} N^{(\alpha)}_{S_x^{(\alpha)}}\,, \quad S_x^{(\alpha)} := \{ y \in \Lambda \ | \ |x-y|_1 \leq \alpha\}\,, \quad N^{(\alpha)}_{S_x^{(\alpha)}} := \sum_{y \in \Lambda \atop |x-y|_1 = \alpha} n_{x, \uparrow} n_{y, \downarrow}\,,
\end{equation}
and in order to ensure strong locality, it is convenient to decompose the perturbation as
\begin{equation}
\begin{gathered}
\varepsilon V(\omega t) = \varepsilon \sum_{x \in \Lambda} V_{S'_x}(\omega t)\,, \quad S'_{x} := \{ y \in \Lambda \ |\ |x-y|_1 \leq 2r +1\}\,, \\
 V_{S'_x}(\omega t) = \sum_{y \in \Lambda \atop  y \sim x}  \sum_{\sigma \in \{\uparrow, \downarrow \}} K_{x, y}(\omega t) (a^+_{x,\sigma} a_{y, \sigma} + a^+_{y, \sigma} a_{x, \sigma})\,.
\end{gathered}
\end{equation} {We point out that, since $\Vert n_{x,\sigma} \Vert_{\mathrm{op}}=1$, each $N^{(\alpha)}$ satisfies \eqref{norma.alvise} with $\Vert N^{(\alpha)}_{S^{(\alpha)}} \Vert_{\mathrm{op}} \leq C_{\alpha,d}$ and $\bar{s}=\alpha^d$ for certain constants $C_{\alpha,d}>0$, and therefore $\Vert N^{(\alpha)} \Vert_{\kappa} < C_{\alpha,\kappa,d}$. With analogous argument, it is possible to show that, for a certain explicit constant $C_{d,\alpha}>0$,  $\Vert N^{(\alpha)} \Vert_{\mathrm{op}}=C_{d,\alpha} |\Lambda|$. Analogously, using that $V$ has only nearest-neighbour interactions, the analyticity in time of $K_{x,y}(\omega t)$ and the boundedness of the fermionic creation and annihilation operators, we have $V \in \mathcal{O}_{\kappa,\rho}$ for any $\kappa,\rho > 0$.}

If $\varepsilon=0$, then all $N^{(\alpha)}$'s are conserved. For $\varepsilon \neq 0$ and $K_{x,y}$ non trivial, one has $[N^{(\alpha)}, V(\omega t)] \neq 0$. Theorem \ref{teo:SlowHeating} ensures that each $N^{(\alpha)}$
is quasi-conserved for exponentially long times in $\varepsilon$:
\begin{equation}
 |\Lambda|^{-1} \| U^*_{H}(t) N^{(\alpha)} U_H(t) - N^{(\alpha)} \|_{\mathrm{op}} \leq C \varepsilon^{\frac 12} \quad \forall 0< t < e^{ \varepsilon^{-2\tb}} \varepsilon^{\frac 1 2}\,.
\end{equation}
The dynamical consequences of these quasi-conservation laws are pictorically represented in Figures \ref{gialli.e.verdi} and \ref{verdi.e.gialli}. In particular, for an initial datum as in the first line of Figures \ref{gialli.e.verdi} and \ref{verdi.e.gialli}, Figure \ref{gialli.e.verdi} shows that the highlighted package can not disperse, while Figure \ref{verdi.e.gialli} shows that it can not get too close to the group of particles on the right, as it happens for hard spheres.


Furthermore, Theorem \ref{thm:EvLocObs} ensures that there exists an effective  time independent and local Hamiltonian $H_{\mathrm{eff}}$ which approximates the dynamics of local observables for exponentially long times, according to \eqref{eq:prosciutto.crudo} and \eqref{eq:melone}. $H_\eff$ is also explicitly computable: at the first order in $\varepsilon$, using \eqref{H.eff}, \eqref{new.objects} and \eqref{cacio.e.pepe}, one obtains
\begin{eqnarray}
	\nonumber
	H_{\mathrm{eff}}\!\!\!\!&=&\!\!\!\!J\cdot N + Z^{(1)} +O(\varepsilon^2)\,,
	\\
	Z^{(1)}  \!\!\!\!&=&\!\!\!\! \varepsilon \sum_{x \sim y, \sigma} {(\widehat{ K_{x,y}})_0} (a^+_{y,\sigma} a_{x,\sigma}+a^+_{x,\sigma} a_{y,\sigma}) \prod_{\alpha=1}^rP_{\mathcal{K}_{x,y,\alpha,\sigma}}(0) \nonumber \\
	&\,&+\varepsilon \sum_{x \sim y, \sigma} \sum_{\substack{l \in \mathbb{Z}^m \\ k \in \mathbb{Z}^r}}\frac{J\cdot k}{\omega \cdot l+J \cdot k} \left( ( \widehat{K_{x,y}} )_l  a_{x,\sigma}^+a_{y,\sigma}+ ( \widehat{K_{x,y}} )_l^* a_{y,\sigma}^+a_{x,\sigma}  \right) \prod_{\alpha=1}^rP_{\mathcal{K}_{x,y,\alpha,\sigma}}(k_\alpha) \nonumber
\end{eqnarray}
where {$(\widehat{K_{x,y}})_l:= (2\pi)^{-m}\int_{\mathbb{T}^m} K_{x,y}(\varphi) e^{-\ii l\cdot \varphi} \ud \varphi$,} and $P_{\mathcal{K}_{x,y,\alpha,\sigma}}(k_\alpha)$ denotes the projector onto the space $\mathcal{K}_{x,y,\alpha,\sigma}(k_\alpha)$ defined as
\begin{equation}
\mathcal{K}_{x,y,\alpha,\sigma}(k_\alpha):=\ker \left(\sum_{\eta \, : \, |\eta-x|_1=\alpha} n_{\eta,-\sigma}-\sum_{\eta \, : \, |\eta-y|_1=\alpha} n_{\eta,-\sigma}-k_\alpha \right) \, .
\end{equation}

\begin{figure}[h]
	\begin{center}
		\begin{tikzpicture}[thick,scale=1.1]
		
		\draw[->] (0,0.8) -- (0,-1.8);
		\node at (0.3,-1.8) {$t$};
		
		\draw[dotted] (1,0) -- (1.5,0);
		\draw[fill] (2,0) circle[radius=0.09]; 
		\draw[fill] (3,0) circle[radius=0.09];
		\draw[fill] (4,0) circle[radius=0.09];
		\draw[fill] (5,0) circle[radius=0.09];
		\draw[fill] (6,0) circle[radius=0.09];
		\draw[fill] (7,0) circle[radius=0.09];
		\draw[fill] (8,0) circle[radius=0.09];
		\draw[fill] (9,0) circle[radius=0.09];
		\draw[fill] (10,0) circle[radius=0.09];
		\draw[dotted] (10.5,0) -- (11,0);
		
		\node at (4,0) [above=0.3cm, circle, draw, ball color = verdino] {$\uparrow$};	
		\node at (5,0) [above=0.3cm, circle, draw, ball color = giallino] {$\downarrow$};	
		\node at (6,0) [above=0.3cm, circle, draw, ball color = verdino] {$\uparrow$};
		\node at (9,0) [above=0.3cm, circle, draw, ball color = verdino] {$\uparrow$};	
			\node at (10,0) [above=0.3cm, circle, draw, ball color = giallino] {$\uparrow$};

			\draw[color=red] (3.5,-0.2) -- (3.5,1.1);	
			\draw[color=red] (3.5,1.1) -- (6.5, 1.1);
			\draw[color=red] (6.5,1.1) -- (6.5,-0.2);
			\draw[color=red] (3.5,-0.2) -- (6.5,-0.2);

			\draw[color=red] (3.5,-1.7) -- (3.5,-0.4);	
			\draw[color=red] (3.5,-0.4) -- (5.5, -0.4);
			\draw[color=red] (3.5,-1.7) -- (5.5,-1.7);
			\draw[color=red] (6.5,-1.7) -- (7.5,-1.7);
			\draw[color=red] (7.5,-0.4) -- (7.5,-1.7);
			\draw[color=red] (7.5,-0.4) -- (6.5,-0.4);
			
			\draw[dashed, color = red, decoration={zigzag,segment length=10mm}] decorate{(5.5,-1.7)--(5.5,-0.4)};
			\draw[dashed, color = red, decoration={zigzag,segment length=10mm}] decorate{(6.5,-1.7)--(6.5,-0.4)};
		
		\draw[dotted] (1,-1.5) -- (1.5,-1.5);
		\draw[fill] (2,-1.5) circle[radius=0.09]; 
		\draw[fill] (3,-1.5) circle[radius=0.09];
		\draw[fill] (4,-1.5) circle[radius=0.09];
		\draw[fill] (5,-1.5) circle[radius=0.09];
		\draw[fill] (6,-1.5) circle[radius=0.09];
		\draw[fill] (7,-1.5) circle[radius=0.09];
		\draw[fill] (8,-1.5) circle[radius=0.09];
		\draw[fill] (9,-1.5) circle[radius=0.09];
		\draw[fill] (10,-1.5) circle[radius=0.09];
		\draw[dotted] (10.5,-1.5) -- (11,-1.5);
		
		\node at (4,-1.5) [above=0.3cm, circle, draw, ball color = verdino] {$\uparrow$};	
		\node at (5,-1.5) [above=0.3cm, circle, draw, ball color = giallino] {$\downarrow$};	
		\node at (7,-1.5) [above=0.3cm, circle, draw, ball color = verdino] {$\uparrow$};
		\node at (9,-1.5) [above=0.3cm, circle, draw, ball color = verdino] {$\uparrow$};	
			\node at (10,-1.5) [above=0.3cm, circle, draw, ball color = giallino] {$\uparrow$};	
			\node at(12.8,-1) {\textbf{not allowed}};
	\node at(12.8,0.5) {initial datum};		
		
%
%
		\end{tikzpicture}
	\end{center}
	\caption{
	Pictorial representation of the Fermi-Hubbard dynamics. Time evolution is represented vertically from top to bottom; at each level, solid dots represent the lattice sites. A particle is represented by a colored sphere above the site where it lies. The arrow inside the particle represents the spin. The highlighted package of particles cannot disperse, since this is not allowed by conservation laws. Indeed, given the initial datum on the first line, one has $N^{(1)}=6$ and $N^{(2)}=0$ while, after the first time-lapse one has $N^{(1)}=4$ and $N^{(2)}=2$. 
	}
	\label{gialli.e.verdi}
\end{figure}
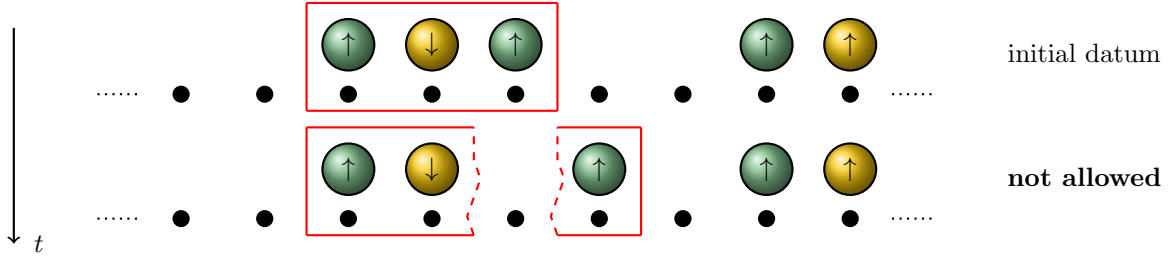

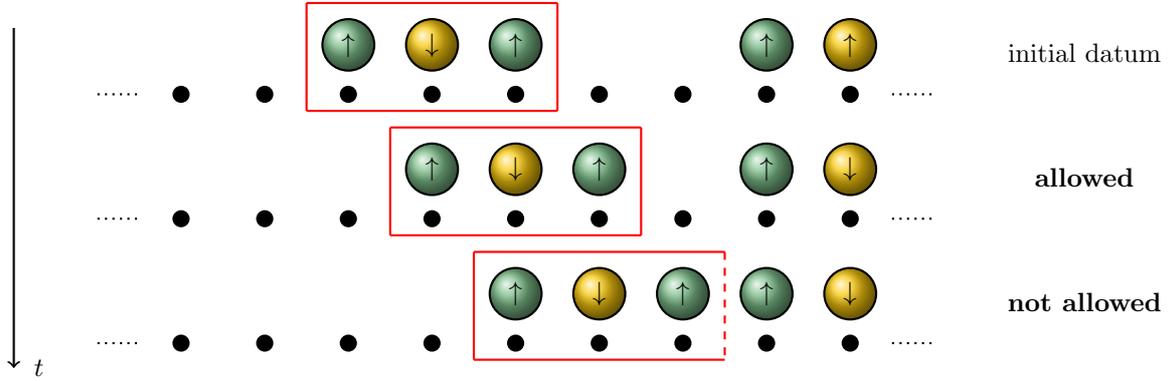
\begin{figure}[h]
	\begin{center}
		\begin{tikzpicture}[thick,scale=1.1]
		
		\draw[->] (0,0.8) -- (0,-3.3);
		\node at (0.3,-3.3) {$t$};
		
		\draw[dotted] (1,0) -- (1.5,0);
		\draw[fill] (2,0) circle[radius=0.09]; 
		\draw[fill] (3,0) circle[radius=0.09];
		\draw[fill] (4,0) circle[radius=0.09];
		\draw[fill] (5,0) circle[radius=0.09];
		\draw[fill] (6,0) circle[radius=0.09];
		\draw[fill] (7,0) circle[radius=0.09];
		\draw[fill] (8,0) circle[radius=0.09];
		\draw[fill] (9,0) circle[radius=0.09];
		\draw[fill] (10,0) circle[radius=0.09];
		\draw[dotted] (10.5,0) -- (11,0);
	
			\draw[color=red] (3.5,-0.2) -- (3.5,1.1);	
			\draw[color=red] (3.5,1.1) -- (6.5, 1.1);
			\draw[color=red] (6.5,1.1) -- (6.5,-0.2);
			\draw[color=red] (3.5,-0.2) -- (6.5,-0.2);

			\draw[color=red] (4.5,-1.7) -- (4.5,-0.4);	
			\draw[color=red] (4.5,-0.4) -- (7.5, -0.4);
			\draw[color=red] (4.5,-1.7) -- (7.5,-1.7);
			\draw[color=red] (7.5,-0.4) -- (7.5,-1.7);

			\draw[color=red] (5.5,-3.2) -- (5.5,-1.9);	
			\draw[color=red] (5.5,-1.9) -- (8.5, -1.9);
			\draw[color=red] (5.5,-3.2) -- (8.5,-3.2);
			\draw[color=red,dashed] (8.5,-1.9) -- (8.5,-3.2);

		\node at (4,0) [above=0.3cm, circle, draw, ball color = verdino] {$\uparrow$};	
		\node at (5,0) [above=0.3cm, circle, draw, ball color = giallino] {$\downarrow$};	
		\node at (6,0) [above=0.3cm, circle, draw, ball color = verdino] {$\uparrow$};
		\node at (9,0) [above=0.3cm, circle, draw, ball color = verdino] {$\uparrow$};	
			\node at (10,0) [above=0.3cm, circle, draw, ball color = giallino] {$\uparrow$};	
		
		\draw[dotted] (1,-1.5) -- (1.5,-1.5);
		\draw[fill] (2,-1.5) circle[radius=0.09]; 
		\draw[fill] (3,-1.5) circle[radius=0.09];
		\draw[fill] (4,-1.5) circle[radius=0.09];
		\draw[fill] (5,-1.5) circle[radius=0.09];
		\draw[fill] (6,-1.5) circle[radius=0.09];
		\draw[fill] (7,-1.5) circle[radius=0.09];
		\draw[fill] (8,-1.5) circle[radius=0.09];
		\draw[fill] (9,-1.5) circle[radius=0.09];
		\draw[fill] (10,-1.5) circle[radius=0.09];
		\draw[dotted] (10.5,-1.5) -- (11,-1.5);
		
		\node at (5,-1.5) [above=0.3cm, circle, draw, ball color = verdino] {$\uparrow$};	
		\node at (6,-1.5) [above=0.3cm, circle, draw, ball color = giallino] {$\downarrow$};	
		\node at (7,-1.5) [above=0.3cm, circle, draw, ball color = verdino] {$\uparrow$};
		\node at (9,-1.5) [above=0.3cm, circle, draw, ball color = verdino] {$\uparrow$};	
		\node at (10,-1.5) [above=0.3cm, circle, draw, ball color = giallino] {$\downarrow$};	
		
		\draw[dotted] (1,-3) -- (1.5,-3);
		\draw[fill] (2,-3) circle[radius=0.09]; 
		\draw[fill] (3,-3) circle[radius=0.09];
		\draw[fill] (4,-3) circle[radius=0.09];
		\draw[fill] (5,-3) circle[radius=0.09];
		\draw[fill] (6,-3) circle[radius=0.09];
		\draw[fill] (7,-3) circle[radius=0.09];
		\draw[fill] (8,-3) circle[radius=0.09];
		\draw[fill] (9,-3) circle[radius=0.09];
		\draw[fill] (10,-3) circle[radius=0.09];
		\draw[dotted] (10.5,-3) -- (11,-3);
		
		\node at (6,-3) [above=0.3cm, circle, draw, ball color = verdino] {$\uparrow$};	
		\node at (7,-3) [above=0.3cm, circle, draw, ball color = giallino] {$\downarrow$};	
		\node at (8,-3) [above=0.3cm, circle, draw, ball color = verdino] {$\uparrow$};
		\node at (9,-3) [above=0.3cm, circle, draw, ball color = verdino] {$\uparrow$};	
		\node at (10,-3) [above=0.3cm, circle, draw, ball color = giallino] {$\downarrow$};	
%
%
%

			\node at(12.8,-1) {\textbf{allowed}};
	\node at(12.8,0.5) {initial datum};	
			\node at(12.8,-2.5) {\textbf{not allowed}};
		\end{tikzpicture}
	\end{center}
	\caption{Pictorial representation of the Fermi-Hubbard dynamics. Time evolution is represented as in Figure \ref{gialli.e.verdi}. The picture shows that, according to conservation laws, the highlighted package of particles can translate rigidly but cannot get too close to the pair of particles on the right. 
} 
	\label{verdi.e.gialli}
\end{figure}

\subsubsection{Quantum Ising Model with Magnetic Field}
Let us consider the quantum Ising model in one spatial dimension on a lattice $\Lambda=[-3L,3L-1] \cap \mathbb{Z}$ with periodic boundary conditions and in presence of an external magnetic field. We consider the case where the external magnetic field has two components. The one along the (3)-axis is space periodic of period 3 sites. The one along the (1)-axis is time quasi-periodic with Diophantine frequency $\omega$ satisfying \eqref{eq:siamo.dei.cani.ff.2}. To write the Hamiltonian of the model, it is convenient to introduce the sublattices $\Lambda_1$, $\Lambda_2$ and $\Lambda_3$ defined as
\begin{subequations}
\begin{eqnarray}
	\Lambda_1&=&\{-3L,-3L+3, \dots, 0,3,\dots\} \, , \\
	\Lambda_2&=&\{-3L+1,-3L+4,\dots,1,4,\dots\} \, ,\\
	\Lambda_3&=&\{-3L+2,-3L+5,\dots,2,5,\dots\} \, .
\end{eqnarray}
\end{subequations}
%
%
%

The Hamiltonian of the model is
\begin{equation}\label{eq:HamiltonianIsole}
	H\;=\; -J \sum_{x \in \Lambda} \sigma_{x}^{(3)} \sigma_{x+e_j}^{(3)}-\sum_{\alpha=1}^{3} h_\alpha \sum_{x \in \Lambda_\alpha} \sigma_x^{(3)}-B(\lambda \omega t) \sum_{x \in \Lambda} \sigma_x^{(1)}
\end{equation}
where $h_1,h_2,h_3 \in \mathbb{R}$ are the values of the magnetic field on the sublattices $\Lambda_1,\Lambda_2,\Lambda_3$ respectively, $\sigma^{(a)}$ is the $a$-th Pauli matrix, $(h_1, h_2,h_3, J) \in \R^{4}$  satisfies the Diophantine condition \eqref{eq:siamo.dei.cani.ff.1}, $\lambda \gg 1$ and $B: \mathbb{T}^m \rightarrow \R$ is an analytic function.
We consider as number operators:
	$$
	N^{(\alpha)}:=\sum_{x \in \Lambda_{\alpha}} \sigma_x^{(3)} \, {, \quad  \alpha=1,2,3\, ,} \qquad
	N^{(4)}:=\sum_{x \in \Lambda} \sigma_{x}^{(3)} \sigma_{x+1}^{(3)} \, .
	$$
The operators $N^{(\alpha)}$ are local, since they admit the decomposition
\begin{equation}
\begin{gathered}
 N^{(\alpha)} = \sum_{x \in \Lambda_\alpha} N^{(\alpha)}_{S_x}\,, \quad S_{x} := \{x\}\,, \quad N^{(\alpha)}_{S_x} = \sigma^{(3)}_x\,, \quad \forall \alpha = 1, 2, 3\,,\\
 N^{(4)} = \sum_{x \in \Lambda} N^{(4)}_{S'_x}\,, \quad S'_x := \{ x, x + 1 \}\,, \quad N^{(4)}_{S'_x} := \sigma^{(3)}_{x} \sigma^{(3)}_{x+1}\,.
\end{gathered}
\end{equation}
For the perturbation, we consider the decomposition
\begin{equation}
\begin{gathered}
B(\lambda \omega t) \sum_{x \in \Lambda} \sigma^{(1)}_x = \sum_{x \in \Lambda} V_{S^{''}_x}(\lambda \omega t)\,, \\
S^{''}_x : = \{ y \in \Lambda \ |\ |x-y|_1 \leq 2 \}\,, \quad V_{S^{''}_x}(\lambda \omega t) = B(\lambda \omega t) \sigma_{x}^{(1)}\,. 
\end{gathered}
\end{equation}
{We remark that, using the notation of \eqref{norma.alvise}, for $\alpha=1,2,3$, $\Vert N^{(\alpha)}_{S_x} \Vert_{\mathrm{op}} \leq 1$ and $\bar{s}=1$. Therefore, $\Vert N^{(\alpha)} \Vert_{\kappa} \leq e^{\kappa}$ and $N^{(\alpha)} \in \mathcal{O}_{\kappa}$ for any $\kappa >0$. Also, $\Vert N^{(\alpha)} \Vert_{\mathrm{op}}=|\Lambda|$. For $\alpha=4$, $\Vert N_{S_x'}^{(4)} \Vert_{\mathrm{op}}=1$, $\bar{s}=2$ and therefore $\Vert N^{(4)} \Vert_{\kappa} = 2 e^{2 \kappa}$. On the other hand, $\Vert N^{(4)} \Vert_{\mathrm{op}} = |\Lambda|$.
}

Theorem \ref{thm:SlowHeatingFF} guarantees that each $N^{(\alpha)}$ is quasi-conserved for times $t \leq \lambda^{-\frac{1}{8\tau_\omega}} e^{\lambda^{2\beta}}$.
Note that, since the Hamiltonian \eqref{eq:HamiltonianIsole} has a term with $\sigma_x^{(1)}$, then the total magnetization along the $(3)$-direction is not conserved. 
Each $N^{(\alpha)}$ for $\alpha=1,2,3$ represents the total magnetization along the direction $(3)$ restricted on the sites $x \in \Lambda_\alpha$. In the case $\alpha = 4$, $N^{(4)}$ represents the number of domain walls, namely the number of pairs of consecutive particles with opposite spin along the $(3)$ component. As a first consequence of Theorem \ref{thm:SlowHeatingFF}, both \emph{the total magnetization of each $\Lambda_\alpha$} and the \emph{number of domain walls} are quasi-conserved up to exponentially long times in $\lambda$. 

\begin{figure}[h]
		\begin{center}
			\begin{tikzpicture}[thick,scale=0.8]
			
			\draw[dotted] (1,0) -- (1.5,0);
			\draw[fill] (2,0) circle[radius=0.09]; 
			\draw[fill] (3,0) circle[radius=0.09];
			\draw[fill] (4,0) circle[radius=0.09];
			\draw[fill] (5,0) circle[radius=0.09];
			\draw[fill] (6,0) circle[radius=0.09];
			\draw[fill] (7,0) circle[radius=0.09];
			\draw[fill] (8,0) circle[radius=0.09];
			\draw[fill] (9,0) circle[radius=0.09];
			\draw[fill] (10,0) circle[radius=0.09];
			\draw[dotted] (11,0) -- (11.5,0);
			
			\node at (2,0) [above=0.3cm] {$\downarrow$};
			\node at (3,0) [above=0.3cm] {$\downarrow$};	
			\node at (4,0) [above=0.3cm] {$\downarrow$};	
			\node at (5,0) [above=0.3cm] {$\uparrow$};	
			\node at (6,0) [above=0.3cm] {$\uparrow$};
			\node at (7,0) [above=0.3cm] {$\downarrow$};	
			\node at (8,0) [above=0.3cm] {$\downarrow$};
			\node at (9,0) [above=0.3cm] {$\downarrow$};
			\node at (10,0) [above=0.3cm] {$\downarrow$};
			
			\draw[dashed, color = blue] (4.5,-5) -- (4.5,1.2);
			\draw[dashed, color = blue] (7.5,-5) -- (7.5,1.2);
			\draw[dashed, color = blue] (10.5,-5) -- (10.5,1.2);
			\draw[color=red] (4.6,-0.3) -- (6.4,-0.3);	
			\draw[color=red] (6.4,-0.3) -- (6.4, 1.1);
			\draw[color=red] (4.6,-0.3) -- (4.6,1.1);
			\draw[color=red] (4.6,1.1) -- (6.4,1.1);
			
			{
			\begin{scope}[shift={(0,-2)}]
			\draw[dotted] (1,0) -- (1.5,0);
			\draw[fill] (2,0) circle[radius=0.09]; 
			\draw[fill] (3,0) circle[radius=0.09];
			\draw[fill] (4,0) circle[radius=0.09];
			\draw[fill] (5,0) circle[radius=0.09];
			\draw[fill] (6,0) circle[radius=0.09];
			\draw[fill] (7,0) circle[radius=0.09];
			\draw[fill] (8,0) circle[radius=0.09];
			\draw[fill] (9,0) circle[radius=0.09];
			\draw[fill] (10,0) circle[radius=0.09];
			\draw[dotted] (11,0) -- (11.5,0);
			
			\node at (2,0) [above=0.3cm] {$\uparrow$};
			\node at (3,0) [above=0.3cm] {$\uparrow$};	
			\node at (4,0) [above=0.3cm] {$\downarrow$};	
			\node at (5,0) [above=0.3cm] {$\downarrow$};	
			\node at (6,0) [above=0.3cm] {$\downarrow$};
			\node at (7,0) [above=0.3cm] {$\downarrow$};	
			\node at (8,0) [above=0.3cm] {$\downarrow$};
			\node at (9,0) [above=0.3cm] {$\downarrow$};
			\node at (10,0) [above=0.3cm] {$\downarrow$};
			
			\draw[color=red] (1.6,-0.3) -- (3.4,-0.3);	
			\draw[color=red] (3.4,-0.3) -- (3.4, 1.1);
			\draw[color=red] (1.6,-0.3) -- (1.6,1.1);
			\draw[color=red] (1.6,1.1) -- (3.4,1.1);
			\end{scope}
			}
			\begin{scope}[shift={(0,-2.5)}]
			\node at (0.7, -0.25) [left=0.2cm] {$h_j$};
			\draw[->] (0.7,-2.2) -- (0.7,-0.25);
			\draw[dotted](0.5,-2) -- (1.5,-2);
			\draw[->] (1.5,-2) -- (11.5,-2);
			
			\draw[dotted,color=blue] (1,-1.3) to[out=40, in= 180] (2,-0.7); 
			\draw[color=blue] (2,-0.7) to[out=0, in= 180] (3,-1.8);  	
			\draw[color=blue] (3,-1.8) to[out=0, in= 220] (4,-1.3); 
			\draw[color=blue] (4,-1.3) to[out=40, in= 180] (5,-0.7); 
			\draw[color=blue] (5,-0.7) to[out=0, in= 180] (6,-1.8);  	
			\draw[color=blue] (6,-1.8) to[out=0, in= 220] (7,-1.3);	
			\draw[color=blue] (7,-1.3) to[out=40, in= 180] (8,-0.7); 
			\draw[color=blue] (8,-0.7) to[out=0, in= 180] (9,-1.8);  	
			\draw[color=blue] (9,-1.8) to[out=0, in= 220] (10,-1.3);
			\draw[color=blue,dotted] (10,-1.3) to[out=40, in= 180] (11.5,-0.7);
			
			\node[color=blue] at (2,-0.7) {$\boldsymbol{\times}$};
			\node[color=blue] at (2,-0.7) [above=0.1] {$h_1$};
			\node[color=blue] at (3,-1.8) {$\boldsymbol{\times}$};
			\node[color=blue] at (3,-1.8) [above=0.1] {$h_2$};
			\node[color=blue] at (4,-1.3) {$\boldsymbol{\times}$};
			\node[color=blue] at (4,-1.3) [above=0.1] {$h_3$};			
			
			\node[color=blue] at (5,-0.7) {$\boldsymbol{\times}$};
			\node[color=blue] at (5,-0.7) [above=0.1] {$h_1$};
			\node[color=blue] at (6,-1.8) {$\boldsymbol{\times}$};
			\node[color=blue] at (6,-1.8) [above=0.1] {$h_2$};
			\node[color=blue] at (7,-1.3) {$\boldsymbol{\times}$};
			\node[color=blue] at (7,-1.3) [above=0.1] {$h_3$};			
			\node[color=blue] at (8,-0.7) {$\boldsymbol{\times}$};
			\node[color=blue] at (8,-0.7) [above=0.1] {$h_1$};
			\node[color=blue] at (9,-1.8) {$\boldsymbol{\times}$};
			\node[color=blue] at (9,-1.8) [above=0.1] {$h_2$};
			\node[color=blue] at (10,-1.3) {$\boldsymbol{\times}$};
			\node[color=blue] at (10,-1.3) [above=0.1] {$h_3$};	
			\end{scope}
			\end{tikzpicture}
		\end{center}
	\caption{Pictorial representation of the dynamics of the quantum Ising model. Solid dots represent the lattice sites; above each dot, the arrow represents the (3)-component of the spin state of the particle on the related site. Two different time-lapses are represented vertically. Below them, a pictorial representation of the magnetic field is given (blue line). In each portion of lattice delimited by two vertical blue dashed lines there is (exactly) one site per element of the partition $\{\Lambda_{\alpha}\}$. In particular, from left to right, the first site belongs to $\Lambda_1$, the second site belongs to $\Lambda_2$, and the third site belongs to $\Lambda_3$. Starting from a configuration as in the first line, with only two spins up in two consecutive sites, conservation of the operators $N^{(\alpha)}$ implies that the two spin up packet has to move rigidly from one portion of the lattice $\Lambda$ delimited by two blue dashed lines to another.
	}
	\label{fig:calimero}
\end{figure}
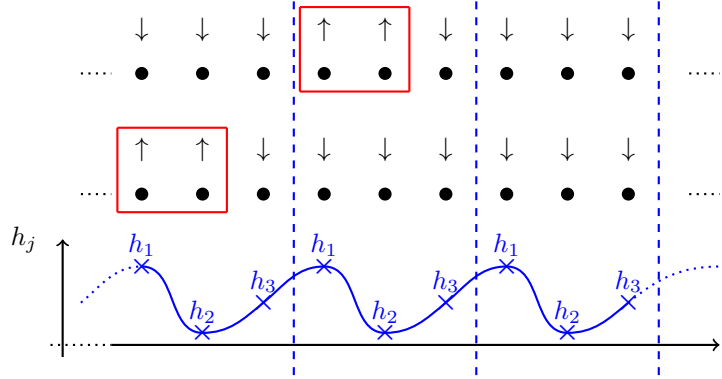

Let us now focus on the dynamical consequences of such quasi-conservation laws for particular initial states:
\begin{itemize}
\item If at time $t=0$ the system is in the state where all sites in the lattice $\Lambda_1$ have spin up and all the other spins are down, equation \eqref{vicini.small.ff} states that, calling $N_{\uparrow}^{(2)}(t),N_{\uparrow}^{(3)}(t)$ the number of sites with spin up in the sublattices $\Lambda_2$ and $\Lambda_3$ respectively, $N_\downarrow^{(1)}(t)$ the number of sites with spin down in the sublattice $\Lambda_1$, one has 
$$N_\downarrow^{(1)}(t) \leq 12 L K \lambda^{-\frac{1}{8\tau_\omega}} \, , \qquad N_\uparrow^{(2)}(t) \leq 12 L K \lambda^{-\frac{1}{8\tau_\omega}} \, , \qquad N_\uparrow^{(3)}(t) \leq 12 L K \lambda^{-\frac{1}{8\tau_\omega}} \, , $$
for times $t < e^{\lambda^{2\beta}} \lambda^{-\frac{1}{8\tau_\omega}}$. 
\item If at time $t=0$ the system is in the state where only two consecutive sites of the lattice have spin up and all the others have spin down, constraints on the dynamics are discussed in Figure \ref{fig:calimero}.
\end{itemize} 

One can easily generalize this example to Ising models in higher spatial dimensions and more general structures of the sublattices $\Lambda_\alpha$.


\section{Normal Form Results}\label{sec:NormalFormResults}
{A crucial role in the proof of Theorems \ref{teo:SlowHeating} and \ref{thm:EvLocObs} is played by the following normal form results.}
\begin{proposition}[Normal form for small perturbations]\label{teo:main}
	Let $H(t)$ satisfy assumptions of Theorem \ref{thm:EvLocObs}.
	Then there exists $\varepsilon_0 := \varepsilon_0(r, \tau, \gamma, \kappa, \rho, \Omega, \|V\|_{\kappa, \rho}, \{ \|N^{(\alpha)}\|_{0}\}_{\alpha}) >0$ such that $\forall\ 0< \varepsilon < \varepsilon_0$, the following holds.
	\begin{itemize}
	\item[(a)] There exist a unitary time quasi-periodic transformation $Y_{\mathrm{inv}}$ and a time quasi-periodic self-adjoint operator $H_{\mathrm{inv}}$ such that
	\begin{equation}\label{uH.uHstar}
	U_{H_{\mathrm{inv}}}(t) \;=\; Y_{\mathrm{inv}}(\omega t) U_{H}(t)  {Y_{{\mathrm{inv}}}^*(0)}
	\end{equation}
	and
	\begin{equation}
	H_{{\mathrm{inv}}}(\omega t)=J \cdot N + Z_{\mathrm{inv}}+V_{\mathrm{inv}}(\omega t)\,,
	\end{equation}
	with
	\begin{itemize}
		\item[(a.i)] $[Z_{\mathrm{inv}}, N^{(\alpha)}] = 0\quad \forall \alpha = 1, \dots, r$;
		\item[(a.ii)]  $Z_{\mathrm{inv}} \in \mathcal{O}^\strong_{\kappa_*}$, $\Vert Z_{\mathrm{inv}} \Vert_{\kappa_*} \leq {\frac{{e}}{e-1}} \varepsilon \Vert V \Vert_{\kappa, \rho}$;
		\item[(a.iii)]  $V_{\mathrm{inv}} \in \mathcal{O}^\strong_{\kappa_*, \rho_*}$, with $\Vert V_{\mathrm{inv}} \Vert_{\kappa_{*}, \rho_*} \leq e^{-{\varepsilon^{-2\tb}}} \Vert V \Vert_{\kappa,\rho}\,$;
		\item[(a.iv)] $\forall P \in \mathcal{O}_{\kappa_*, \rho_*},$ $\|Y_{\mathrm{inv}} P Y_{\mathrm{inv}}^* -P\|_{\kappa_{*}, \rho_*} \leq {\frac{2{e}}{e-1}} \varepsilon^{\frac 1 2} \|P\|_{\kappa_*, \rho_*}\,$;
	\end{itemize}
	\item[(b)] There exist a unitary time quasi-periodic transformation $Y_{\mathrm{obs}}$ and a time quasi-periodic self-adjoint operator $H_{\mathrm{obs}}$ such that
	\begin{equation}\label{uH.uHstar.Hobbes}
	U_{H_{\mathrm{obs}}}(t) \;=\; Y_{\mathrm{obs}}(\omega t) U_{H}(t) 
	\end{equation}
	and
	\begin{equation}
	H_{{\mathrm{obs}}}(\omega t)=J \cdot N + Z_{\mathrm{obs}}+V_{\mathrm{obs}}(\omega t)\,,
	\end{equation}
	with
	\begin{itemize}
		\item[(b.i)]  $Y_{\mathrm{obs}}(0) = \mathbbm{1}$;
		\item[(b.ii)]  $Z_{\mathrm{obs}} \in \mathcal{O}^\strong_{\kappa_*}$, $\Vert Z_{\mathrm{obs}} \Vert_{\kappa_*} \leq {\frac{{e}}{e-1}} \varepsilon^{\frac12} \Vert V \Vert_{\kappa, \rho}$;
		\item[(b.iii)] $V_{\mathrm{obs}} \in \mathcal{O}^\strong_{\kappa_*, \rho_*}$, with $\Vert V_{\mathrm{obs}} \Vert_{\kappa_{*}, \rho_*} \leq e^{-{\varepsilon^{-\tb}}} \Vert V \Vert_{\kappa,\rho}\,$;
		\item[(b.iv)] $\forall P \in \mathcal{O}_{\kappa_*, \rho_*},$ $\|Y_{\mathrm{obs}} P Y_{\mathrm{obs}}^* -P\|_{\kappa_{*}, \rho_*} \leq {\frac{2{e}}{e-1}} \varepsilon^{\frac 1 2} \|P\|_{\kappa_*, \rho_*}$. {More precisely, there exist} $n^* \in \N$ and self-adjoint operators $G^{(0)}_\obs, \dots, G^{({n^*-1})}_\obs \in \mathcal{O}^\strong_{\kappa_*, \rho_*}$ such that $Y_\obs(\omega t) = e^{-\ii G^{({n^*-1})}_\obs(\omega t)} \cdots e^{-\ii G^{(0)}_\obs(\omega t)}$, with $\|G^{(n)}_\obs\|_{\kappa_*, \rho_*} \leq \varepsilon^{\frac 1 2} e^{-n} \|V\|_{\kappa, \rho}$ for any $n$.
	\end{itemize}
	\end{itemize}
	Here $\tb$ is the constant defined in \eqref{parametri.vitali} and
	\begin{equation}\label{def:oggetti.finali}
	 \kappa_* := \frac{\min\{\kappa, \rho\}}{\sqrt{2}}\,, \quad \rho_* := \frac{\min\{\kappa,\rho\}}{8 \sqrt{2} r \max_\alpha\{\|N^{(\alpha)}\|_0\} }\,. 
	\end{equation}
\end{proposition}

\begin{remark}\label{pipistrello} The two normal form Hamiltonians $H_{\mathrm{inv}}$ and $H_{\mathrm{obs}}$ differ because of Item (a.i) and (b.i). In general, it is not possible to perform a unique normal form such that both Items (a.i) and (b.i) hold, as the case of stationary perturbations $V(\omega t) \equiv V$ easily shows. However, Item (a.i) plays a fundamental role in showing the quasi-invariance of the $N^{(\alpha)}$ operators (Theorem \ref{teo:SlowHeating}), while Item (b.i) is necessary to obtain the results on recurrence times in Theorem \ref{thm:EvLocObs}.
\end{remark}

{Analogously, a crucial role in the proof of Theorems \ref{thm:SlowHeatingFF} and \ref{thm:EvLocObs.ff} is played by the following:}

\begin{proposition}[Normal form {for fast forcing perturbations}]\label{teo:main.ff}
	Let $H(\lambda \omega t)$ be as in \eqref{eq:our.H.lambda} and suppose that {$\gamma_J>0$, $\gamma_\omega>0$, $\tau_J> r-1$, $\tau_\omega> m-1$ and $\mathcal{J}>0$ are such that $J$ satisfies \eqref{eq:siamo.dei.cani.ff.1} with $|J| \leq \mathcal{J}$ and $\omega$ satisfies \eqref{eq:siamo.dei.cani.ff.2}.} There exists $\lambda_0 = \lambda_0(r, \tau_J, \tau_\omega, \gamma_J, \gamma_\omega, {\mathcal{J}}, \kappa, \rho, {\{\| N^{(\alpha)} \|_{0}\}_\alpha}, \|V\|_{\kappa, \rho}) >0$ such that $\forall \lambda > \lambda_0$, the following holds.
	\begin{itemize}
	\item[(a)] There exist a unitary time quasi-periodic transformation $Y_{\mathrm{inv}}$ and a time quasi-periodic self-adjoint operator $H_{\mathrm{inv}}$ such that
	\begin{equation}
	U_{H_{\mathrm{inv}}}(t) \;=\; Y_{\mathrm{inv}}(\lambda \omega t) U_H(t) {Y_\inv^*(0)}
	\end{equation}
	and
	\begin{equation}
	H_{\mathrm{inv}}(\lambda \omega t)=J \cdot N + Z_{\mathrm{inv}}+V_{\mathrm{inv}}(\lambda \omega t)\,,
	\end{equation}
	with
	\begin{itemize}
		\item[(a.i)] $[Z_{\mathrm{inv}}, N^{(\alpha)}] = 0\quad \forall \alpha = 1, \dots, r$;
		\item[(a.ii)] $Z_{\mathrm{inv}} \in \mathcal{O}^\strong_{\kappa_*}$, $\Vert Z_{\mathrm{inv}} \Vert_{\kappa_*} \leq {\frac{e}{e-1}} \lambda^{-\frac{1}{4\tau_\omega}} \Vert V \Vert_{\kappa, \rho}$;
		\item[(a.iii)] $V_{\mathrm{inv}} \in \mathcal{O}^\strong_{\kappa_*, \rho_*}$, with $\Vert V_{\mathrm{inv}} \Vert_{\kappa_{*}, \rho_*} \leq e^{- \lambda^{2\beta}} \Vert V \Vert_{\kappa,\rho}$; 
		\item[(a.iv)] $\forall P \in \mathcal{O}_{\kappa_*, \rho_*},$ $\|Y_{\mathrm{inv}}(\lambda \omega t) P Y_{\mathrm{inv}}^*(\lambda \omega t)-P \|_{\kappa_{*}, \rho_*} \leq {\frac{2e}{e-1}} \lambda^{-\frac{1}{8 \tau_\omega}} \|P\|_{\kappa_*, \rho_*}$.
	\end{itemize}
	\item[(b)] There exist a unitary time quasi-periodic transformation $Y_{\mathrm{obs}}$ and a time quasi-periodic self-adjoint operator $H_{\mathrm{obs}}$ such that
	\begin{equation}
	U_{H_{\mathrm{obs}}}(t) \;=\; Y_{\mathrm{obs}}(\lambda \omega t) U_H(t)
	\end{equation}
	and
	\begin{equation}
	H_{\mathrm{obs}}(\lambda \omega t)=J \cdot N + Z_{\mathrm{obs}}+V_{\mathrm{obs}}(\lambda \omega t)\,,
	\end{equation}
	with
	\begin{itemize}
		\item[(b.i)] $Y_{\mathrm{obs}}(0) = \mathbbm{1}$;
		\item[(b.ii)] $Z_{\mathrm{obs}} \in \mathcal{O}^\strong_{\kappa_*}$, $\Vert Z_{\mathrm{obs}} \Vert_{\kappa_*} \leq {\frac{{e}}{e-1}} \lambda^{-\frac{1}{8 \tau_\omega}} \Vert V \Vert_{\kappa, \rho}$;
		\item[(b.iii)] $V_{\mathrm{obs}} \in \mathcal{O}^\strong_{\kappa_*, \rho_*}$, with $\Vert V_{\mathrm{obs}} \Vert_{\kappa_{*}, \rho_*} \leq e^{- \lambda^{\beta}} \Vert V \Vert_{\kappa,\rho},$;
		\item[(b.iv)]  $\forall P \in \mathcal{O}_{\kappa_*, \rho_*},$ $\|Y_{\mathrm{obs}}(\lambda \omega t) P Y_{\obs}^*(\lambda \omega t)-P \|_{\kappa_{*}, \rho_*} \leq {\frac{2{e}}{e-1}} \lambda^{-\frac{1}{8 \tau_\omega}} \|P\|_{\kappa_*, \rho_*}\,$. {More precisely, there} exist $n^\star \in \N$ and self-adjoint operators $G^{(0)}_\obs, \dots, G^{({n^\star-1})}_\obs \in \mathcal{O}^\strong_{\kappa_*, \rho_*}$ such that $Y_\obs(\lambda \omega t) = e^{-\ii G^{({n^\star-1})}_\obs(\lambda \omega t)} \cdots e^{-\ii G^{(0)}_\obs(\lambda \omega t)}$, with $\|G^{(n)}_\obs\|_{\kappa_*, \rho_*} \leq \lambda^{-\frac{1}{8 \tau_\omega}} e^{-n} \|V\|_{\kappa, \rho}$ for any $n$.
	\end{itemize}
	\end{itemize}
	Here $\beta$ is the constant defined in \eqref{parametri.vitali.ff} and $\kappa_*, \rho_*$ are defined as in \eqref{def:oggetti.finali}.
\end{proposition}

\begin{remark}\label{rem:GranchiettiEBussolai}
	The result in Item (b) of Proposition \ref{teo:main.ff} does not depend on the particular structure $J \cdot N$ of the unperturbed Hamiltonian; furthermore, with minor changes in the regularity parameters, it is proven in \cite{Else2020}. Here we state it in the particular case of $H_0 = J \cdot N$ only for self-consistency. The same does not hold true for Item (b) of Proposition \ref{teo:main}, which instead explicitly relies on the structure of $H_0$ and is one of the novelties of the present work.
\end{remark}

\subsection{Ideas of the Normal Form Proofs} \label{sec:Ideas}
The class of models we deal with is described by the Hamiltonian
\begin{equation}
	H(\nu t) \;=\; J \cdot N + W(\nu t)
\end{equation}
where {$\nu \in \R^m$ and} $W(\nu t)$ can be either a small perturbation (i.e.\ $W(\nu t)=\varepsilon V(\omega t)$, {$\varepsilon \ll 1$}) or a fast-forcing one (i.e.\ $W(\nu t)=V(\lambda \omega t)$, {$\lambda \gg 1$}).

We look for a coordinate transformation of the form $Y_1(\nu t)=e^{\ii G_0(\nu t)}$ generated by a quasi-local self-adjoint operator $G_0(\nu t) \in \mathcal{O}^\strong$ which reduces the size of the perturbation $W$. If $G_0$ is small (which will be the case), denoting by $H_1$ the transformed Hamiltonian (see Lemma \ref{lem.pushfwd}), we have
\begin{equation}
	H_1(\varphi)=J\cdot N +\nu \cdot \partial_\varphi G_0(\varphi) + W(\varphi) + \text{h.o.t.}
\end{equation}
with $\varphi:=\nu t$, where h.o.t.\ denotes terms of higher order in the perturbative parameter.
The core of the procedure is to solve for all $\varphi$ the so-called homological equation:
\begin{equation}\label{hom.0}
	\ii [J\cdot N,G_0(\varphi)]+ \nu \cdot \partial_\varphi G_0(\varphi)+ W(\varphi)-Z_1=0
\end{equation}
in the unknowns $G_0(\varphi)$ and $Z_1$, for a suitable time independent operator $Z_1$ (thus, independent of $\varphi$). 

{For our scopes, we have two possible solutions, depending on an additional requirement on the solution to the homological equation: either we require $[Z_1,N^{(\alpha)}]=0$ for any $\alpha = 1, \dots, r$, which leads to the normal form $H_\inv$, or we require $G(0)=0$, and this leads to the normal form $H_\obs$. (Both requirements are not compatible, in general, see Remark \ref{pipistrello}.)}

To solve equation \eqref{hom.0} we follow the ideas in \cite{BGMR_growth,QN, QNbari}, and we add an additional parametric dependence to our operators: we define
\begin{equation}\label{eq:HomologicalFormal}
	\mathcal{G}_0(\theta,\varphi)\;:=\;  e^{\ii \theta \cdot N} G_0(\varphi) e^{-\ii \theta \cdot N}  \, , \qquad \mathcal{W}(\theta,\varphi) \;:=\; e^{\ii \theta \cdot N} W(\varphi) e^{-\ii \theta \cdot N} \, ,
\end{equation}
and we look for a solution of
\begin{equation}\label{eq:Hom.Introduttiva}
\ii [J\cdot N,\mathcal{G}_0(\theta,\varphi)]+ \nu \cdot \partial_\varphi \mathcal{G}_0(\theta,\varphi)+ \mathcal{W}(\theta,\varphi)-Z_1=0
\end{equation}
for any $\theta \in \mathbb{T}^r$. Clearly, the solution $\mathcal{G}_0(\theta,\varphi)$ we find for any $\theta$, is a solution to \eqref{eq:HomologicalFormal} once evaluated at $\theta=0$. Since the spectrum of each of the number operators is integer, i.e. $\sigma(N^{(\alpha)})\subset \mathbb{Z}$ $\forall \alpha$, the dependence on $\theta$ of all the operators is quasi-periodic, allowing us to introduce their Fourier series based on the local structures as
\begin{equation}
	\mathcal{G}_0(\theta,\varphi)=\sum_{S \in \mathcal{P}_c(\Lambda)} \sum_{l \in \mathbb{Z}^m} \sum_{k \in \mathbb{Z}^r } \big( \widehat{\mathcal{G}_{0, S}} \big)_{l,k} \, e^{\ii k \cdot \theta} e^{\ii l \cdot \varphi} \, ,
\end{equation}
and analogously for $\mathcal{W}$. Then one solves \eqref{eq:Hom.Introduttiva} exploiting that $\ii [J\cdot N, \mathcal{G}_0(\theta,\varphi)] = J \cdot \partial_\theta \mathcal{G}_0(\theta,\varphi)$ and passing to Fourier coefficients.
In the solution of \eqref{eq:Hom.Introduttiva}, one has to deal with small denominators which are controlled differently in the small perturbation and fast forcing case. As mentioned in the Introduction, in the small perturbation case one imposes a Diophantine condition on the vector $(J,\omega)$ while on the fast-forcing case, one requires that $J$ and $\omega$ are Diophantine separately, and that $\lambda$ is large enough to separate the energy scales of the perturbation and the ones of the unperturbed Hamiltonian.
Once equation \eqref{hom.0} has been solved, we are left with $H_1(t) = J \cdot N + Z_1 + W_1(\nu t)$, with $W_1(\nu t)$ smaller in size, and we iterate the procedure a finite number of times to further reduce the size of the time dependent part.

\section{Strongly Local Operators and Their Properties}
\label{sec:Preliminaries}

\subsection{General Properties}
We start with the following result, which relates operator and local norms:
\begin{lemma}\label{lem.op.loc.norm}
	Let $t \mapsto A(\omega t)$ be such that $A \in \mathcal{O}_{0, 0}$. Then
	\begin{equation}\label{op.loc.1}
	\sup_{t \in \R}\|A(\omega t)\|_{\mathrm{op}} \leq |\Lambda| \|A\|_{0,0}\,.
	\end{equation}
	Furthermore, if $\forall t$ $A(\omega t) = \sum_{S'\subseteq S} A_{S'}(\omega t)$ for some $S \in \mathcal{P}_c(\Lambda)$, one has
	\begin{equation}\label{op.loc.S}
	\sup_{t \in \R} \|A(\omega t)\|_{\mathrm{op}} \leq |S| \|A\|_{0,0}\,.
	\end{equation}
\end{lemma}
\begin{proof}
	If $A \in \mathcal{O}_{0,0}$, one has
	\begin{align*}
	\sup_{t\in \R} \|A(\omega t)\|_{\mathrm{op}} &= \sup_{t \in \R} \left\|\sum_{S \in \mathcal{P}_c(\Lambda)}  \sum_{l \in \Z^m} (\widehat{A}_S)_l e^{\ii l \cdot \omega t}\right\|_{\mathrm{op}}
	\leq \sum_{x \in \Lambda} \sum_{S \in \mathcal{P}_c(\Lambda) \atop S \ni x} \sum_{l \in \Z^m} \|(\widehat{A}_S)_l\|_{\mathrm{op}}\\
	&\leq |\Lambda| \sup_{x \in \Lambda}  \sum_{S \in \mathcal{P}_c(\Lambda) \atop S \ni x} \sum_{l \in \Z^m} \|(\widehat{A}_S)_l\|_{\mathrm{op}} = |\Lambda| \|A\|_{0,0}\,.
	\end{align*}
	Suppose now $A = \sum_{S' \subseteq S} A_{S'}$; one has
	$$
	\begin{aligned}
	\sup_{t\in \R} \|A(\omega t)\|_{\mathrm{op}} &= \sup_{t \in \R} \left\|\sum_{S' \in \mathcal{P}_c(\Lambda) \atop S' \subseteq S}  \sum_{l \in \Z^m} (\widehat{A}_{S'})_l e^{\ii l \cdot \omega t}\right\|_{\mathrm{op}}
	\leq \sum_{x \in S} \sum_{S' \subseteq S \atop S' \ni x} \sum_{l \in \Z^m} \|(\widehat{A}_{S'})_l\|_{\mathrm{op}}\\
	&\leq |S| \sup_{x \in \Lambda}  \sum_{S' \in \mathcal{P}_c(\Lambda) \atop S' \ni x} \sum_{l \in \Z^m} \|(\widehat{A}_{S'})_l\|_{\mathrm{op}} = |S| \|A\|_{0,0}\,.
	\end{aligned}
	$$
\end{proof}
The normal form procedure described in Section \ref{sec:Ideas} consists in a conjugation of the Hamiltonian with coordinate transformations of the form $e^{\ii G(\nu t)}$, with $G$ being a family of strongly exponentially localized self-adjoint operators. Thus, considering a generic $A \in \mathcal{O}^\strong_{\kappa,\rho}$, we need to control {locality} properties of the quantity
\begin{equation}\label{eq:ConjugationAd}
e^{\ii G(\nu t)} A (\nu t) e^{-\ii G(\nu t)} \;=\; \sum_{q=0}^\infty \frac{(-\ii)^q}{q!} \mathrm{Ad}_{ G(\nu t)}^q A(\nu t) \, ,
\end{equation}
where we have defined
\begin{equation}
\mathrm{Ad}^0_{B} A := A \, ,\qquad \mathrm{Ad}^q_{B} A := [\mathrm{Ad}^{q-1}_B A , B] \qquad \forall q \geq 1, \quad \forall A,B \in \mathcal{O}_{\kappa,\rho}\, .
\end{equation}
This is done in Lemma \ref{cor.lignano.pineta} below. We point out that an analogous result to Lemma \ref{cor.lignano.pineta}, with different quantitative estimates, was given in \cite{Abanin2017} for time periodic operators. For this reason, here we only state it and we postpone a detailed proof of it, suited to our setting, to Appendix \ref{append:ProofTec}.  

\begin{lemma}[see also Lemma 4.1 of \cite{Abanin2017}]\label{cor.lignano.pineta}
	Let $A, B \in \mathcal{O}^\strong_{\kappa + \sigma, \rho}$ for some $\kappa, \rho \geq 0$ and $\sigma >0$, and assume that there exists $\eta \in (0, 1)$ such that
	\begin{equation}
	\frac{4 e^{-\kappa}\Vert A\Vert_{\kappa + \sigma, \rho}}{\sigma} \leq \eta\,.
	\end{equation}
	Then,
	\begin{gather}
	\label{anna}
	\Big\Vert e^{A} B e^{-A}-B\Big\Vert_{\kappa, \rho} \leq C e^{-\kappa} \frac{1}{\sigma} \Vert A\Vert_{\kappa + \sigma, \rho} \Vert B \Vert_{\kappa + \sigma, \rho}\,,\\
	\label{ha.ragione}
	\Big\Vert e^{A} B e^{-A}-B-[A,B] \Big\Vert_{\kappa, \rho} \leq 4 C e^{-2\kappa} \frac{1}{\sigma^2} \Vert A\Vert^2_{\kappa + \sigma, \rho} \Vert B \Vert_{\kappa + \sigma, \rho}\,,
	\end{gather}
	with $C=\frac{4}{1-\eta}>0$.
\end{lemma}

\subsection{Conjugation with the Number Operators}

In the normal form construction we need to exploit the notion of  \emph{strong} locality to obtain that for a strongly local operator $A$, 
\begin{equation}\label{finanziamo.trenitalia}
	\big(e^{\ii \theta \cdot N} A e^{-\ii \theta \cdot N} \big)_{S}\;=\;e^{\ii \theta \cdot N} A_S e^{-\ii \theta \cdot N} \, .
\end{equation}
The aim of this Subsection is to prove \eqref{finanziamo.trenitalia} and some properties of $e^{\ii \theta \cdot N} A e^{-\ii \theta \cdot N}$.
First we prove that the series $\sum_{q} \frac{(-\ii)^q}{q!} \mathrm{Ad}_{\theta \cdot N}^q A$ converges uniformly in $\theta$ (but with a bad dependence on the volume $|\Lambda|$):
	\begin{lemma}\label{lem.ad.q.lambda}
	Fix $\alpha \in \{1, \dots, r\}$ and suppose that $A = A_S$ is a strongly local operator. Then $\forall q \in \N$ $\textnormal{Ad}^q_{N^{(\alpha)}} A_S$ is a strongly local operator, supported on $S$ only, and $\forall \alpha$
	\begin{equation}\label{q.comm.con.N}
	\left\| \textnormal{Ad}_{N^{(\alpha)}}^q A_S \right\|_{\mathrm{op}}   =	\left\| \left(\textnormal{Ad}_{N^{(\alpha)}}^q A_S\right)_S\right\|_{\mathrm{op}}\leq 2^q |S|^q \|N^{(\alpha)}\|^q_{0} \|{A_S}\|_{\mathrm{op}} \leq  2^q |\Lambda|^q \|N^{(\alpha)}\|^q_{0} \|{A_S}\|_{\mathrm{op}}\,.
	\end{equation}
	\end{lemma}
	\begin{proof}
	We inductively prove that the first inequality in \eqref{q.comm.con.N} holds, with the second one trivially following from $|S| \leq |\Lambda|.$
	If $q=1$, one has
	\begin{equation}\label{stupida.1}
	\textnormal{Ad}^1_{N^{(\alpha)}} A_S = \left[A_S, N^{(\alpha)}\right] = \sum_{S' \subseteq S}  \left[A_S, N^{(\alpha)}_{S'}\right]\,,
	\end{equation}
	thus $\textrm{Ad}^1_{N^{(\alpha)}} A_S$ is localized on $S$. We now observe that, by \eqref{stupida.1}, one also has
	$$
	\begin{aligned}
	\|\textnormal{Ad}^1_{N^{(\alpha)}} A_S\|_{\mathrm{op}} &\leq 	2 \sum_{ S' \subseteq S} \|A_S\|_{\mathrm{op}} \|N^{(\alpha)}_{S'}\|_{\mathrm{op}}\leq 2  \|A_S\|_{\mathrm{op}} \sum_{S' \subseteq S} \|N^{(\alpha)}_{S'}\|_{\mathrm{op}}\\
	&\hspace{-17pt}\leq 2 \|A_S\|_{\mathrm{op}} \sum_{x \in S} \sum_{S' \ni x \atop S' \subseteq S} \|N^{(\alpha)}_{S'}\|_{\mathrm{op}}  \leq  2 \|A_S\|_{\mathrm{op}} |S| \sup_{x \in \Lambda} \sum_{S' \ni x} \|N^{(\alpha)}_{S'}\|_{\mathrm{op}} = 2 \|A_S\|_{\mathrm{op}}  |S| \|N^{(\alpha)}\|_{0}\,,
	\end{aligned}
	$$
	which gives \eqref{q.comm.con.N} for $q=1$. Now let us suppose that $\textnormal{Ad}^{q}_{N^{(\alpha)}} A_S$ is a local operator acting within $S$ and that \eqref{q.comm.con.N} holds for some $q$; then one has
	$$
	\begin{aligned}
	\textnormal{Ad}^{q+1}_{N^{(\alpha)}} A_S = \left[	\textnormal{Ad}^{q}_{N^{(\alpha)}} A_S, N^{(\alpha)}\right] =\sum_{S' \subseteq S} \left[	\textnormal{Ad}^{q}_{N^{(\alpha)}} A_S, N^{(\alpha)}_{S'}\right] \,,
	\end{aligned}
	$$
	which shows that $\textnormal{Ad}^{q+1}_{N^{(\alpha)}} A_S$ is a local operator acting on $S$; furthermore, arguing as above one obtains
	$$
	\begin{aligned}
	\|\textnormal{Ad}^{q+1}_{N^{(\alpha)}} A_S\|_{\mathrm{op}}
	&  \leq  2 \left\|\textnormal{Ad}^{q}_{N^{(\alpha)}} A_S \right\|_{\mathrm{op}} |S| \|N^{(\alpha)}\|_{0}\\
	&\leq 2^{q+1}  |S|^{q+1} \|N^{(\alpha)}\|^{q+1}_{0}  \|A_S\|_{\mathrm{op}}\,,
	\end{aligned}
	$$
	which gives \eqref{q.comm.con.N} for $q+1$. The fact that $\mathrm{Ad}_{N^{(\alpha)}}^q A_S$ is strongly local follows as in the proof of Lemma \ref{cor.lignano.pineta}, using the fact that the operators $N^{(\alpha)}$ are strongly {local}.
\end{proof}

As a second step, we have the following.
	\begin{lemma}\label{cor.converge.male}
	Suppose that $A$ is a quasi-periodic family of strongly {local} operators, with
	$$
	A(\varphi) = \sum_{S \in \mathcal{P}_c(\Lambda)} A_S(\varphi)\,, 
	$$
 then $\forall \theta \in \R^r$
	\begin{equation}\label{decomposizione.1}
	\mathcal{A}(\theta, \varphi):= e^{\ii \theta \cdot N} A(\varphi) e^{-\ii \theta \cdot N} = \sum_{S \in \mathcal{P}_c(\Lambda)}	\left(e^{\ii \theta\cdot N} A(\varphi) e^{-\ii \theta \cdot N}\right)_S\,
	\end{equation}
	is a strongly {local} operator and $\forall S$
	\begin{equation}\label{decomposizione.2}
	\left(e^{\ii \theta \cdot N} A(\varphi) e^{-\ii \theta \cdot N}\right)_S = e^{\ii \theta \cdot N} \left( A(\varphi)\right)_S e^{-\ii \theta \cdot N} \, .
	\end{equation}
	\end{lemma}
	\begin{proof}
	Since the operators $\{N^{(\alpha)}\}_{\alpha = 1}^r$ mutually commute, we observe that
	$$
	e^{\ii \theta \cdot N} A(\varphi) e^{-\ii \theta \cdot N} = e^{\ii \theta_r N^{(r)}} \cdots e^{\ii \theta_1 N^{(1)}} A(\varphi) e^{-\ii \theta_1 N^{(1)}} \cdots e^{-\ii \theta_r N^{(r)}}\,.
	$$
	Define for any $\alpha = 1, \dots, r$ $\mathcal{A}_{\alpha}(\theta, \varphi):= 	e^{\ii \theta_\alpha N^{(\alpha)}} \mathcal{A}_{\alpha -1}(\varphi) e^{-\ii \theta_\alpha N^{(\alpha)}}$, $\mathcal{A}_{0}(\theta, \varphi) = A(\varphi)$, so that $\mathcal{A}_r(\theta, \varphi) = \mathcal{A}(\varphi)$; 
	we now use Lemma \ref{lem.ad.q.lambda} to show inductively on $\alpha$ that
	\begin{equation}\label{uno.per.volta}
	\begin{split}
	\mathcal{A}_{\alpha}(\theta, \varphi) &= \sum_{S \in \mathcal{P}_c(\Lambda)} \left( \mathcal{A}_{\alpha}(\theta, \varphi)\right)_S \,,\\
	 \left( \mathcal{A}_{\alpha}(\theta, \varphi)\right)_S &= e^{\ii \theta_\alpha N^{(\alpha)}} \left( \mathcal{A}_{\alpha-1}(\theta,\varphi)\right)_S e^{-\ii \theta_\alpha N^{(\alpha)}} \in \mathscr{B}(\cH_\Lambda) \quad \forall S \in \mathcal{P}_c(\Lambda)\,,
	\end{split}
	\end{equation}
	and that $\forall S \in \mathcal{P}_c(\Lambda)$, $ \left( \mathcal{A}_{\alpha}(\theta, \varphi)\right)_S $ is a strongly local operator acting within $S$.
	Once \eqref{uno.per.volta} has been proven, one recursively deduces that
	$$
 	 e^{\ii \theta \cdot N} A(\theta, \varphi) e^{-\ii \theta \cdot N}= \mathcal{A}_{r}(\theta, \varphi) = \sum_{S \in \mathcal{P}_c(\Lambda)} \left(\mathcal{A}_{r}(\theta, \varphi)\right)_S = \sum_{S \in \mathcal{P}_c(\Lambda)} e^{\ii \theta \cdot N}  \left(A(\varphi)\right)_S e^{-\ii \theta \cdot N}\,,
	$$
	with the operators $e^{\ii \theta \cdot N}  \left(A(\varphi)\right)_S e^{-\ii \theta \cdot N} \in \mathscr{B}(\cH_\Lambda)$ being strongly local operators acting within $S$,
	which gives the thesis. 
	
	We now prove \eqref{uno.per.volta}.	First of all, we observe that
	\begin{equation}\label{per.carita}
	\begin{aligned}
	\mathcal{A}_{\alpha}(\theta, \varphi) &= e^{\ii \theta_\alpha N^{(\alpha)}}  \mathcal{A}_{\alpha-1}(\theta, \varphi) e^{-\ii \theta_\alpha N^{(\alpha)}}\\
	&= \sum_{q \geq 0} \frac{(-\ii \theta_\alpha)^q}{q!} \textnormal{Ad}^q_{N^{(\alpha)}}  \mathcal{A}_{\alpha-1}(\theta, \varphi)
	= \sum_{q \geq 0} \frac{(-\ii\theta_\alpha)^q}{q!} \sum_{S \in \mathcal{P}_c(\Lambda)} \textnormal{Ad}^q_{N^{(\alpha)}}  \left(\mathcal{A}_{\alpha-1}(\theta, \varphi)\right)_S\,.
	\end{aligned}
	\end{equation}
	Since by inductive hypothesis  $\left(\mathcal{A}_{\alpha-1}(\theta, \varphi)\right)_S \in \mathscr{B}(\cH_\Lambda)$ is a strongly {local} operator acting within $S$, by Lemma \ref{lem.ad.q.lambda} with $A$ replaced by $\left(\mathcal{A}_{\alpha-1}(\theta, \varphi)\right)_S$ one has, $\forall S$
	\begin{equation}\label{bombie.the.zombie}
	e^{\ii \theta_\alpha N^{(\alpha)}} \left(\mathcal{A}_{\alpha-1}(\theta, \varphi)\right)_S e^{-\ii \theta_\alpha N^{(\alpha)}} =\sum_{q \geq 0} \frac{(-\ii\theta_\alpha)^q}{q!} \textnormal{Ad}^q_{N^{(\alpha)}}  \left(\mathcal{A}_{\alpha-1}(\theta, \varphi)\right)_S \in \mathscr{B}(\cH_\Lambda)\,,
	\end{equation}
	and
	$e^{\ii \theta_\alpha N^{(\alpha)}} \left(\mathcal{A}_{\alpha-1}(\theta, \varphi)\right)_S e^{-\ii \theta_\alpha N^{(\alpha)}}$ is a strongly local operator acting within $S$, since by Lemma~\ref{lem.ad.q.lambda} the operators $\textnormal{Ad}^q_{N^{(\alpha)}}  \left(\mathcal{A}_{\alpha-1}(\theta, \varphi)\right)_S$ are strongly local and act within $S$, and the series in \eqref{bombie.the.zombie} is convergent, again due to Lemma \ref{lem.ad.q.lambda}. We are allowed to exchange the sum over $S$ and over $q$ in \eqref{per.carita} and then \eqref{uno.per.volta} follows immediately.
	\end{proof}

As a consequence, the following is the main technical result we are using in the analysis of the homological equation.	
	
	\begin{lemma}\label{lem:op.tau}
	Given $\kappa, \sigma, \rho >0$ and $A \in \mathcal{O}^{\strong}_{\kappa + \sigma, \rho}$, let
	\begin{equation}\label{A.tau}
	\mathcal{A}(\theta, \varphi) := e^{\ii \theta \cdot N} A(\varphi) e^{-\ii \theta \cdot N}\,, \quad \theta \in \mathbb{R}^r\,.
	\end{equation}
	Then
	\begin{itemize}
	\item[(i)] $\forall \theta \in \R^r$, $\mathcal{A}(\theta, \cdot) \in \mathcal{O}^\strong_{\kappa + \sigma, \rho}$, with $\|\mathcal{A}(\theta, \cdot)\|_{\kappa + \sigma, \rho} = \|A\|_{\kappa + \sigma, \rho}$.
	\item[(ii)]$\forall S \in \mathcal{P}_c(\Lambda)$, $\mathcal{A}_S$ is periodic w.r.t. $\theta,$ namely $\mathcal{A}_S(\theta + 2\pi e_j, \varphi) = \mathcal{A}_S(\theta, \varphi)$ $\forall j =1, \dots, r$. In particular, $\mathcal{A}$ admits the representation
	\begin{equation}\label{A.tau.fou}
	\begin{split}
	\mathcal{A}(\theta, \varphi) &= \sum_{S \in \mathcal{P}_c(\Lambda)} \sum_{k \in \Z^r} \sum_{l \in \Z^m} (\widehat{\mathcal{A}_S})_{l,k} \, e^{\ii k \cdot \theta} e^{\ii l \cdot \varphi}\,, \\
	(\widehat{\mathcal{A}_S})_{l,k} &:= \frac{1}{(2\pi)^{r+m}} \int_{\mathbb{T}^m}\int_{\mathbb{T}^r} \big(\mathcal{A}(\theta, \varphi)\big)_{S} e^{-\ii \varphi \cdot l} e^{-\ii \theta \cdot k} \, \ud \theta \ud \varphi \quad \forall S \in \mathcal{P}_c(\Lambda),l \in \mathbb{Z}^r,k \in \mathbb{Z}^m\,.
	\end{split}
	\end{equation}
	\item[(iii)] Defining
	\begin{equation}\label{def:norma.tau}
	\| \mathcal{A}\|_{\kappa, \rho, \zeta} := \sup_{x \in \Lambda} \sum_{S \in \mathcal{P}_c(\Lambda)} \sum_{k \in \Z^r} \sum_{l \in \Z^m} \big\| (\widehat{\mathcal{A}_S})_{l,k}\big \|_{\mathrm{op}} e^{\zeta |k|} e^{\rho |l|} e^{\kappa |S|}\,,
	\end{equation}
	$\forall \zeta>0$ such that
	\begin{equation}\label{zitina}
	\frac{4 \zeta \max_{\alpha} \|N^{(\alpha)}\|_{0}}{\sigma} \leq \frac 1 2\,,
	\end{equation}
	one has
	\begin{equation}\label{eq:norma.A.norma.A.tau}
	\| \mathcal{A}\|_{\kappa, \rho, \zeta} \leq  2^{2r}\left({ \frac{1+\zeta}{\zeta}}\right)^r\|A\|_{\kappa + r\sigma, \rho}\,,
	\end{equation}
	and
	\begin{equation}\label{mandarino}
	\|A\|_{\kappa, \rho} \leq  \| \mathcal{A}\|_{\kappa, \rho, 0} \leq  \| \mathcal{A}\|_{\kappa, \rho, \zeta}\,. 
	\end{equation}
	\end{itemize} 
	\end{lemma}
	In the following, given $\kappa, \rho, \zeta \geq 0$, we will denote
	\begin{equation}\label{ostrong.tau}
	\mathcal{O}^\strong_{\kappa, \rho, \zeta} :=\big\{ \mathcal{A} : \mathbb{T}^r\times \mathbb{T}^m \rightarrow \mathscr{B}(\cH_\Lambda)\ |\ \forall (\theta, \varphi)\quad \mathcal{A}(\theta, \varphi) \in \mathcal{O}^\strong_\kappa \ \textrm{ and }\|\mathcal{A}\|_{\kappa, \rho, \zeta} < \infty\big\}\,,
	\end{equation}
	where $\| \cdot \|_{\kappa, \rho, \zeta}$ is defined as in \eqref{def:norma.tau}.
	
	\begin{proof}[Proof of Lemma \ref{lem:op.tau}]
	By Lemma \ref{cor.converge.male}, one has $	\mathcal{A}(\theta, \varphi) = \sum_{S \in \mathcal{P}_c(\Lambda) } \left(\mathcal{A}(\theta, \varphi)\right)_S$,
	\begin{equation}
	\begin{split}
	\left(\mathcal{A}(\theta, \varphi)\right)_S &=  e^{\ii \theta \cdot N} \left( A(\varphi)\right)_S  e^{-\ii \theta \cdot N} \quad \forall S \in \mathcal{P}_c(\Lambda)\,,
	\end{split}
	\end{equation}
	with $\left(\mathcal{A}(\theta, \varphi)\right)_S$ strongly local operators acting within $S$. Since the operators $N^{(\alpha)}$ are independent of $\varphi$, one has that the operators $\left(\mathcal{A}(\theta, \cdot )\right)_S$ are still periodic w.r.t $\varphi$, and
	\begin{equation}\label{sanza.tempo}
	(\widehat{(\mathcal{A}(\theta, \cdot ))_S})_{l} = e^{\ii \theta \cdot N} (\widehat{A_S})_{l}  e^{-\ii \theta \cdot N} \quad \forall l \in \Z^m\,.
	\end{equation}
	Thus, one has $\forall \theta \in \R^r$
	\begin{align*}
	\left\| \mathcal{A}(\theta, \cdot )\right\|_{\kappa, \rho} &= \sup_{x \in \Lambda} \sum_{S \in \mathcal{P}_c(\Lambda) \atop \text{s.t.} x \in S} \sum_{l \in \Z^m} \big\|\big(\widehat{\big(\mathcal{A}(\theta, \cdot )\big)_S}\big)_{l}\big\|_{\mathrm{op}} e^{\rho|l|} e^{\kappa|S|} \\
	&\hspace{-5pt}\stackrel{\eqref{sanza.tempo}}{=} \sup_{x \in \Lambda} \sum_{S \in \mathcal{P}_c(\Lambda) \atop \text{s.t.} x \in S} \sum_{l \in \Z^m} \| (\widehat{A_S})_{l}   \|_{\mathrm{op}} e^{\rho|l|} e^{\kappa|S|}= \|A\|_{\kappa, \rho}\,,
	\end{align*}
	and Item (i) is proven.
	
	Item (ii) follows from spectral theorem and from the fact that the spectrum of the operators $N^{(\alpha)}$ is integer.
	
	It remains to prove Item (iii). First note that, since $A(\,\cdot\,) = \mathcal{A}(0,\,\cdot\,)$ and $\sup_{\theta \in \T^r} \|\mathcal{A}(\theta, \cdot) \|_{\kappa, \rho} \leq \|\mathcal{A}\|_{\kappa, \rho, 0}$, one has $\|A\|_{\kappa, \rho} = \|\mathcal{A}(0, \cdot)\|_{\kappa, \rho} \leq \|\mathcal{A}\|_{\kappa, \rho, 0}$. Thus, using the monotonicity of $\| \cdot\|_{\kappa, \rho, \zeta}$ in $\zeta$, \eqref{mandarino} is proven. We now prove \eqref{eq:norma.A.norma.A.tau}.
	
	We observe that, by \eqref{A.tau.fou}, $\forall k \in \mathbb{Z}^r$ one has
	$$
	(\widehat{\mathcal{A}_S})_{l,k} = \frac{1}{(2\pi)^r}\int_{\mathbb{T}^r} e^{\ii \theta \cdot N} (\widehat{A_S})_{l}\,  e^{-\ii \theta \cdot N}  e^{-\ii \theta \cdot k} \, \ud\theta\,.
	$$
We first claim that, $\forall k \in \mathbb{Z}^r$, $\forall l \in \mathbb{Z}^m$ and $\forall S \in \mathcal{P}_c(\Lambda)$ and $\forall \zeta > 0$ one has
\begin{eqnarray}
	\Vert (\mathcal{A}_S)_{l,k} \Vert_{\mathrm{op}} &\leq& e^{-2 \zeta |k|_1} \sup_{\upsilon \in \Z^r\,, |\upsilon|_\infty {\leq 1}} \left\{\left\| e^{2\zeta \upsilon \cdot N} (\widehat{{A}_S})_{l}  e^{-2\zeta \upsilon \cdot N}\right\|_{\mathrm{op}} \right\} \label{decay}\\
	&\leq & e^{-2 \zeta |k|_1} 2^r \Vert (\widehat{A_S})_l \Vert_{\mathrm{op}} e^{r \sigma |S|} \, . \label{u.bound}
\end{eqnarray}
Then, one concludes as follows:
\begin{equation*}
	\begin{aligned}
	\left\| \mathcal{A}\right\|_{\kappa, \rho, \zeta} &= \sup_{x \in \Lambda} \sum_{\substack{S \in \mathcal{P}_c(\Lambda) \\ x \in S}} \sum_{k \in \Z^r} \sum_{l \in \Z^m} \| (\widehat{\mathcal{A}_S})_{l,k} \|_{\mathrm{op}} e^{\zeta |k|} e^{\rho |l|} e^{\kappa |S|}\\
	&\leq \sup_{x \in \Lambda} \sum_{\substack{S \in \mathcal{P}_c(\Lambda) \\ x \in S}} \sum_{k \in \Z^r} \sum_{l \in \Z^m}  2^r \|(\widehat{{A}_S})_{l}\|_{\mathrm{op}} e^{-2\zeta |k|_1 + \zeta |k|} e^{\rho |l|} e^{(r \sigma + \kappa) |S|}\\
	&\leq \sup_{x \in \Lambda} \sum_{\substack{S \in \mathcal{P}_c(\Lambda) \\ x \in S}} \sum_{k \in \Z^r} \sum_{l \in \Z^m}  2^r \|(\widehat{{A}_S})_{l}\|_{\mathrm{op}} e^{-\zeta|k|_1} e^{\rho |l|} e^{(r \sigma + \kappa) |S|}\\
	&= \frac{2^{2 r}}{(1-e^{-\zeta})^r}\|A\|_{\kappa + r \sigma, \rho} \leq 2^{2r} \left(\frac{\zeta+1}{\zeta}\right)^r \Vert A \Vert_{\kappa+r\sigma,\rho}\,.
	\end{aligned}
\end{equation*}
It remains to prove \eqref{decay} and \eqref{u.bound}. 

To prove \eqref{decay}, one has to estimate in operator norm
\[
	 (\widehat{\mathcal{A}_S})_{l,k} = \frac{1}{(2\pi)^r} \int_0^{2\pi} \dots \int_0^{2\pi} e^{\ii \theta \cdot N} (\widehat{A_S})_{l} e^{-\ii \theta \cdot N} e^{-\ii \theta \cdot k} \, \ud \theta_1 \dots \ud \theta_r \, .
\]
This is done using the fact that $(\widehat{\mathcal{A}_S})_{l,k}$ is analytic on a strip of width $2 \zeta$ around the real segment $[0,2\pi]$ (with periodic boundary conditions) and using Cauchy Theorem on integration of analytic functions. For simplicity let us consider the case $r=1$ and suppose $k_1>0$, one has
\begin{align*}
	 (\widehat{\mathcal{A}_S})_{l,k} &=  \frac{1}{2\pi} \int_{0}^{2\pi} e^{\ii \theta_1  N_1} (\widehat{A_S})_{l} e^{-\ii \theta_1  N_1} e^{-\ii \theta_1 k_1} \,\ud \theta_1 =
	  \frac{1}{2\pi} \int_{-2\ii \zeta}^{2\pi -2\ii \zeta}  e^{\ii \theta_1  N_1} (\widehat{A_S})_{l} e^{-\ii \theta_1N_1} e^{-\ii \theta_1 k_1}  \,\ud \theta_1\\
	 &=  \frac{1}{2\pi} \int_{0}^{2\pi} e^{\ii (\theta_1 - 2\ii \zeta) N_1} (\widehat{A_S})_{l} e^{-\ii (\theta_1 - 2\ii \zeta ) N_1} e^{-\ii (\theta_1 - 2\ii \zeta) k_1} \,\ud \theta_1 \, .
\end{align*}
Taking $\| \cdot \|_{\mathrm{op}}$ norms on both sides, one gets
	\begin{align*}
	\left\|  (\widehat{\mathcal{A}_S})_{l,k} \right\|_{\mathrm{op}} &\leq 
	 \frac{1}{2\pi} e^{- 2\zeta k_1} \int_{\mathbb{T}} \left\| e^{\ii \theta_1 N_1} e^{ 2 \zeta N_1} (\widehat{A_S})_{l}   e^{-2  \zeta  N_1} e^{-\ii \theta_1 N_1}\right\|_{\mathrm{op}} \,\ud \theta_1\\
	 &= \frac{1}{2\pi} e^{- 2\zeta k_1} \int_{\mathbb{T}^r} \left\|e^{ 2 \zeta N_1} (\widehat{A_S})_{l}   e^{ -2 \zeta  N_1} \right\|_{\mathrm{op}} \,\ud \theta_1\\
	 &\leq  e^{-2\zeta k_1} \left\|e^{ 2\zeta N_1} (\widehat{A_S})_{l}   e^{ -2 \zeta N_1} \right\|_{\mathrm{op}}\,.
	\end{align*}
Repeating the argument in the general case yields
\[
	\left\|  (\widehat{\mathcal{A}_S})_{l,k} \right\|_{\mathrm{op}} 
\leq  e^{-2\zeta |k|_1} \left\|e^{ 2\textrm{sign}(k) \zeta \cdot N} (\widehat{A_S})_{l}   e^{- 2\textrm{sign}(k) \zeta \cdot N} \right\|_{\mathrm{op}}\,,
\]
where, for $a \in \R^r$, $\R^r \ni \textrm{sign}(a) := (\textrm{sign}(a_1), \dots, \textrm{sign}(a_r))$. Noting that $\upsilon:=\mathrm{sign}(a)$ has $|\upsilon|_\infty\leq1$, this proves \eqref{decay}.

		To prove \eqref{u.bound}, we argue as in the proof of Lemma \ref{lem.ad.q.lambda}. For any $\upsilon, S$ and $l$ and any $\alpha = 0, \dots, r$ we define $\mathscr{A}_{\alpha} := e^{2\zeta \upsilon_\alpha N^{(\alpha)}} \mathscr{A}_{\alpha-1} e^{-2\zeta \upsilon_\alpha N^{(\alpha)}}$, $\mathscr{A}_{0} := (\widehat{{A}_S})_{l}$ (note that, as a consequence of Lemma~\ref{lem.ad.q.lambda}, the $\mathscr{A}_{\alpha}$'s are strongly local operators acting within a certain $S \in \mathcal{P}_c(\Lambda)$ and they are not the operators $\mathcal{A}$ defined in the proof of Lemma \ref{cor.converge.male}) and we prove that
	\begin{equation}\label{u.bound.recursive}
	\left\| \mathscr{A}_{\alpha}\right\|_{\mathrm{op}} \leq 2 e^{\sigma |S|} 	\left\| \mathscr{A}_{\alpha-1}\right\|_{\mathrm{op}} \,,
	\end{equation}
	and that $\mathscr{A}_\alpha$ is a local operator acting within $S$. Then, \eqref{u.bound} will follow observing that $ e^{2\zeta \upsilon \cdot N} (\widehat{{A}_S})_{l}  e^{-2\zeta \upsilon \cdot N} = \mathscr{A}_{r}$, applying recursively \eqref{u.bound.recursive} and taking the supremum over $\upsilon$.\\
	To prove that \eqref{u.bound.recursive} holds, for any $\alpha = 1, \dots, r$ we observe that
	\begin{equation}\label{giornata}
	\begin{aligned}
	\left\| \mathscr{A}_{\alpha} \right\|_{\mathrm{op}} &= \left\|  e^{2\zeta \upsilon_\alpha N^{(\alpha)}} \mathscr{A}_{\alpha-1} e^{-2\zeta \upsilon_\alpha N^{(\alpha)}} \right\|_{\mathrm{op}} = \Big\| \sum_{q \geq 0} \frac{(2\zeta \upsilon_\alpha)^q}{q!} \textnormal{Ad}_{N^{(\alpha)}}^q \mathscr{A}_{\alpha-1} \Big\|_{\mathrm{op}}
	\end{aligned}
	\end{equation}
	and, by Lemma \ref{lem.ad.q.lambda}, since $\mathscr{A}_{\alpha-1}$ is a local operator acting only within $S$, one has
	\begin{equation}\label{fine}
	\begin{aligned}
	\left\| \mathscr{A}_{\alpha} \right\|_{\mathrm{op}} &\leq \sum_{q \geq 0} \frac{(2\zeta \upsilon_\alpha)^q}{q!} (2 |S| \|N^{(\alpha)}\|_{0})^q \left\|  \mathscr{A}_{\alpha-1}\right\|_{\mathrm{op}}\\
	&\leq \sum_{q \geq 0} \frac{(4\zeta \|N^{(\alpha)}\|_{0})^q}{q!} \left(\frac{q}{e\sigma}\right)^q \left\|  \mathscr{A}_{\alpha-1}\right\|_{\mathrm{op}} e^{\sigma |S|}\\
	&\leq \sum_{q \geq 0} \left(\frac{4\zeta \|N^{(\alpha)}\|_{0}}{\sigma}\right)^q \left\|  \mathscr{A}_{\alpha-1}\right\|_{\mathrm{op}} e^{\sigma |S|}\leq 2 \left\|  \mathscr{A}_{\alpha-1}\right\|_{\mathrm{op}} e^{\sigma |S|}\,,
	\end{aligned}
	\end{equation}
	where we have used Cauchy estimate, Stirling bound $q! \geq \left(\frac{q}{e}\right)^q \sqrt{2\pi q}$ and the smallness assumption \eqref{zitina} on $\zeta$. Furthermore, again by Lemma \ref{lem.ad.q.lambda} one has that $\textnormal{Ad}_{N^{(\alpha)}}^q \mathscr{A}_{\alpha-1}$ is a strongly local operator acting within $S$, and since by \eqref{fine} the series in \eqref{giornata} is convergent, also $\mathscr{A}_{\alpha}$ is a strongly local operator acting within $S$.
	\end{proof}
	
\begin{lemma}\label{besciamella.alla.cannella}
	An operator $\mathcal{A}(\theta,\varphi)=\sum_{S \in \mathcal{P}_c(\Lambda)} \sum_{(k,l) \in \mathbb{Z}^{r+m}} (\widehat{\mathcal{A}_S})_{l,k} e^{\ii k \cdot \theta} e^{\ii l \cdot \varphi}$ is of the form $\mathcal{A}(\theta,\varphi)=e^{\ii \theta \cdot N} A(\varphi) e^{-\ii \theta \cdot N}$ if and only if for any $S \in \mathcal{P}_c(\Lambda)$, $l \in \mathbb{Z}^m$, $k \in \mathbb{Z}^r$ and $\alpha = 1, \dots, r$ one has
	\begin{equation}
		[N^{(\alpha)},(\widehat{\mathcal{A}_S})_{l,k}]=k_\alpha (\widehat{\mathcal{A}_S})_{l,k} \, .
	\end{equation}
\end{lemma}	
\begin{proof}
	This lemma is proved by observing that $\mathcal{A}(\theta,\varphi)$ is of the form $\mathcal{A}(\theta,\varphi)=e^{\ii \theta \cdot N} A(\varphi) e^{-\ii \theta \cdot N}$ if and only if for any $\alpha=1,\dots,r$, one has $\frac{\ud}{\ud \theta_\alpha} (e^{-\ii \theta \cdot N} \mathcal{A}(\theta,\varphi) e^{\ii \theta \cdot N})=0$.
\end{proof}

Since we are going to solve the homological equation using the Fourier representation of the operators, we need the guarantee that strong {locality} can be reconstructed back with anti-Fourier transform. 
	\begin{lemma}\label{lem.fourier.e.bello}
	Let $\mathcal{A}:\mathbb{T}^r \times \mathbb{T}^m \rightarrow \mathcal{A}(\theta, \varphi)$, with $\mathcal{A}(\theta, \varphi) \in \mathcal{O}_{0,0,0}$ as in \eqref{A.tau.fou}. Then $\mathcal{A}(\theta, \varphi)$ is strongly {local} $\forall (\theta, \varphi) \in \mathbb{T}^{r+m}$ if and only if $(\widehat{\mathcal{A}_S})_{l,k}$ is strongly local $\forall (k, l) \in \Z^{r+m}$ $\forall S \in \mathcal{P}_c(\Lambda)$.
	\end{lemma}
	\begin{proof}
	If $\forall (l,k) \in \Z^{r+m}$ $(\mathcal{\widehat{A}_S})_{l,k}$ is strongly local, then $\forall (\theta, \varphi)$ $\mathcal{A}$ is a convergent series of strongly local operators, and $\forall S' \nsubseteq S$
	$$
	\left[\mathcal{A}_S(\theta, \varphi), N^{(\alpha)}_{S'}\right] = 	\left[\sum_{k \in \Z^{r}, l \in \Z^m} (\widehat{\mathcal{A}_S})_{l,k} e^{\ii k \cdot \theta} e^{\ii l \cdot \varphi}, N^{(\alpha)}_{S'}\right] = \sum_{k \in \Z^{r}, l \in \Z^m}\left[ (\widehat{\mathcal{A}_S})_{l,k}, N^{(\alpha)}_{S'}\right]  e^{\ii k \cdot \theta} e^{\ii l \cdot \varphi} = 0\,.
	$$
	On the contrary, suppose that $\mathcal{A}(\theta, \varphi)$ is strongly {local} $\forall (\theta, \varphi)$ and let $S' \nsubseteq S$. One has
	\begin{align*}
	\left[(\widehat{\mathcal{A}(\theta, \cdot)})_{S,l}, N^{(\alpha)}_{S'} \right] &= \frac{1}{(2\pi)^m}\int_{\mathbb{T}^m} [\mathcal{A}_S(\theta, \varphi),  N^{(\alpha)}_{S'}] e^{-\ii l \cdot \varphi}\,\ud\varphi = 0\,,
	\end{align*}
	and
	$$
	\begin{aligned}
	\left[(\widehat{\mathcal{A}_S})_{l,k}, N^{(\alpha)}_{S'} \right] &= \frac{1}{(2\pi)^r}\int_{\mathbb{T}^r} \left[(\widehat{(\mathcal{A}(\theta, \cdot))_S})_{l}, N^{(\alpha)}_{S'}\right] e^{-\ii \theta \cdot k}\,\ud \theta = 0\,,
	\end{aligned}
	$$
	which proves that $(\widehat{\mathcal{A}_S})_{l,k}$ is strongly local $\forall k, l$. 
\end{proof}

Last, let us see that  the average of the operator $\mathcal{A}(\theta, \varphi)$ defined in \eqref{decomposizione.1} over $\varphi$ and $\theta$ commutes with all the number operators.

	\begin{lemma}\label{lem.zeta.commutes}
	Given $A \in \mathcal{O}_{\kappa, \rho}^{\strong}$ for some $\kappa, \rho>0$, let
	\begin{equation}\label{zeta.esplicita}
	\langle A \rangle := \frac{1}{(2\pi)^{r+m}} \int_{\mathbb{T}^r} \int_{\mathbb{T}^m} e^{\ii\theta \cdot N} A(\varphi) e^{-\ii \theta \cdot N}\,\ud \varphi \ud \theta\,.
	\end{equation}
	Then, $\forall \alpha = 1, \dots, r$ one has $[\langle A \rangle, N^{(\alpha)}] = 0$ and 
	\begin{equation}\label{eq:estimate.z.std}
		\|\langle A \rangle\|_{\kappa} \leq \|A\|_{\kappa, \rho}\,.
	\end{equation}
	\end{lemma}
	\begin{proof}
	Given $\alpha = 1, \dots, r$, consider the function $\R \ni s_\alpha \mapsto e^{\ii s_\alpha N^{(\alpha)}} \langle A \rangle e^{-\ii s_\alpha N^{(\alpha)}}$. By reparametrizing $\theta_\alpha \mapsto \theta_\alpha+s_\alpha$ in \eqref{zeta.esplicita}, one has $e^{\ii s_\alpha N^{(\alpha)}} \langle A \rangle e^{-\ii s_\alpha N^{(\alpha)}}=\langle A \rangle$, which is the first part of the thesis.
	
Finally, to prove equation \eqref{eq:estimate.z.std}, one observes
	\begin{equation*}
		\begin{split}
			\Vert \langle A \rangle \Vert_{\kappa} & \leq \sup_{x \in \Lambda} \sum_{\substack{S \in \mathcal{P}_c(\Lambda) \\ x \in S}} \sup_{\theta \in \mathbb{R}^r} \Vert e^{\ii \theta \cdot N} (\widehat{A_S})_{0} e^{-\ii \theta \cdot N} \Vert_{\text{op}} e^{\kappa|S|} \\
			&=\sup_{x \in \Lambda} \sum_{\substack{S \in \mathcal{P}_c(\Lambda) \\ x \in S}} \Vert (\widehat{A_S})_{0} \Vert_{\text{op}} e^{\kappa|S|} \leq \Vert S \Vert_{\kappa,0}\,.
		\end{split}
	\end{equation*}
	\vspace{-25pt}
	
	\end{proof}

\subsection{Homological Equations for Small Perturbations}
In this Subsection we study the solution to the homological equation \eqref{eq:Hom.Introduttiva} in the case of small perturbations. That is, given an operator $P \in \mathcal{O}^\strong_{\kappa, \rho}$ for some $\kappa, \rho >0$ with $\| P\|_{\kappa, \rho} \ll 1$, we solve
 \begin{equation}\label{hom.original}
 	\ii [J \cdot N, G(\varphi)] + \omega \cdot \partial_\varphi G(\varphi) + P(\varphi) = Z\,,
 \end{equation}
for suitable operators $G \in \mathcal{O}^\strong_{\kappa', \rho'}$ and $Z \in \mathcal{O}^\strong_{\kappa'}$, with $\kappa' \leq \kappa$, $\rho' \leq \rho$.  Following the strategy presented in Subsection \ref{sec:Ideas}, as a first step we define
 \begin{equation}\label{g.tau.p.tau}
 \mathcal{G}(\theta, \varphi) := e^{\ii \theta \cdot N} G(\varphi) e^{-\ii \theta \cdot N}\,, \quad \mathcal{P}(\theta, \varphi) = e^{\ii \theta \cdot N} P(\varphi) e^{-\ii \theta \cdot N}\,, \quad \theta \in \mathbb{T}^r\,,
 \end{equation}
 and we solve $\forall \theta \in \mathbb{T}^r$ an equation of the following form:
 \begin{equation}\label{hom.tau}
 	\ii [J \cdot N, \mathcal{G}(\theta, \varphi)] + \omega \cdot \partial_\varphi \mathcal{G}(\theta, \varphi) + \mathcal{P}(\theta, \varphi) = \langle P\rangle\,,
 \end{equation}
 with $\langle P \rangle $ defined in \eqref{zeta.esplicita}.
 
 \begin{lemma}[$\theta$-dependent homological equation]\label{lem.hom.tau}
 	Let $\kappa, \rho, \zeta >0$ and  $\mathcal{P} \in \mathcal{O}^\strong_{\kappa, \rho, \zeta}$. There exists $\mathcal{G} \in \mathcal{O}^{\strong}_{\kappa, \rho-\delta, \zeta-\delta}$ $\forall 0< \delta \leq \min\{\rho, \zeta\}$ which solves \eqref{hom.tau} and
the following estimate holds:
 	\begin{equation}
 	\| \mathcal{G}\|_{\kappa, \rho-\delta, \zeta-\delta} \leq \frac{\tau^\tau}{e^\tau \delta^\tau \gamma} \|\mathcal{P}\|_{\kappa, \rho, \zeta}\,.
 	\end{equation}
	\end{lemma}
	\begin{proof}
	We observe that
	$$
	\ii [J \cdot N, \mathcal{G}(\theta, \varphi)] = (J \cdot \partial_{\theta}) \mathcal{G}(\theta, \varphi) = \sum_{S \in \mathcal{P}_c(\Lambda)} \sum_{l \in \Z^m} \sum_{k \in \Z^r} \ii J \cdot k (\widehat{\mathcal{G}_S})_{l,k} e^{\ii k \cdot \theta} e^{\ii \l \cdot \varphi}\,.
	$$
	Passing to Fourier coefficients, \eqref{hom.tau} is equivalent to
	$$
	\ii \left(J \cdot k + \omega \cdot l\right) (\widehat{\mathcal{G}_{S}})_{l,k} + (\widehat{\mathcal{P}_S})_{l,k} = 0 \quad \forall (l,k) \neq (0,0)\,,
	$$
	observing that $\langle P \rangle = \sum_{S \in \mathcal{P}_c(\Lambda)} (\widehat{\mathcal{P}_S})_{0,0} $.
	Then $\forall S \in \mathcal{P}_c(\Lambda)$ one sets
	\begin{equation}\label{ghat.phat}
	(\widehat{\mathcal{G}_{S}})_{l,k} = -\frac{(\widehat{\mathcal{P}_S})_{l,k}}{\ii \left(J \cdot k + \omega \cdot l\right)} \quad \forall (l,k) \neq0\,, \quad (\widehat{\mathcal{G}_{S}})_{0,0} = 0\,,
	\end{equation}
	and $\forall 0<\delta\leq \min\{\zeta, \rho\}$ has
	\begin{align*}
	\| \mathcal{G}\|_{\kappa, \rho-\delta, \zeta-\delta} &\leq \gamma^{-1} \sup_{x \in \Lambda} \sum_{S\in \mathcal{P}_c(\Lambda)} \sum_{l \in \Z^m} \sum_{k \in \Z^r} \| (\widehat{\mathcal{P}_S})_{l,k}\|_{\mathrm{op}} (|k| + |l|)^{\tau} e^{-(|k|+ |l|)\delta}e^{|k|\zeta} e^{\rho|l|} e^{|S|\kappa}\\
	&\leq \frac{\tau^\tau}{e^\tau \delta^\tau \gamma} \sup_{x \in \Lambda} \sum_{S\in \mathcal{P}_c(\Lambda)} \sum_{l \in \Z^m} \sum_{k \in \Z^r} \| (\widehat{\mathcal{P}_S})_{l,k}\|_{\mathrm{op}} e^{|k|\zeta} e^{\rho|l|} e^{|S|\kappa} =\frac{\tau^\tau}{e^\tau \delta^\tau \gamma} \|\mathcal{P}\|_{\kappa, \rho , \zeta }\,,
	\end{align*}
	where we have used the fact that $(J, \omega)$ satisfies the Diophantine condition \eqref{eq:siamo.dei.cani}. It remains to prove that $\mathcal{G}(\theta, \varphi)$ is strongly {local} $\forall (\theta, \varphi)$. This follows from the fact that $\mathcal{P}(\theta, \varphi)$ is strongly {local}, from Lemma \ref{lem.fourier.e.bello} and from the fact that, by \eqref{ghat.phat}, the Fourier coefficients of $\mathcal{G}_S(\theta, \varphi)$ are strongly local if and only if the Fourier coefficients of $\mathcal{P}_S(\theta, \varphi)$ are strongly local.
	\end{proof}
	From Lemma \ref{lem.hom.tau} we deduce the two following results:
	\begin{lemma}[Homological equation for $H_\inv$]\label{cor.trova.g}
	Let $\kappa, \rho>0$ and $P \in \mathcal{O}^\strong_{\kappa, \rho}$ and let
	\begin{equation}\label{costantina}
 \ccccc := 8 r \max_{\alpha} \|N^{(\alpha)}\|_0\,.
	\end{equation}
	There exist  $G:= G_\inv \in \mathcal{O}^\strong_{\kappa - \ccccc \delta, \rho - \delta}$ and  $Z:=Z_\inv \in \mathcal{O}^\strong_{\kappa- \ccccc \delta, \rho }$ for any $0<\delta< \min\{\rho, \ccccc^{-1} \kappa\}$ such that \eqref{hom.original} holds, with
	\begin{equation}\label{commuta.con.tutti}
	\left[Z_\inv, N^{(\alpha)}\right] = 0 \quad \forall \alpha = 1, \dots, r\,.
	\end{equation}
	Moreover, $Z_\inv=\langle P \rangle$ and $\forall \delta>0$ one has
	\begin{equation}\label{stime.g.z}
	\| G_\inv\|_{\kappa- \ccccc \delta, \rho-\delta} \leq \frac{4^r \tau^{\tau}}{e^{\tau} \gamma \delta^{\tau+r}}(1+\delta)^r \| P\|_{\kappa, \rho}\,.
	\end{equation}
	\end{lemma}
	\begin{proof}
	Let $\zeta :=\delta$ and $\sigma := 8 \max_{\alpha} \|N^{(\alpha)}\|_0 \delta$; by Item (iii) of Lemma \ref{lem:op.tau},  one has that, defining for any $\theta \in \mathbb{T}^r$ and for any $\varphi \in \mathbb{T}^m$ $\mathcal{P}(\theta, \varphi)$ as in \eqref{g.tau.p.tau}, $\mathcal{P} \in \mathcal{O}^\strong_{\kappa - r \sigma, \rho, \zeta} $, with
	\begin{equation}\label{miruna}
	\|\mathcal{P}\|_{\kappa-r \sigma, \rho, \zeta } \leq  2^{2r} \left(\frac{\zeta+1}{\zeta}\right)^r \|P\|_{\kappa, \rho}\,.
	\end{equation}
	Then by Lemma \ref{lem.hom.tau} there exists $\mathcal{G} \in \mathcal{O}^\strong_{\kappa, \rho-\delta, 0}$ solving \eqref{hom.tau} and
	\begin{equation}\label{stefana}
	\|\mathcal{G}\|_{\kappa-r \sigma, \rho-\delta, 0} \leq \frac{\tau^{\tau}}{e^{\tau} \gamma \delta^{\tau}} \| \mathcal{P}\|_{\kappa - r \sigma, \rho , \delta}\,.
	\end{equation}
	Comparing \eqref{hom.tau}, \eqref{hom.original} and \eqref{zeta.esplicita} one gets $Z_{\mathrm{inv}}=\langle P \rangle$.
	Furthermore we have that $G_\inv(\varphi):= \mathcal{G}(0, \varphi)$ solves \eqref{hom.tau} for $\theta = 0$. Moreover, again by Item (iii) of Lemma \ref{lem:op.tau}, one has
	\begin{equation}\label{giuliani}
 \|G_\inv \|_{\kappa-r \sigma, \rho - \delta} \leq \|\mathcal{G}\|_{\kappa - r \sigma, \rho - \delta, 0}\,.
	\end{equation}
	Combining \eqref{costantina}, \eqref{miruna}, \eqref{stefana} and \eqref{giuliani}, one gets
	$$
	\|G_\inv\|_{\kappa -r \sigma, \rho - \delta} = \|G_\inv\|_{\kappa -\ccccc \delta, \rho - \delta} \leq \frac{2^{2r}\tau^{\tau}}{e^{\tau} \gamma \delta^{\tau+r}}(1+\delta)^r \|P\|_{\kappa, \rho}\,,
	$$
	which gives \eqref{stime.g.z}. 
 \end{proof}

\begin{lemma}[Homological equation for $H_\obs$]\label{lem.Sailor.Moon.is.back}
	Let $\kappa,\rho>0$ and $P \in \mathcal{O}^\strong_{\kappa,\rho}$, there exist $G:= G_\obs \in \mathcal{O}^\strong_{\kappa-\ccccc \delta, \rho-\delta}$ and $Z:= Z_\obs \in \mathcal{O}^\strong_{\kappa-\ccccc\delta}$ for any $0<\delta<\min\{\ccccc^{-1}\kappa,\rho\}$ such that \eqref{hom.original} holds with $G_\obs(0) = 0$. One has
	\begin{equation}\label{bagno.maria}
	\begin{split}
	\|G_\obs\|_{\kappa -\ccccc\delta, \rho-\delta} &\leq \frac{2^{2r+1} \tau^{\tau}}{e^{\tau} \gamma \delta^{\tau+r}}(1+\delta)^r \| P\|_{\kappa, \rho}\,, \qquad	\\
	\| Z_\obs\|_{\kappa - \ccccc \delta} &\leq \frac{\mathcal{J} 2 ^{2r}(1 +\delta)^r \tau^\tau}{\gamma e^\tau \delta^{\tau + r + 1}} \|P\|_{\kappa, \rho} + \|P\|_{\kappa, \rho}\,.
	\end{split}
	\end{equation}
	In particular,
	\begin{equation}\label{cacio.e.pepe}
	Z_\obs = \llan P \rran \,, \qquad \llan P \rran := \langle P \rangle + \sum_{S \in \mathcal{P}_c(\Lambda)} \sum_{(k, l) \neq (0,0) \atop k \in \Z^{r}, l \in \Z^m} \frac{(J \cdot k) (\widehat{\mathcal{P}_S})_{l,k}}{\omega \cdot l + J \cdot k}\,.
	\end{equation}
\end{lemma}
\begin{proof}
	Since by Lemma \ref{cor.trova.g} $\mathcal{G}(0, \varphi)$ solves \eqref{hom.original} with $Z = \langle P \rangle$, then $G_\obs(\varphi) := \mathcal{G}(0, \varphi) - \mathcal{G}(0,0)$ solves \eqref{hom.original} with $Z = \langle P \rangle - \ii [J \cdot N, \mathcal{G}(0,0)]$ and by construction it satisfies $G_\obs(0) = 0$. Using Lemma \ref{besciamella.alla.cannella}, one obtains that $Z_{\mathrm{obs}}$ is given by the expression in \eqref{cacio.e.pepe}. The first estimate in \eqref{bagno.maria} follows from the estimate \eqref{stime.g.z} on $\mathcal{G}(0, \varphi)$, while to obtain the second estimate in \eqref{bagno.maria} we observe that $\|\llan P \rran\|_{\kappa - \ccccc \delta} \leq \|\langle P \rangle\|_{\kappa - \ccccc\delta} + \| [J \cdot N, \mathcal{G}(0,0)]\|_{\kappa - \ccccc\delta}$, with
	$$
	\| [J \cdot N,\ \mathcal{G}(0,0)]\|_{\kappa - \ccccc\delta} \leq \sup_{x \in \Lambda} \sum_{S \ni x \atop S \in \mathcal{P}_c(\Lambda)} \sum_{(k, l) \neq (0,0)} \frac{|J \cdot k|}{|\omega \cdot l + J \cdot k|} \| (\widehat{\mathcal{P}_S})_{l,k} \|_{\mathrm{op}} e^{(\kappa - \ccccc\delta )|S|}\,.
	$$
	Then the second estimate in \eqref{bagno.maria} follows using Diophantine condition \eqref{eq:siamo.dei.cani}, Cauchy inequality, equation \eqref{eq:norma.A.norma.A.tau}, and estimate \eqref{eq:estimate.z.std} on $\|\langle P \rangle\|_{\kappa - \ccccc\delta}$.
	
\end{proof}
\subsection{Homological Equation for Fast-Forcing Perturbations} In this Subsection, we study the solution to the homological equation \eqref{eq:Hom.Introduttiva} in the fast-forcing case, that is, given $\lambda \gg 1$ and $P \in \mathcal{O}^\strong_{\kappa, \rho}$ for some $\kappa, \rho>0$, we want to solve the equation
 \begin{equation}\label{hom.original.ff}
 \ii [J \cdot N, G(\varphi)] + \lambda \omega \cdot \partial_\varphi G(\varphi) + P(\varphi) = Z\,,
 \end{equation}
 for some operators $G \in \mathcal{O}^\strong_{\kappa', \rho'}$ and $Z \in \mathcal{O}^\strong_{\kappa'}$ with $\kappa' \leq \kappa$, $\rho' \leq \rho$. 
 {In particular, we set $Z := \langle P \rangle$, with $\langle P \rangle$ defined in \eqref{zeta.esplicita} and,} after having defined $\mathcal{G}(\theta,\varphi)$ and $\mathcal{P}(\theta,\varphi)$ as in \eqref{g.tau.p.tau}, we first solve $\forall \theta \in \mathbb{T}^r$ the following equation:
 \begin{equation}\label{hom.tau.ff}
 \ii [J \cdot N, \mathcal{G}(\theta, \varphi)] + \lambda\omega \cdot \partial_\varphi \mathcal{G}(\theta, \varphi) + \mathcal{P}(\theta, \varphi) = \langle P \rangle\,.
 \end{equation}
 The strategy {we use} to solve this equation (up to higher order terms in $\lambda^{-1}$) makes use of the splitting of Fourier modes into $\uv$/$\ir$ regions. That is, given
 \begin{equation}\label{m1.m2}
 K := \frac{\gamma_\omega}{2|J|}\lambda^{\frac 12}\,, \quad L := \lambda^{\frac{1}{2\tau_\omega}}\,,
 \end{equation}
 we define the following sets:
 \begin{gather}
	 \mathcal{K}_{\uv} := \left\{ k \in \Z^r\ \left|\ |k| > K\right.\right\}\,, \quad
	 \mathcal{K}_{\ir} := \left\{ k \in \Z^r\ \left|\ |k| \leq K\right.\right\}\,,\\
	 \mathcal{L}_{\uv} := \left\{ l \in \Z^m\ \left|\ |l| > L\right.\right\}\,, \quad
	 \mathcal{L}_{\ir} := \left\{ l \in \Z^m\ \left|\ |l| \leq L\right.\right\}\,,
 \end{gather}
 and for $\mathtt{i}, \mathtt{j} \in \{\uv, \ir\}$,
 \begin{equation}\label{p.a.b}
 	\mathcal{P}^{\mathtt{i}, \mathtt{j}} (\theta, \varphi) := \sum_{S \in \mathcal{P}_c(\Lambda) } \sum_{k \in \mathcal{K}_{\mathtt{i}}} \sum_{l \in \mathcal{L}_{\mathtt{j}}} (\widehat{\mathcal{P}_S})_{l, k} e^{\ii k \cdot \theta} e^{\ii l \cdot \varphi}\,.
 \end{equation}
 Consequently, we also define
 \begin{equation}\label{p.normale.a.b}
 {P}^{\mathtt{i}, \mathtt{j}} (\varphi) := \mathcal{P}^{\mathtt{i}, \mathtt{j}} (0, \varphi)\,.
 \end{equation}
 Also, for simplicity, we define
 \begin{equation}\label{p.uv.p.ir}
 \begin{gathered}
 \mathcal{P}^{\uv} := \mathcal{P}^{\uv, \ir}  + \mathcal{P}^{\uv, \uv}  + \mathcal{P}^{\ir, \uv}\,, \quad
 \mathcal{P}^{\ir} := \mathcal{P}^{\ir , \ir}\,,\\
 P^{\uv} := \mathcal{P}^{\uv}(0 ,\cdot)\,, \quad P^{\ir} := \mathcal{P}^{\ir}(0 ,\cdot)\,. 
 \end{gathered}
 \end{equation}
 First, we prove that the $\uv$ terms can be treated as remainders in the solution of the homological equation:
 \begin{lemma}\label{lem.uv.small}
  Let $\kappa, \rho, \zeta, \sigma, \eta >0$ and  $\mathcal{P} \in \mathcal{O}^\strong_{\kappa, \rho + \sigma, \zeta + \eta}$. One has
  \begin{gather}
  \label{uv.small}
  \| \mathcal{P}^{\uv}\|_{\kappa, \rho, \zeta} \leq \frac{2}{e} \max\left\{\frac{1}{ \eta K}, \frac{1}{\sigma L}\right\} \| \mathcal{P}\|_{\kappa, \rho + \sigma, \zeta + \eta}\,,\\
  \label{uv.ir.trivial}
  \| \mathcal{P}^{\uv}\|_{\kappa, \rho, \zeta} \leq \| \mathcal{P}\|_{\kappa, \rho, \zeta}\,, \quad \| \mathcal{P}^{\ir}\|_{\kappa, \rho, \zeta} \leq \| \mathcal{P}\|_{\kappa, \rho, \zeta}\,.
  \end{gather}
 \end{lemma}
 \begin{proof}
	Estimates \eqref{uv.ir.trivial} immediately follow from the definition of $\mathcal{P}^{\uv}, \mathcal{P}^\ir$. Concerning \eqref{uv.small}, one has
 	\begin{align*}
 	\left\| \mathcal{P}^{\uv}\right\|_{\kappa, \rho, \zeta} &\leq \left\| \mathcal{P}^{\uv, \ir} + \mathcal{P}^{\uv, \uv}\right\|_{\kappa, \rho, \zeta} + \left\|  \mathcal{P}^{\ir, \uv} \right\|_{\kappa, \rho, \zeta}\,.
 	\end{align*}
 	Using Cauchy estimate, the first summand is estimated by
 	\begin{align*}
	\left\| \mathcal{P}^{\uv, \ir} + \mathcal{P}^{\uv, \uv}\right\|_{\kappa, \rho, \zeta} &\leq \frac{1}{K} \sup_{x \in \Lambda} \sum_{S \ni x} \sum_{|k| > K} \sum_{l \in \Z^m} \Big\| (\widehat{\mathcal{P_S}})_{l,k} \Big\|_{\mathrm{op}} |k|  e^{\zeta|k|} e^{\rho|l|} e^{\kappa |S|}\\
	& \leq \frac{1}{e \eta K}  \| \mathcal{P}\|_{\kappa, \rho, \zeta + \eta} \,,
 	\end{align*}
 	and an analogous estimate entails
 	$
 	\left\| \mathcal{P}^{\ir, \uv} \right\|_{\kappa, \rho, \zeta} \leq \frac{1}{e \sigma L} \| \mathcal{P}\|_{\kappa, \rho + \sigma, \zeta}
 	$.
 \end{proof}
 \begin{corollary}\label{cor.p.small}
 Let $\kappa, \rho>0$ and $P \in \mathcal{O}^\strong_{\kappa, \rho}$. Defining $\ccccc$ as in \eqref{costantina}, $\forall 0<\delta < \min\{\rho, {\tt c}^{-1} \kappa\}$  one has $ P^{\uv} \in \mathcal{O}^\strong_{\kappa - \ccccc\delta, \rho - \delta}$, with
 \begin{gather}
 \label{p.uv.small}
 	\| P^{\uv}\|_{\kappa - \ccccc\delta, \rho - \delta} \leq \frac{2^{2r+1}}{e\delta^{r+1}} (1+\delta)^r\max\left\{\frac{1}{K}, \frac{1}{L}\right\} \|P\|_{\kappa, \rho}\,.
 \end{gather}
 \end{corollary}
 \begin{proof}
 	Define $\zeta:= \delta$ and $\sigma := 8 \max_\alpha \|N^{(\alpha)}\|_{0}$; then estimate \eqref{zitina} is fulfilled and by Item (iii) of Lemma \ref{lem:op.tau} one has
 	$$
 	\|\mathcal{P}\|_{\kappa - r \sigma, \rho, \zeta } = \| \mathcal{P}\|_{\kappa - \ccccc \delta, \rho, \delta} \leq 2^{2r}\left(\frac{1+\delta}{\delta}\right)^r \|P\|_{\kappa, \rho}\,.
 	$$
 	Then one applies Lemma \ref{lem.uv.small} with $\zeta =\eta = \sigma = \delta$ to obtain
 	$$
 	\|\mathcal{P}^{\uv}\|_{\kappa - \ccccc \delta, \rho - \delta, 0} \leq \frac{2}{e \delta}  \|\mathcal{P}\|_{\kappa - \ccccc \delta, \rho, \delta}\,,
 	$$
 	and \eqref{p.uv.small} follow again by Item (iii) of Lemma \ref{lem:op.tau}.
 \end{proof}
 
Given $\kappa,\zeta,\rho > 0$, $\mathcal{P} \in \mathcal{O}^\strong_{\kappa,\rho,\zeta}$ we define
\begin{equation}
	\langle \mathcal{P} \rangle_{\mathbb{T}^m}(\theta) \;:=\; \frac{1}{(2 \pi)^m} \int_{\mathbb{T}^m} \mathcal{P}(\theta,\varphi) \, \ud \varphi \, .
\end{equation} 
 We are now ready to solve an equation of the form of \eqref{hom.tau.ff}, with $\mathcal{P} = \mathcal{P}^\ir$:
 \begin{lemma}[$\theta$-dependent homological equation]\label{lem.Senza.Nome}
 	Let $\kappa,\rho,\zeta>0$ and $\mathcal{P} \in \mathcal{O}^\strong_{\kappa,\rho,\zeta}$, there exists $\mathcal{G} \in \mathcal{O}^{\strong}_{\kappa,\rho,\zeta-\delta}$ for any $0< \delta \leq \zeta$ which solves
 	\begin{equation}\label{eq:HomologicalWithLambda}
 		\ii [J\cdot N, \mathcal{G}]+\mathcal{P}^\ir + \lambda \omega \cdot \partial_\varphi \mathcal{G} = \langle P \rangle\,,
 	\end{equation}
 	with $\langle P \rangle$ as in \eqref{zeta.esplicita}
 	and
 	\begin{equation}
 		\Vert \mathcal{G} \Vert_{\kappa, \rho,\zeta-\delta} \; \leq \; \frac{4}{\gamma_\omega} \frac{1}{\lambda^{1/2}} \Vert \mathcal{P}-\langle \mathcal{P} \rangle_{\mathbb{T}^m} \Vert_{\kappa,\rho,\zeta} + \frac{2 \tau_J^{\tau_J}}{\delta^{\tau_J} e^{\tau_J} \gamma_J} \Vert \langle \mathcal{P} \rangle_{\mathbb{T}^m} \Vert_{\kappa,0,\zeta} \, .
 	\end{equation}
 \end{lemma}

\begin{proof}
 	We observe that
 $$
 \ii [J \cdot N, \mathcal{G}(\theta, \varphi)] = (J \cdot \partial_{\theta}) \mathcal{G}(\theta, \varphi) = \sum_{S \in \mathcal{P}_c(\Lambda)} \sum_{l \in \Z^m} \sum_{k \in \Z^r} \ii J \cdot k \, (\widehat{\mathcal{G}_S})_{l,k} e^{\ii k \cdot \theta} e^{\ii \l \cdot \varphi}\,.
 $$
The first equation in \eqref{eq:HomologicalWithLambda} is equivalent to
 \begin{equation}\label{derivatives.only}
 (J \cdot \partial_\theta) \mathcal{G} + \lambda \omega \cdot \partial_\varphi \mathcal{G} = -\mathcal{P}^\ir + \langle P \rangle\,,
 \end{equation}
 and, passing to Fourier coefficients,
 $$
 \ii \left(J \cdot k + \omega \cdot l\right) (\widehat{\mathcal{G}_{S}})_{l,k} + (\widehat{\mathcal{P}_S})_{l,k} = 0 \quad \forall (l,k) \neq (0,0)\,, \quad (l,k) \in \mathcal{K}_{\ir} \times \mathcal{L}_{\ir}\,.
 $$
 Then $\forall S \in \mathcal{P}_c(\Lambda)$ one sets
 \begin{equation}\label{ghat.phat.lambda}
 (\widehat{\mathcal{G}_{S}})_{l,k} =
 \begin{cases}
 -\dfrac{(\widehat{\mathcal{P}_S})_{l,k}}{\ii \left(J \cdot k + \omega \cdot l\right)} &\forall (l,k) \in \mathcal{K}_\ir \times \mathcal{L}_\ir \setminus\{(0,0)\}\,,\\
  0 & \textrm{otherwise}\,.
 \end{cases}
 \end{equation}
 We now provide an estimate on the small divisors appearing in \eqref{ghat.phat.lambda}. We estimate them differently according if $l = 0$ or not. Due to the choice of $K, L$ as in \eqref{m1.m2} and to the fact that $\omega$ is Diophantine \eqref{eq:siamo.dei.cani}, if $l \neq 0$ one has
 \begin{equation}\label{delicato.equilibrio}
 |J \cdot k + \lambda \omega \cdot l| \geq \lambda |\omega \cdot l| - |J| |k| \geq \lambda \frac{\gamma_\omega}{ |l|^{\tau_\omega}} - |J||k| \geq \lambda \frac{\gamma_\omega}{L^{\tau_\omega}} - |J| K \geq \frac{\gamma_\omega}{2} \lambda^{\frac{1}{2}}\,.
 \end{equation}
 On the other hand, if $l=0$ due to the fact that $J$ is Diophantine \eqref{eq:siamo.dei.cani} we have
 $
 |J \cdot k | \geq \frac{\gamma_J}{|k|^{\tau_J}}.
 $
 Then $\forall 0<\delta \leq \zeta$ one has $\| \mathcal{G}\|_{\kappa, \rho, \zeta-\delta} \leq T_1 + T_2$, with
  \begin{align*}
 T_1 &:=  \sup_{x \in \Lambda} \sum_{S\in \mathcal{P}_c(\Lambda)} \sum_{0 \neq k \in \mathcal{K}_\ir} \| (\widehat{\mathcal{G}_S})_{k,0}\|_{\mathrm{op}}e^{(\zeta-\delta)|k|} e^{|S|\kappa}\\
 & \stackrel{\eqref{eq:siamo.dei.cani.ff.1}}{\leq} \sup_{x \in \Lambda} \sum_{S\in \mathcal{P}_c(\Lambda)} \sum_{0 \neq k \in \mathcal{K}_\ir} \Vert (\widehat{\mathcal{P}_S})_{l,k}\|_{\mathrm{op}} \frac{|k|^{\tau_J}}{\gamma_J} e^{-\delta|k|}e^{\zeta |k|}  e^{\kappa |S|}\\
 &\leq \sup_{x \in \Lambda} \sum_{S\in \mathcal{P}_c(\Lambda)}  \sum_{0 \neq k \in \mathcal{K}_\ir} \Vert (\widehat{\mathcal{P}_S})_{l,k}\|_{\mathrm{op}} \frac{\tau_J^{\tau_J}}{e^{\tau_J} \delta^{\tau_J}\gamma_J} e^{\zeta |k|} e^{\kappa |S|} \leq \frac{\tau_J^{\tau_J}}{e^{\tau_J } \delta^{\tau_J} \gamma_J} \|\langle \mathcal{P}\rangle_{\mathbb{T}^m}\|_{\kappa, \rho , \zeta }
 \end{align*}
 and
 \begin{align*}
 T_2 &:=  \sup_{x \in \Lambda} \sum_{S\in \mathcal{P}_c(\Lambda)} \sum_{0 \neq l \in \mathcal{L}_\ir} \sum_{k \in \mathcal{K}_\ir} \| (\widehat{\mathcal{G}_S})_{l,k}\|_{\mathrm{op}} e^{\rho|l|}e^{(\zeta-\delta)|k|}e^{|S|\kappa}\\
 & \stackrel{\eqref{delicato.equilibrio}}{\leq} \sup_{x \in \Lambda} \sum_{S\in \mathcal{P}_c(\Lambda)} \sum_{0 \neq l \in \mathcal{L}_\ir} \sum_{k \in \mathcal{K}_\ir} \| (\widehat{\mathcal{P}_S})_{l,k}\|_{\mathrm{op}} \frac{2}{\gamma_{\omega} \lambda^{\frac 1 2}} e^{\rho|l|}e^{\zeta |k|} e^{\kappa |S|} \leq  \frac{2}{\gamma_{\omega} \lambda^{\frac 1 2}}  \|\mathcal{P} - \langle\mathcal{P}\rangle_{\mathbb{T}^m} \|_{\kappa, \rho , \zeta}\,,
 \end{align*}
 where in the last step we have used the fact that
 $$
	\mathcal{P} - \langle\mathcal{P}\rangle_{\mathbb{T}^m} = \sum_{S \in \mathcal{P}_c(\Lambda)} \sum_{k \in \Z^r, l \in \Z^m} (\widehat{\mathcal{P}_S})_{l,k}- \sum_{S\in \mathcal{P}_c(\Lambda)} \sum_{k \in\Z^r} (\widehat{\mathcal{P}_S})_{k,0} = \sum_{S \in \mathcal{P}_c(\Lambda)} \sum_{k \in \Z^r, 0 \neq l \in \Z^m} (\widehat{\mathcal{P}_S})_{l,k}\,.
 $$
 It remains to prove that $\mathcal{G}(\theta, \varphi)$ is strongly {local} $\forall (\theta, \varphi)$. This follows from the fact that $\mathcal{P}(\theta, \varphi)$ is strongly {local}, from Lemma \ref{lem.fourier.e.bello} and from the fact that, by \eqref{ghat.phat}, the Fourier coefficients of $\mathcal{G}_S(\theta, \varphi)$ are strongly local if and only if the Fourier coefficients of $\mathcal{P}_S(\theta, \varphi)$ are strongly local.
 \end{proof}

 \begin{corollary}\label{cor.trova.g.ff}
 	Let $\kappa, \rho>0$ and $P \in \mathcal{O}^\strong_{\kappa, \rho}.$ Let $\ccccc$ be as in \eqref{costantina}, 
 	there exists  $G \in \mathcal{O}^\strong_{\kappa - \ccccc \delta, \rho}$  for any $0<\delta< \min\{\rho, \ccccc^{-1} \kappa\}$ such that 
\begin{equation}\label{hom.tau.lambda}
 \ii [J \cdot N, G] + \lambda\omega \cdot \partial_\varphi G + P^{\ir} = \langle P \rangle\,.
\end{equation} 	
 Moreover, $\forall \delta>0$ one has
 	\begin{equation}\label{stime.g.z.lambda}
 	\| G\|_{\kappa- \ccccc \delta, \rho} \leq 2^{2r} \left(\frac{\delta+1}{\delta}\right)^r \left(\frac{4}{\gamma_\omega \lambda^{1/2}} \Vert P - \langle P \rangle_{\mathbb{T}^m} \Vert_{\kappa,\rho} + \frac{2 \tau_J^{\tau_J}}{\gamma_J e^{\tau_J} \delta^{\tau_J}} \Vert \langle P \rangle_{\mathbb{T}^m} \Vert_{\kappa,0} \right)\,.
 		\end{equation}
 \end{corollary}
 \begin{proof}
 	Let $\zeta :=\delta$ and $\sigma :=  8 \max_{\alpha} \|N^{\alpha}\|_0 \delta$; by Item (iii) of Lemma \ref{lem:op.tau},  one has that, defining $\forall(\theta, \varphi)$ $\mathcal{P}(\theta, \varphi)$ as in \eqref{g.tau.p.tau}, $\mathcal{P} \in \mathcal{O}^\strong_{\kappa - r \sigma, \rho, \zeta} $, with
 	\begin{equation}\label{miruna.tullio}
 	\|\mathcal{P}\|_{\kappa-r \sigma, \rho, \zeta } \leq 2^{2r}\left(\frac{\zeta+1}{\zeta}\right)^r \|P\|_{\kappa, \rho}\,.
 	\end{equation}
 	Then, by Lemma \ref{lem.Senza.Nome} there exists $\mathcal{G} \in \mathcal{O}^\strong_{\kappa-r \sigma, \rho, 0}$ solving \eqref{eq:HomologicalWithLambda} with $\langle P \rangle$ as in \eqref{zeta.esplicita} and
 	\begin{equation}\label{stefana.tullio}
 	\|\mathcal{G}\|_{\kappa-r \sigma, \rho, 0} \leq \frac{4}{\gamma_\omega \lambda^{1/2}} \Vert \mathcal{P} - \langle \mathcal{P}\rangle_{\mathbb{T}^m} \Vert_{\kappa,\rho,\delta} + \frac{2\tau_J^{\tau_J}}{\gamma_J e^{\tau_J} \delta^{\tau_J}} \Vert \langle \mathcal{P} \rangle_{\mathbb{T}^m} \Vert_{\kappa,0,\delta}\,.
 	\end{equation}
 	Furthermore, we have that $G(\varphi):= \mathcal{G}(0, \varphi)$ solves \eqref{eq:HomologicalWithLambda} for $\theta = 0$, namely \eqref{hom.tau.lambda}. Moreover, again by Item (iii) of Lemma \ref{lem:op.tau}, one has
 	\begin{equation}\label{giuliani.tullio}
 	\|G \|_{\kappa-r \sigma, \rho } \leq \|\mathcal{G}\|_{\kappa - r \sigma, \rho, 0}\,.
 	\end{equation}
 	Combining \eqref{miruna.tullio}, \eqref{stefana.tullio} and \eqref{giuliani.tullio}, one gets
 	$$
 	\|{G}\|_{\kappa -r \sigma, \rho } \leq 2^{2r}\left(\frac{\delta+1}{\delta}\right)^r \left(\frac{4}{\gamma_\omega \lambda^{1/2}} \Vert {P} - \langle {P}\rangle_{\mathbb{T}^m} \Vert_{\kappa,\rho} + \frac{2\tau_J^{\tau_J}}{\gamma_J e^{\tau_J} \delta^{\tau_J}} \Vert \langle {P} \rangle_{\mathbb{T}^m} \Vert_{\kappa,0}\right)\,,
 	$$
 	which gives \eqref{stime.g.z}. 
 \end{proof}

\section{Proof of Normal Form Results}\label{sec:norma.normale}

In this Section we use the estimates of the previous sections to prove the rigorous normal form results, as stated in Propositions \ref{teo:main} and \ref{teo:main.ff}.

\subsection{Small Perturbations} 
Item (a) of Proposition \ref{teo:main} immediately follows from the following lemma.

\begin{lemma}\label{nf.thm}
	Under the assumptions of Proposition \ref{teo:main}, define
	\begin{equation}\label{ho.ho.ho}
	n_* := \left\lfloor \varepsilon^{-2 \mathtt{b}}\right\rfloor\,, \quad  \kappa_0 := \min\{\kappa, \rho\}\,, \quad \rho_0 := 2(\ccccc + 1)^{-1} \kappa_0\,,
	\end{equation}
	with $\ccccc$ as in \eqref{costantina} and $\mathtt{b}$ as in \eqref{parametri.vitali},
	and define $\forall n = 1, \dots, n_*$
	\begin{equation}
	\kappa_n := {\kappa_0}\sqrt{1 - b n}\,, \quad \rho_n := {\rho_0}\sqrt{1 - b n}\,,\quad b := \frac{1}{2 n_*}\,,
	\end{equation}
 	and $\forall n = 0, \dots, n_*$
 	\begin{equation}\label{granchietti}
 	\varepsilon_n := \left(\frac{1}{e}\right)^n \varepsilon \| V\|_{\kappa_0, \rho_0}\,.
 	\end{equation}
	Then $\forall n = 0, \dots, n_*$ there exist a time quasi-periodic family of unitary maps $Y^{(n)}_{\inv}(\omega t)$, a self-adjoint operator $Z^{(n)}_{\inv}$ and a time quasi-periodic family of self-adjoint operators $V^{(n)}_{\inv}(\omega t)$ such that, defining
	\begin{equation}
	H^{(n)}_{\inv}(\omega t):= J \cdot N + Z^{(n)}_{\inv} + V^{(n)}_{\inv}( \omega t) \,,
	\end{equation}
	one has
	\begin{equation}
	U_{H^{(n)}_{\inv}}(t) = Y^{(n)}_{\inv}( \omega t) U_{H}(t) (Y^{(n)}_{\inv})^*(0)\,,
	\end{equation}
	and the following are satisfied:
	\begin{enumerate}
		\item $Z^{(n)}_{\inv} \in \mathcal{O}^\strong_{\kappa_n}$ is time independent, with $Z^{(0)}=0$, and for $n>0$
		\begin{equation}\label{z.small}
		[Z^{(n)}_{\inv}, N^{(\alpha)}] = 0 \quad \forall \alpha = 1, \dots, r\,, \quad \|Z^{(n)}_{\inv}\|_{\kappa_n} \leq \sum_{n'=0}^{{n-1}} \varepsilon_{{n'}}\,;
		\end{equation}
		\item $V^{(n)}_{\inv} \in \mathcal{O}^\strong_{\kappa_n, \rho_n}$, and 
		\begin{equation}\label{v.small}
		\| V^{(n)}_{\inv}\|_{\kappa_n, \rho_n} \leq \varepsilon_n\,;
		\end{equation}
		\item The maps $Y^{(n)}_{\inv}(\omega t)$ are close to the identity, in the sense that $Y_\inv^{(0)}=\mathbbm{1}$, and for $n>0$
		\begin{equation}\label{y.diff}
		\| Y^{(n)}_{\inv} P (Y^{(n)}_{\inv})^*- P\|_{\kappa_n, \rho_n} \leq 2 \varepsilon^{\frac 12} \sum_{n'=0}^{{n-1}} \left(\frac 1 e\right)^{{n'}}\|V\|_{\kappa, \rho} \|P\|_{\kappa, \rho} \quad \forall P \in \mathcal{O}_{\kappa, \rho}\,.
		\end{equation}
	\end{enumerate}
\end{lemma}

To prove the lemma, we recall the following well-known result:
\begin{lemma}[Lemma 3.1 of \cite{BGMR_growth}]\label{lem.pushfwd}
	Given $G(t)$ and $H(t)$ two families of self-adjoint operators with smooth dependence on time, define $\widetilde{U}(t) := e^{-\ii G(t)} U_{H}(t) e^{\ii G(0)}$. Then $\forall t$ $\widetilde{U}(t) = U_{H^\prime}(t)$, with
	\begin{equation}\label{solita.roba.ma.la.copio}
	H^\prime(t) =e^{-\ii G(t)} H(t) e^{\ii G(t)} + \int_0^1 e^{-\ii G(t)s} \partial_t{G}(t) e^{\ii G(t)s}\,ds\,.
	\end{equation}
\end{lemma}
\begin{proof}[Proof of Lemma \ref{nf.thm}]
	If $n=0$, the result holds true with $Y^{(0)}_\inv = \mathbbm{1}$, $Z^{(0)}_\inv = 0$, and $V^{(0)}_\inv = \varepsilon V$. Let us assume the thesis has been proven for $n' = 0, \dots, n$; we are now going to prove that it holds also for $n' = n+1$. We look for a change of variables of the form
	$$
	\widetilde{U}^{(n)}_{\inv}(t) = e^{-\ii G^{(n)}_{\inv}(\omega t)} U_{H^{(n)}_{\inv}}(t) e^{\ii G^{(n)}_{\inv}(0)}\,,
	$$
	where $G^{(n)}_{\inv}$ solves
	\begin{equation}\label{hom.again}
	\ii \left[J \cdot N, G^{(n)}_{\inv}(\varphi)\right]  + \omega \cdot \partial_\varphi{G}^{(n)}_{\inv}(\varphi) + V^{(n)}_{\inv}( \varphi) - \langle V^{(n)}_{\inv} \rangle = 0\,,
	\end{equation}
	and $\langle \cdot \rangle$ is defined in \eqref{zeta.esplicita}. 	Recalling that $H^{(n)}_{\inv} = J \cdot N + Z^{(n)}_{\inv} + V^{(n)}_{\inv}$, by Lemma \ref{lem.pushfwd} one has $\widetilde{U}^{(n)}_{\inv}(t) = U_{H^{(n+1)}_{\inv}}(t)$, with $H^{(n+1)}_{\inv}(\varphi)$ given by \eqref{solita.roba.ma.la.copio}: 
	\begin{subequations}
		\begin{align}
		\label{is.the.boss}
		H^{(n+1)}_{\inv}(\varphi) &= J \cdot N + Z^{(n)}_{\inv} + \langle V^{(n)}_{\inv} \rangle\\
		\label{vanishes}
		&+ \ii \left[J \cdot N, G^{(n)}_{\inv}(\varphi)\right]  + \omega\cdot \partial_\varphi {G}^{(n)}_{\inv}(\varphi) + V^{(n)}_{\inv}( \varphi) - \langle V^{(n)}_{\inv} \rangle\\
		\label{rn}
		& + e^{-\ii G^{(n)}_{\inv}(\varphi)} (J \cdot N) e^{\ii G^{(n)}_{\inv}(\varphi)} - J \cdot N - \ii [J \cdot N, G^{(n)}_{\inv}(\varphi)] \\
		\label{rv}
		& +  e^{-\ii G^{(n)}_{\inv}(\varphi)} V^{(n)}_{\inv}( \varphi) e^{\ii G^{(n)}_{\inv}(\varphi)} - V^{(n)}_{\inv}( \varphi)\\
		\label{rz}
		& +  e^{-\ii G^{(n)}_{\inv}( \varphi)} Z^{(n)}_{\inv} e^{\ii G^{(n)}_{\inv}(\varphi)} - Z^{(n)}_{\inv}\\
		\label{rg}
		& + \int_0^1 e^{-\ii G^{(n)}_{\inv}(\varphi)s} \omega \cdot \partial_\varphi G^{(n)}_{\inv}( \varphi) e^{\ii G^{(n)}_{\inv}(\varphi)s }\,\ud s - \omega \cdot \partial_\varphi {G}^{(n)}_{\inv}( \varphi)\,.
		\end{align}
	\end{subequations}
	We set
	\begin{equation}\label{eq:AiutoIlCapoMiRapisce}
	Z^{(n+1)}_{\inv} := Z^{(n)}_{\inv} + \langle V^{(n)}_{\inv} \rangle\,, \quad V^{(n+1)}_{\inv} := \eqref{rn} + \eqref{rv} + \eqref{rz} + \eqref{rg}\,.
	\end{equation}	
	We observe that, due to our choice of the parameters $\{\kappa_n\}_{n=1}^{n_*}$, $\{\rho_n\}_{n=1}^{n_*}$, one has $\kappa_n = {\frac{(\ccccc + 1)}{2}} \rho_n$ $\forall n$, thus defining $\delta_n:=\frac{\rho_n-\rho_{n+1}}{2}$, also
	$$
	\kappa_{n} - \kappa_{n+1} = {\frac{( \ccccc + 1)}{2}} (\rho_n -\rho_{n+1}) = ( \ccccc + 1) \delta_n
	$$ 
	and $\kappa_{n} - \ccccc \delta_n = \kappa_{n+1} + \delta_n$.
	Using Lemma \ref{cor.trova.g}, we find $G^{(n)}_{\inv}$  such that \eqref{hom.again} holds and, using the inductive hypothesis on $V^{(n)}_{\inv}$,
	\begin{equation}\label{one.more.time}
	\begin{split}
	\| G^{(n)}_{\inv}\|_{\kappa_{n+1} + \delta_{n}, \rho_{n+1}+\delta_n} &\leq \frac{2^{2r}\tau^{\tau}(1+\delta_n)^r}{e^{\tau} \gamma \delta_{n}^{\tau+r}} \| V^{(n)}_{\inv}\|_{\kappa_n, \rho_n} \leq \frac{2^{2r}\tau^{\tau}(\delta_n+1)^r}{e^{\tau} \gamma \delta_{n}^{\tau+r}} \varepsilon_n\,, \\
	\|\langle V^{(n)}_{\inv} \rangle\|_{\kappa_{n+1} + \delta_n} &\leq  \|V^{(n)}_{\inv}\|_{\kappa_{n}, \rho_n} \leq \varepsilon_n\,.
	\end{split}
\end{equation}
In particular, the second estimate in \eqref{one.more.time} implies \eqref{z.small} with $n$ replaced by $n+1$. Furthermore, since by Lemma \ref{lem.zeta.commutes} $\langle V^{(n+1)}_{\inv}\rangle$ commutes with all the $N^{(\alpha)}$'s, also the first condition in \eqref{z.small} is still satisfied also at $n+1$.
	We now estimate $V^{(n+1)}_{\inv}$. By Lemma \ref{cor.lignano.pineta}, if
\begin{equation}\label{prove.me}
	\frac{4 e^{-\kappa_{n+1}}}{\delta_{n}}\| G^{(n)}_{\inv}\|_{\kappa_{n+1} + \delta_n, \rho_{n+1}+\delta_n} < \frac{1}{2}\,,	
	\end{equation}
we have 
\begin{equation} \label{grappino}
	\begin{split}
	\Vert \, \text{\eqref{rv}} \, \Vert_{\kappa_{n+1},\rho_{n+1}} &\leq \frac{8 e^{-\kappa_{n+1}}}{\delta_n} \Vert G^{(n)}_{\inv} \Vert_{\kappa_{n+1}+\delta_n,\rho_{n+1}+\delta_n} \Vert V^{(n)}_{\inv} \Vert_{\kappa_{n+1}+\delta_n,\rho_{n+1}+\delta_n} \, , \\ 
	\Vert \, \text{\eqref{rz}} \, \Vert_{\kappa_{n+1},\rho_{n+1}} &\leq \frac{8 e^{-\kappa_{n+1}}}{\delta_n} \Vert G^{(n)}_{\inv} \Vert_{\kappa_{n+1}+\delta_n,\rho_{n+1}+\delta_n} \Vert Z^{(n)}_{\inv} \Vert_{\kappa_{n+1}+\delta_n} \, ,
	\\ 
	\Vert \, \text{\eqref{rg}} \, \Vert_{\kappa_{n+1},\rho_{n+1}} &\leq \frac{8 |\omega|e^{-\kappa_{n+1}}}{\delta_n^2 } \Vert G^{(n)}_{\inv} \Vert_{\kappa_{n+1}+\delta_n,\rho_{n+1}+\delta_n}^2  \, , \\ 
	\Vert \, \text{\eqref{rn}}\, \Vert_{\kappa_{n+1}, \rho_{n+1}} &\leq \frac{8 e^{-\kappa_{n+1}}}{\delta_n} \Vert G^{(n)}_{\inv} \Vert_{\kappa_{n+1}+\delta_n,\rho_{n+1}+\delta_n} \Vert \omega\cdot \partial_\varphi G^{(n)}_{\inv}+V^{(n)}_{\inv}- \langle V^{(n)}_{\inv} \rangle \Vert_{\kappa_{n+1}+\delta_n,\rho_{n+1}} \, ,
	\end{split}
\end{equation}
where to obtain the last inequality we have used the fact that $-\ii [J \cdot N, G^{(n)}_{\inv}(\varphi)] = \omega \cdot \partial_\varphi G^{(n)}_{\inv} + V^{(n)}_{\inv} - \langle V^{(n)}_{\inv} \rangle$, since $G^{(n)}_{\inv}$ solves \eqref{hom.again}.
Then the inductive estimate on $V^{(n)}_n$ follows if all the terms in \eqref{grappino} can be bounded by $\frac{\varepsilon_{n+1}}{4}$. Observing that
\begin{equation}\label{delta.below}
	\begin{split}
	\delta_{n} =& \frac{\rho_{n} - \rho_{n+1}}{2} = \frac{1}{2}\rho_0 \sqrt{ 1 - b n}\left(1 - \sqrt{1-\frac{b}{1-bn}}\right) \\
	\geq&  \frac{\rho_0}{2\sqrt{ 2}} \left(1-\sqrt{1-b}\right) \geq  \frac{\rho_0}{\sqrt{ 2}} \frac{b}{4} = \frac{\rho_0}{  8 \sqrt{2} n_*}\,,
	\end{split}
	\end{equation}
both the $\frac{\varepsilon_{n+1}}{4}$ bounds and condition \eqref{prove.me} are satisfied, choosing $n_*$ as in \eqref{ho.ho.ho} and $\varepsilon < \varepsilon_0$ for some $\varepsilon_0$ depending only on $r, \tau, \gamma, \kappa, \rho, \Omega, \|V\|_{\kappa, \rho}, \{ \|N^{(\alpha)}\|_{0}\}_{\alpha}$.

Last, we show that
	$$
	Y^{(n+1)}_{\inv}(\varphi) := e^{-\ii G^{(n)}_{\inv}( \varphi)} Y^{(n)}_{\inv}( \varphi)
	$$
	satisfies \eqref{y.diff}. {To this aim we observe that \eqref{one.more.time}, \eqref{delta.below} and \eqref{ho.ho.ho} imply}
	\begin{equation}\label{prove.me.too}
	\frac{8 e^{-\kappa_{n+1}}}{\delta_{n}}\| G^{(n)}_{\inv}\|_{\kappa_{n+1} + \delta_n, \rho_{n+1}} < \varepsilon^{\frac 1 2} \| V\|_{\kappa_0, \rho_0} \left({\frac{1}{e}}\right)^n\,.
	\end{equation}	
	{Thus, due to \eqref{prove.me.too}, by Lemma \ref{cor.lignano.pineta} one has}
	$\| e^{\ii G^{(n)}_{\inv}} P e^{-\ii G^{(n)}_{\inv}} - P\|_{\kappa_{n+1}, \rho_{n+1}} \leq \varepsilon^{\frac 1 2} \left({\frac{1}{e}}\right)^n \| V\|_{\kappa_0, \rho_0} \|P\|_{\kappa_{0}, \rho_0}$.
	Then \eqref{y.diff} follows estimating
	\begin{align*}
		\|Y_\inv^{(n+1)} P (Y_\inv^{(n+1)})^* \|_{\kappa_{n+1},\rho_{n+1}} &=\|e^{-\ii G_\inv^{(n)}} Y_\inv^{(n)} P (Y_\inv^{(n)})^* e^{\ii G_\inv^{(n)}}-P \|_{\kappa_{n+1},\rho_{n+1}} \\
		&\leq \| e^{-\ii G_\inv^{(n)}} (Y^{(n)}_\inv P (Y_\inv^{(n)})^*-P) e^{\ii G^{(n)}_\inv} - (Y_\inv^{(n)} P (Y_\inv^{(n)})^*-P) \|_{\kappa_{n+1},\rho_{n+1}} \\
		&\quad +\Vert Y^{(n)}_\inv P (Y^{(n)}_\inv)^*-P \|_{\kappa_{n+1},\rho_{n+1}} + \| e^{-\ii G^{(n)}_\inv} P e^{\ii G^{(n)}_\inv} - P \|_{\kappa_{n+1},\rho_{n+1}} \\
		&\leq \varepsilon^{\frac12} \|V\|_{\kappa_0, \rho_0} e^{-n} \left(2 \varepsilon^{\frac12} \| V\|_{\kappa_0,\rho_0} \sum_{n'=0}^{n-1} e^{-n'} \| P \|_{\kappa_0,\rho_0} \right)  \\
		&\quad + 2 \varepsilon^{\frac12} \| V \|_{\kappa_0,\rho_0} \sum_{n'=0}^{n-1} e^{-n'} \| P \|_{\kappa_0,\rho_0} + \varepsilon^{\frac12} e^{-n}  \| V \|_{\kappa_0,\rho_0}  \| P \|_{\kappa_0,\rho_0} \\
		& \leq 2 \varepsilon^{\frac12} \| V \|_{\kappa_0,\rho_0} \sum_{n'=0}^{n} e^{-n'} \| P \|_{\kappa_0,\rho_0}
	\end{align*}	
%
%
	where in the last step we have assumed
	${2 \varepsilon^{\frac 1 2}\| V\|_{\kappa_0, \rho_0}} \sum_{n'=0}^{{n-1}} e^{{-n'}} \leq 1$,
	which is satisfied provided
	$\varepsilon \leq \left( {\frac{2e\| V\|_{\kappa_0, \rho_0}}{e-1}}\right)^{-2}$.
\end{proof}

Analogously, Item (b) of Proposition \ref{teo:main} immediately follows from the following lemma.

\begin{lemma}\label{nf.thm.Hobbes}
	Under the assumptions of Proposition \ref{teo:main}, define
	\begin{equation}\label{ho.ho.ho.Hobbes}
	n^* := \left\lfloor \varepsilon^{-\mathtt{b}}\right\rfloor\,, 
	\end{equation}
	with $\mathtt{b}$ as in \eqref{parametri.vitali}, and $\forall n = 1, \dots, n^*$ define
	\begin{equation}
	\kappa_n := {\kappa_0}\sqrt{1 - b n}\,, \quad \rho_n := {\rho_0}\sqrt{1 - b n}\,,\quad b := \frac{1}{2 n^*}\,,
	\end{equation}
	with $\kappa_0, \rho_0$ as in \eqref{ho.ho.ho} and $\ccccc$ as in \eqref{costantina}.
	Let $\{\varepsilon_n\}_{n= 0}^{n^*}$ be defined with $\varepsilon_n$ as in \eqref{granchietti}, then $\forall n = 0, \dots, n^*$ there exist a time quasi-periodic family of unitary maps $Y^{(n)}_{\obs}(\omega t)$, a self-adjoint operator $Z^{(n)}_{\obs}$ and a time quasi-periodic family of self-adjoint operators $V^{(n)}_{\obs}(\omega t)$ such that, defining
	\begin{equation}
	H^{(n)}_{\obs}(\omega t):= J \cdot N + Z^{(n)}_{\obs} + V^{(n)}_{\obs}( \omega t) \,,
	\end{equation}
	one has
	\begin{equation}
	U_{H^{(n)}_{\obs}}(t) = Y^{(n)}_{\obs}( \omega t) U_{H}(t)\,,
	\end{equation}
	and the following are satisfied:
	\begin{enumerate}
		\item $Z^{(n)}_{\obs} \in \mathcal{O}^\strong_{\kappa_n}$ is time independent, with $Z^{(0)}_\obs=0$ and for $n>0$
		\begin{equation}\label{z.small.Hobbes}
		\|Z^{(n)}_{\obs}\|_{\kappa_n} \leq \varepsilon^{-\frac 1 2}\sum_{n'=0}^{{n-1}} \varepsilon_{{n'}}\,;
		\end{equation}
		\item $V^{(n)}_{\obs} \in \mathcal{O}^\strong_{\kappa_n, \rho_n}$, and 
		\begin{equation}\label{v.small.Hobbes}
		\| V^{(n)}_{\obs}\|_{\kappa_n, \rho_n} \leq \varepsilon_n\,;
		\end{equation}
		\item The maps $Y^{(n)}_{\obs}(\omega t)$ are such that $Y^{(n)}_{\obs}(0) = \mathbbm{1}$ and they are close to the identity, in the sense that $Y_\obs^{(0)}(\omega t)\equiv\mathbbm{1}$ and for $n>0$
		\begin{equation}\label{y.diff.Hobbes}
		\| Y^{(n)}_{\obs} P (Y^{(n)}_{\obs})^*- P\|_{\kappa_n, \rho_n} \leq 2 \varepsilon^{\frac 12} \sum_{n'=0}^{{n-1}} e^{{-n'}}\|V\|_{\kappa, \rho} \|P\|_{\kappa, \rho} \quad \forall P \in \mathcal{O}_{\kappa, \rho}\,.
		\end{equation}
		Moreover, for $n>0$ they are of the form {$Y^{(n)}_\obs(\omega t) = e^{-\ii G^{({n-1})}_\obs(\omega t)} \cdots e^{-\ii G^{(0)}_\obs(\omega t)}$}, with
		\begin{equation}\label{Gn.per.n}
		\| G^{(n)}_\obs\|_{\kappa_{n+1}, \rho_{n+1}} \leq \varepsilon^{\frac 12} e^{-n} \|V\|_{\kappa_0, \rho_0}\,.
		\end{equation}
	\end{enumerate}
\end{lemma}
\begin{proof}
	The proof follows the same scheme of the proof of Lemma \ref{nf.thm}, thus we only sketch the procedure highlighting the differences. At any step $n$ one looks for a change of coordinates of the form
	$$
	\widetilde{U}^{(n)}_{\obs}(t) = e^{-\ii G^{(n)}_{\obs}(\omega t)} U_{H^{(n)}_{\obs}}(t)\,,
	$$
	where $G^{(n)}_{\obs}$ solves
	\begin{equation}\label{hom.again.Hobbes}
	\ii \left[J \cdot N, G^{(n)}_{\obs}(\varphi)\right]  + \omega \cdot \partial_\varphi{G}^{(n)}_{\obs}(\varphi) + V^{(n)}_{\obs}( \varphi) - \llan V^{(n)}_{\obs} \rran = 0\,,
	\end{equation}
	with $\llan V^{(n)}_\obs \rran$ as in \eqref{cacio.e.pepe}. Then by Lemma \ref{lem.pushfwd} one has
	\begin{subequations}
		\begin{align}		
\nonumber
		H^{(n+1)}_{\obs}(\varphi) &= J \cdot N + Z^{(n)}_{\obs} + \llan V^{(n)}_{\obs} \rran\\
		\label{vanishes.A}
		&+ \ii \left[J \cdot N, G^{(n)}_{\obs}(\varphi)\right]  + \omega\cdot \partial_\varphi {G}^{(n)}_{\obs}(\varphi) + V^{(n)}_{\obs}( \varphi) - \llan V^{(n)}_{\obs} \rran\\
		\label{rn.T}
		& + e^{-\ii G^{(n)}_{\obs}(\varphi)} (J \cdot N) e^{\ii G^{(n)}_{\obs}(\varphi)} - J \cdot N - \ii [J \cdot N, G^{(n)}_{\obs}(\varphi)] \\
		\label{rv.T}
		& +  e^{-\ii G^{(n)}_{\obs}(\varphi)} V^{(n)}_{\obs}( \varphi) e^{\ii G^{(n)}_{\obs}(\varphi)} - V^{(n)}_{\obs}( \varphi)\\
		\label{rz.E}
		& +  e^{-\ii G^{(n)}_{\obs}( \varphi)} Z^{(n)}_{\obs} e^{\ii G^{(n)}_{\obs}(\varphi)} - Z^{(n)}_{\obs}\\
		\label{rg.O}
		& + \int_0^1 e^{-\ii G^{(n)}_{\obs}(\varphi)s} \omega \cdot \partial_\varphi G^{(n)}_{\obs}( \varphi) e^{\ii G^{(n)}_{\obs}(\varphi)s }\,\ud s - \omega \cdot \partial_\varphi {G}^{(n)}_{\obs}( \varphi)\,.
		\end{align}
	\end{subequations} 
	Then one defines
	\begin{equation}\label{new.objects}
	Z_{\obs}^{(n+1)} := Z_{\obs}^{(n)} + \llan V^{(n)}_{\obs} \rran\,, \quad V_{\obs}^{(n+1)} = \eqref{rn.T} + \eqref{rv.T} + \eqref{rz.E} + \eqref{rg.O}\,,
	\end{equation}
	and defining $\delta_n := \frac{\rho_n - \rho_{n+1}}{2}$, Lemma \ref{lem.Sailor.Moon.is.back} implies
	\begin{equation} \label{sketchy}
	\begin{gathered}
	\|G^{(n)}_\obs\|_{\kappa_n -\ccccc\delta_n, \rho_n-\delta_n} \leq \frac{2^{2r+1} \tau^{\tau}}{e^{\tau} \gamma \delta_n^{\tau+r}}(1+\delta_{{n}})^r \| V^{(n)}_\obs\|_{\kappa, \rho}\,, \\
	\| \llan V^{(n)}_\obs \rran \|_{\kappa_n - \ccccc \delta_n} \leq \left( \frac{\mathcal{J} 2 ^{2r}(1 +\delta_n)^r \tau^\tau}{\gamma e^\tau \delta_n^{\tau + r + 1}} + 1 \right) \|V^{(n)}_\obs\|_{\kappa_n, \rho_n}\,.
	\end{gathered}
	\end{equation}
	Then one argues as in the proof of Lemma \ref{nf.thm}, with the only difference that for any $n = 0, \dots, n^*$ one has
	\begin{equation}\label{not.so.normal.form}
	\| \llan V^{(n)}_\obs \rran \|_{\kappa_n - \ccccc \delta_n} \leq \left( \frac{\mathcal{J} 2 ^{2r}(1 +\delta_n)^r \tau^\tau}{\gamma e^\tau \delta_n^{\tau + r + 1}} + 1 \right) \varepsilon_n \leq  \frac{\mathcal{J} 2 ^{4r + 1}\tau^\tau}{\gamma e^\tau} \left(\frac{8 \sqrt{2} n^*}{\rho_0}\right)^{\tau + r + 1} \varepsilon_n \leq \varepsilon^{-\frac 1 2} \varepsilon_n\,,
	\end{equation}
	provided $\varepsilon \leq \varepsilon_0$ with $\varepsilon_0$ small enough and depending on $\mathcal{J}, r, \gamma, \tau, \rho_0$. Summing over $n$ and recalling the definition of $Z^{(n+1)}_\obs$ as in \eqref{new.objects}, by estimate \eqref{not.so.normal.form} one gets \eqref{z.small.Hobbes}. Equation \eqref{Gn.per.n} is verified using \eqref{sketchy} together with the inductive assumption \eqref{v.small.Hobbes} on $V^{(n)}_\obs$, and \eqref{y.diff.Hobbes} is proven defining recursively $Y_\obs^{(n+1)}(\varphi) := e^{-\ii G^{(n)}_\obs(\varphi)} Y^{(n)}_{\obs}(\varphi)$ and arguing as in the proof of Lemma \ref{nf.thm}.  
\end{proof}

\begin{proof}[Proof of Proposition \ref{teo:main}]
	In order to obtain Item (a), it is sufficient to take $Z_\inv := Z^{(n_*)}_\inv$, $V_\inv := V^{(n_*)}_\inv$, $H_\inv := H^{(n_*)}_\inv$ and $Y_\inv := Y^{(n_*)}_\inv$, with $Z^{(n_*)}_\inv,\ V^{(n_*)}_\inv\, H^{(n_*)}_\inv,\ Y^{(n_*)}_\inv$ and $n_*$ as in Lemma \ref{nf.thm}. Analogously, Item (b) is obtained taking $Z_\obs := Z^{(n^*)}_\obs$, $V_\obs := V^{(n^*)}_\obs$, $H_\obs := H^{(n^*)}_\obs$ and $Y_\obs := Y^{(n^*)}_\obs$, with $Z^{(n^*)}_\obs,\ V^{(n^*)}_\obs\, H^{(n^*)}_\obs,\ Y^{(n^*)}_\obs$ and $n^*$ as in Lemma \ref{nf.thm.Hobbes}.
\end{proof}

\subsection{Fast-Forcing Perturbations}

\begin{lemma}\label{nf.thm.ff}
	Under the assumptions of Proposition \ref{teo:main.ff}, define
	\begin{equation}\label{ho.ho.ho.ff}
	n_\star := \left\lfloor \lambda^{2\beta}\right\rfloor\,, \quad \kappa_0:= \min\{\kappa, \rho\}\,, \quad \rho_0:=2 (\ccccc + 1)^{-1} \kappa_0\,,
	\end{equation}
	with $\ccccc$ as in \eqref{costantina} and $\beta$ as in \eqref{parametri.vitali.ff}, and define $\forall n = 1, \dots, n_\star$
	\begin{equation}
	\kappa_n := {\kappa_0}\sqrt{1 - b n}\,, \quad \rho_n := {\rho_0}\sqrt{1 - b n}\,,\quad b = \frac{1}{2 n_\star}\,,
	\end{equation}
	and $\forall n = 0, \dots, n_\star$
	\begin{equation}\label{granchietti.ff}
	\varepsilon_n := \left(\frac 1 e\right)^{n-1} \lambda^{-\frac{1}{4\tau_\omega}} \|V\|_{\kappa_0, \rho_0}\,.
	\end{equation}
	Then, for any $n = 0, \dots, n_\star$ there exist a time quasi-periodic family of unitary maps $Y^{(n)}_\inv(\lambda \omega t)$, a self-adjoint operator $Z^{(n)}_\inv$ and a time quasi-periodic family of self-adjoint operators $V_{\inv}^{(n)}(\lambda \omega t)$ such that, defining
	\begin{equation}
	H^{(n)}_\inv(t):= J \cdot N + Z^{(n)}_\inv + V^{(n)}_\inv( \lambda \omega t) \,,
	\end{equation}
	one has
	\begin{equation}
	U_{H^{(n)}_\inv}(t) = Y^{(n)}_{\inv}(\lambda \omega t) U_{H}(t) (Y^{(n)}_\inv)^*(0)\,,
	\end{equation}	and the following are satisfied:
	\begin{enumerate}
		\item $Z^{(n)}_\inv \in \mathcal{O}^\strong_{\kappa_n}$ is time independent, {$Z^{(0)}_\inv=0$, and for any $n>0$}
		\begin{equation}\label{d.new.ff}
		 [Z^{(n)}_\inv, N^{(\alpha)}] = 0 \quad \forall \alpha = 1, \dots, r\,, \quad \|Z^{(n)}_\inv\|_{\kappa_n} \leq \sum_{n'=0}^{{n-1}} \varepsilon_{{n'}}\,;
		\end{equation}
		\item $V^{(n)}_\inv \in \mathcal{O}^\strong_{\kappa_n, \rho_n}$, and
		\begin{equation}\label{v.new.ff}
		\| V^{(n)}_\inv\|_{\kappa_n, \rho_n} \leq\left\{ \begin{array}{lll} \varepsilon_{n}\,, & \qquad & n \geq 1 \\
		& & \\
		\Vert V \Vert_{\kappa_0,\rho_0} \, , & &n=0
		\end{array} \right. ;
		\end{equation}
		\item The maps $Y^{(n)}_\inv(\lambda \omega t)$ are close to the identity, in the sense that {$Y^{(0)}_\inv(\lambda \omega t)\equiv  \mathbbm{1}$, and for any $n>0$}
		\begin{equation}\label{y.not.ff}
		\| Y^{(n)}_\inv P (Y^{(n)}_\inv)^*- P\|_{\kappa_n, \rho_n} \leq 2 \lambda^{\frac{1}{8 \tau_\omega}} \sum_{n'=0}^{n-1}\varepsilon_{n'} \|P\|_{\kappa_n, \rho_n} \quad \forall P \in \mathcal{O}_{\kappa_n, \rho_n}\,.
		\end{equation}
	\end{enumerate}
\end{lemma}

	\begin{proof} The proof is analogous to the proof of Lemma \ref{nf.thm} with some differences that we highlight here. Instead of \eqref{hom.again}, $G^{(n)}_\inv$ solves 
\begin{equation}\label{hom.again.ff}
	\ii \left[J \cdot N, G^{(n)}_\inv(\varphi)\right]  + \lambda \omega \cdot \partial_\varphi{G}^{(n)}_\inv(\varphi) + (V^{(n)}_{\inv})^\ir( \varphi) - \langle V^{(n)}_\inv \rangle = 0\,,
\end{equation}
	with $\langle V^{(n)}_\inv \rangle$ defined in \eqref{zeta.esplicita}. In the transformed Hamiltonian, we also split the $\uv$/$\ir$ parts and we have to control the following terms
\begin{subequations}
	\begin{align}
	\nonumber
	H^{(n+1)}_\inv &= J \cdot N + Z^{(n)}_\inv + \langle V^{(n)}_\inv \rangle\\
	\label{vanishes.ff}
	&+ \ii \left[J \cdot N, G^{(n)}_\inv\right]  + \lambda \omega \cdot \partial_\varphi{G}^{(n)}_\inv + (V_\inv^{(n)})^\ir - \langle V^{(n)}_\inv \rangle\\
	\label{vuv.ff}
	&+ V_\inv^{(n)} - (V^{(n)}_\inv)^\ir\tag{\theequation$^\prime$}\\
	\label{rn.ff}
	& + e^{-\ii G^{(n)}_\inv} (J \cdot N) e^{\ii G^{(n)}_\inv} - J \cdot N - [J \cdot N, G^{(n)}_\inv] \\
	\label{rv.ff}
	& +  e^{-\ii G^{(n)}_\inv} V^{(n)}_\inv e^{\ii G^{(n)}_\inv} - V^{(n)}_\inv\\
	\label{rz.ff}
	& +  e^{-\ii G^{(n)}_\inv} Z^{(n)}_\inv e^{\ii G_n} - Z^{(n)}_\inv\\
	\label{rg.ff}
	& + \int_0^1 e^{-\ii G^{(n)}_\inv s}\lambda (\omega \cdot \partial_\varphi) (G^{(n)}_\inv) e^{\ii G^{(n)}_\inv s}\,\ud s - \lambda \omega \cdot \partial_\varphi G^{(n)}_\inv\,.
	\end{align}
\end{subequations}
One then sets $Z_\inv^{(n+1)}:=Z_\inv^{(n)}+\langle V^{(n)}_\inv \rangle$ and, differently from \eqref{eq:AiutoIlCapoMiRapisce}, one has to take into account the term \eqref{vuv.ff}:
\begin{equation}
	V^{(n+1)}_\inv := \eqref{vuv.ff}+\eqref{rn.ff} + \eqref{rv.ff} + \eqref{rz.ff} + \eqref{rg.ff}\,.
\end{equation}
	We observe that, due to our choice of the parameters $\{\kappa_n\}_{n=1}^{n_\star}$, $\{\rho_n\}_{n=1}^{n_\star}$, one has $\kappa_n = (\ccccc + 1) \rho_n$ $\forall n$, thus also
	$$
	\kappa_{n} - \kappa_{n+1} = ( \ccccc + 2) (\rho_n -\rho_{n+1}) = ( \ccccc + 2) \delta_n
	$$
and  $\kappa_{n} - \ccccc \delta_n = \kappa_{n+1} + 2 \delta_n$. 
		Using Corollary \ref{cor.trova.g.ff} we find $G^{(n)}_\inv$ solving \eqref{hom.again.ff}, and the operators $G^{(n)}_\inv, \langle V^{(n)}_\inv \rangle$ satisfy
\begin{equation}\label{one.more.time.ff}
	\begin{gathered}
	\begin{aligned}
	\| G^{(n)}_\inv\|_{\kappa_{n+1} + 2 \delta_n, \rho_n} &\leq \frac{2^{2r+2}}{\gamma_\omega \lambda^{1/2}} \left(\frac{1+\delta_n}{\delta_n}\right)^r \Vert V^{(n)}_\inv - \langle V^{(n)}_\inv \rangle_{\mathbb{T}^m} \Vert_{\kappa_n,\rho_n} + \frac{2^{2r+1}\tau_J^{\tau_J}}{\gamma_J e^{\tau_J} \delta_n^{\tau_J}}\left(\frac{1+\delta_n}{\delta_n}\right)^r \Vert \langle V^{(n)}_\inv \rangle_{\mathbb{T}^m} \Vert_{\kappa_n ,0}\\
	& \leq  \frac{2^{2r+2}}{\gamma_\omega \lambda^{1/2}}\left(\frac{1+\delta_n}{\delta_n}\right)^r \varepsilon_n + \frac{2^{2r+1}\tau_J^{\tau_J}}{\gamma_J e^{\tau_J} \delta_n^{\tau_J}} \left(\frac{1+\delta_n}{\delta_n}\right)^r \varepsilon_n \quad \textrm{if} \quad n \geq 1\,,
	\end{aligned}\\
	\| G^{(0)}_\inv\|_{\kappa_{1} + 2\delta_0, \rho_0} \leq \frac{2^{2r+2}}{\gamma_\omega \lambda^{1/2}}\left(\frac{1+\delta_0}{\delta_0}\right)^r \Vert V^{(0)}_\inv \Vert_{\kappa_0,\rho_0} \leq  \frac{2^{2r+3}}{\gamma_\omega \lambda^{1/2}} \left(\frac{1+\delta_0}{\delta_0}\right)^r\|V\|_{\kappa_0, \rho_0}\,,\\
	\|\langle V^{(n)}_\inv \rangle\|_{\kappa_{n+1} +2 \delta_n} \leq \|V^{(n)}_\inv\|_{\kappa_n, \rho_n} \leq \varepsilon_n \quad \textrm{if} \quad n \geq 1\,,\\
	\|\langle V^{(0)}_\inv \rangle\|_{\kappa_{1} + 2\delta_0}  = 0\,,
	\end{gathered}
	\end{equation}
	where the case $n=0$ is treated differently from all other cases $n\geq 1$ since $\langle V \rangle_{\mathbb{T}^m} =0$. The terms
\eqref{rn.ff}, \eqref{rv.ff}, \eqref{rz.ff} are treated analogously as in the proof of Proposition \ref{nf.thm} (with obvious adaptation of the norms), while for \eqref{vuv.ff} we use Corollary \ref{cor.p.small} and for  \eqref{rg.ff} we use the fact that $\lambda \omega \cdot \partial_\varphi G^{(n)}_\inv$ solves the homological equation \eqref{hom.again.ff}, thus obtaining
\begin{eqnarray}
	\!\!\!\!\!\!\!\!\!\Vert \text{\eqref{vuv.ff}} \Vert_{\kappa_{n+1},\rho_{n+1}}\!\!\!\!\!\! &\leq& \!\!\!\!\!\!\frac{2^{2r+1}}{e \delta_n^{r+1}}(1+\delta_n)^r \lambda^{-\frac{1}{2\tau_\omega}} \Vert V^{(n)}_\inv \Vert_{\kappa_n,\rho_n} \\
	\!\!\!\!\!\!\!\!\!\Vert \text{\eqref{rg.ff}} \Vert_{\kappa_{n+1},\rho_{n+1}}\!\!\!\!\!\! &\leq&\!\!\!\!\!\!	\label{g.gdot.mini.ff}
	 \frac{8 e^{-\kappa_{n+1}}}{\delta_{n}}\| G^{(n)}_\inv\|_{\kappa_{n+1} + 2\delta_n, \rho_{n+1}} \|\ii[J \cdot N, G^{(n)}_\inv ] + V^{(n)}_\inv - \langle V^{(n)}_\inv \rangle\|_{\kappa_{n+1} + \delta_n, \rho_{n+1}}.
\end{eqnarray}
Then, the inductive hypothesis follows if each of the terms \eqref{vuv.ff}-\eqref{rg.ff} can be bounded by $\frac{\varepsilon_{n+1}}{5}$. Then, repeating with straightforward adaptations the proof of Lemma \ref{nf.thm}, all the bounds are satisfied by choosing $n_\star$ as in \eqref{ho.ho.ho.ff} with $\lambda > \lambda_0$ depending on $r$, $\tau_\omega$, $\tau_J$, $\gamma_\omega$, $\gamma_J$, $\kappa$, $\rho$, $\Vert H_0\Vert_\kappa$, $\Vert V \Vert_{\kappa,\rho}$.

To prove \eqref{y.not.ff}, one uses the same argument of Lemma \ref{nf.thm} substituting $\varepsilon^{\frac{1}{2}} \mapsto \lambda^{-\frac{1}{8\tau_\omega}}$.
\end{proof}

\begin{lemma}\label{nf.thm.Hobbes.ff}
	Under the assumptions of Proposition \ref{teo:main.ff}, define
	\begin{equation}\label{ho.ho.ho.Hobbes.ff}
	n^{\star} := \left\lfloor \lambda^{\beta}\right\rfloor\,, 
	\end{equation}
	with $\beta$ as in \eqref{parametri.vitali.ff}, and $\forall n = 1, \dots, n^{\star}$ define
	\begin{equation}\label{kaippirini_n.ff}
	\kappa_n := {\kappa_0}\sqrt{1 - b n}\,, \quad \rho_n := {\rho_0}\sqrt{1 - b n}\,,\quad b := \frac{1}{2 n^{\star}}\,,
	\end{equation}
	with $\kappa_0, \rho_0$ as in \eqref{ho.ho.ho.ff} and $\ccccc$ as in \eqref{costantina}.
	Let $\{\varepsilon_n\}_{n= 0}^{n^{\star}}$ be defined with $\varepsilon_n$ as in \eqref{granchietti.ff}, then $\forall n = 0, \dots, n^{\star}$ there exist a time quasi-periodic family of unitary maps $Y^{(n)}_{\obs}(\lambda\omega t)$, a self-adjoint operator $Z^{(n)}_{\obs}$ and a time quasi-periodic family of self-adjoint operators $V^{(n)}_{\obs}(\lambda\omega t)$ such that, defining
	\begin{equation}
	H^{(n)}_{\obs}(\lambda\omega t):= J \cdot N + Z^{(n)}_{\obs} + V^{(n)}_{\obs}( \lambda\omega t) \,,
	\end{equation}
	one has
	\begin{equation}
	U_{H^{(n)}_{\obs}}(t) = Y^{(n)}_{\obs}(\lambda \omega t) U_{H}(t)\,,
	\end{equation}
	and the following are satisfied:
	\begin{enumerate}
		\item $Z^{(n)}_{\obs} \in \mathcal{O}^\strong_{\kappa_n}$ is time independent, $Z^{(0)}_\obs = 0$, and for any $n >0$
		\begin{equation}\label{z.small.Hobbes.obs.ff}
		\|Z^{(n)}_{\obs}\|_{\kappa_n} \leq \lambda^{\frac{1}{8 \tau_\omega}}\sum_{n'=0}^{{n-1}} \varepsilon_{n'}\,;
		\end{equation}
		\item $V^{(n)}_{\obs} \in \mathcal{O}^\strong_{\kappa_n, \rho_n}$, and 
		\begin{equation}\label{v.small.Hobbes.ff}
		\| V^{(n)}_{\obs}\|_{\kappa_n, \rho_n} \leq \varepsilon_n\,;
		\end{equation}
		\item The maps $Y^{(n)}_{\obs}(\lambda \omega t)$ are such that $Y^{(n)}_{\obs}(0) = \mathbbm{1}$ and they are close to the identity, in the sense that $Y^{(0)}_\obs(\lambda \omega t) \equiv \mathbbm{1}$, and $\forall n >0$
		\begin{equation}\label{y.diff.Hobbes.ff}
		\| Y^{(n)}_{\obs} P (Y^{(n)}_{\obs})^*- P\|_{\kappa_n, \rho_n} \leq 2 \lambda^{-\frac{1}{8\tau_\omega}} \sum_{n'=0}^{{n-1}} e^{{-n'}}\|V\|_{\kappa, \rho} \|P\|_{\kappa, \rho} \quad \forall P \in \mathcal{O}_{\kappa, \rho}\,.
		\end{equation}
		Moreover, they are of the form {$Y^{(n)}_\obs(\omega t) = e^{-\ii G^{({n-1})}_\obs(\omega t)} \cdots e^{-\ii G^{(0)}_\obs(\omega t)}$ $\forall n >0$}, with
		\begin{equation}\label{Gn.per.n.ff}
		\| G^{(n)}_\obs\|_{\kappa_{n+1}, \rho_{n+1}} \leq \lambda^{-\frac {1}{8\tau_\omega}} e^{-n} \|V\|_{\kappa_0, \rho_0}\,.
		\end{equation}
	\end{enumerate}
\end{lemma}
\begin{proof}The algebraic scheme of the proof is contained in Appendix H of \cite{Else2020}. For the quantitative estimates, one argues as in the proof of Lemma \ref{nf.thm.Hobbes} and \ref{nf.thm.ff}.
\end{proof}

\begin{proof}[Proof of Proposition \ref{teo:main.ff}]
	In order to obtain Item (a), it is sufficient to take $Z_\obs := Z^{(n_\star)}_\inv$, $V_\inv := V^{(n_\star)}_\inv$, $H_\inv := H^{(n_\star)}_\inv$ and $Y_\inv := Y^{(n_\star)}_\inv$, with $Z^{(n_\star)}_\inv,\ V^{(n_\star)}_\inv,\ H^{(n_\star)}_\inv,\ Y^{(n_\star)}_\inv$ and $n_\star$ as in Lemma \ref{nf.thm.ff}. Analogously, Item (b) is obtained taking $Z_\obs := Z^{(n^\star)}_\obs$, $V_\obs := V^{(n^\star)}_\obs$, $H_\obs := H^{(n^\star)}_\obs$ and $Y_\obs := Y^{(n^\star)}_\obs$, with $Z^{(n^\star)}_\obs,\ V^{(n^\star)}_\obs,\ H^{(n^\star)}_\obs,\ Y^{(n^\star)}_\obs$ and $n^\star$ as in Lemma \ref{nf.thm.Hobbes.ff}.
\end{proof}

\section{Proof of Dynamical Consequences}\label{sec:physical.cons}

\subsection{Small Perturbations}
\begin{proof}[Proof of Theorem \ref{teo:SlowHeating}]
{Let $H_\inv, Z_\inv, V_\inv$ and $Y_\inv$ be defined as in Proposition \ref{teo:main}. Using \eqref{uH.uHstar}, we get
\[
	\begin{split}
		\| U_H^*(t) N^{(\alpha)} U_H(t)-N^{(\alpha)}\|_{\mathrm{op}} & \leq \|U_H^*(t) N^{(\alpha)} U_H(t)-Y_\inv^*(0) N^{(\alpha)} Y_\inv(0)\|_{\mathrm{op}} + \|Y_\inv^*(0) N^{(\alpha)} Y_\inv(0)-N^{(\alpha)}\|_{\mathrm{op}} \\
		& \leq \|U_{H_\inv}^*(t)  Y_\inv(\omega t) N^{(\alpha)} Y_\inv^*(\omega t) U_{H_\inv}(t)- N^{(\alpha)} \|_{\mathrm{op}}\\
		& \quad + \|Y_\inv^*(0) N^{(\alpha)} Y_\inv(0)-N^{(\alpha)}\|_{\mathrm{op}} \\
		&\leq \|U_{H_\inv}^*(t) Y_\inv^*(\omega t) N^{(\alpha)} Y_\inv(\omega t) U_{H_\inv}(t)-U_{H_\inv}^*(t) N^{(\alpha)}U_{H_\inv}(t) \|_{\mathrm{op}} \\
		&\quad + \|Y^*_\inv(0) N^{(\alpha)} Y_\inv(0)-N^{(\alpha)}\|_{\mathrm{op}}+\|U_{H_\inv}^*(t) N^{(\alpha)} U_{H_\inv}(t)-N^{(\alpha)} \|_{\mathrm{op}} \, .
	\end{split}
\]
 Now we estimate the three terms separately. For the first and the second term, we use Lemma \ref{lem.op.loc.norm} and Item (a.iv) of Proposition \ref{teo:main}.} To estimate the third one, we use Duhamel together with the fact that, by Item (a.i) of Proposition \ref{teo:main}, $[H_{\inv},N^{(\alpha)}]=0$, as follows:
	\begin{equation}\label{duhamel.N}
	\begin{split}
	U_{H_\inv}(t) N^{(\alpha)} U_{H_\inv}^*(t)-N^{(\alpha)} &= \ii \int_0^t U_{H_\inv}(s) [H_\inv(\omega s),N^{(\alpha)}] U_{H_\inv}^*(s) \, \ud s  \\
	&=\ii \int_0^t U_{H_\inv}(s) [V_\inv(\omega s),N^{(\alpha)}] U_{H_\inv}^*(s) \, \ud s
	\end{split}
	\end{equation}
	and we use again Lemma \ref{lem.op.loc.norm} and Item (a.iii) of Proposition \ref{teo:main}.
\end{proof}
In the remaining part of the subsection, we prove Theorem \ref{thm:EvLocObs}. The heart of the proof goes as follows: let
\begin{equation}\label{H.eff}
	H_{\textrm{eff}} :=  J \cdot N + Z_{\obs}\,,
\end{equation}
with $Z_\obs$ defined in Item (b) Theorem \ref{teo:main}. For any local observable $O$, one bounds
$$
\left\|U_H^*(t) O U_H(t) - e^{-\ii H_{\mathrm{eff}} t} O e^{\ii H_{\mathrm{eff}}t} \right\|_{\mathrm{op}}\,
$$
by using triangular inequality and estimates separately
\begin{equation}\label{puntopunto}
\left\|U_{H_{\obs}}^*(t) O U_{H_\obs}(t) - e^{-\ii H_{\mathrm{eff}} t} O e^{\ii H_{\mathrm{eff}}t} \right\|_{\mathrm{op}}\,, \quad \left\|U_H^*(t) O U_H(t) - U^*_{H_{\obs}}(t) O U_{H_{\obs}}(t) \right\|_{\mathrm{op}}\,.
\end{equation}
\begin{lemma}[Lieb-Robinson bound \cite{Lieb1972}]\label{lieb.rob}
	Let $Z=\sum_{S \in \mathcal{P}_c(\Lambda)} Z_S$ be a self-adjoint operator with $\Vert Z \Vert_{2\kappa} < +\infty$ for some $\kappa > 0$. Let the operators $A,B$ act within $S_A,S_B \subset \Lambda$ respectively. Then,
	\begin{equation}
	\Vert [A, e^{\ii t Z} B e^{-\ii t Z} ] \Vert_{\mathrm{op}} \; \leq \; \Vert A \Vert_{\mathrm{op}} \Vert B \Vert_{\mathrm{op}} e^{-\kappa(d(S_A,S_B)-vt)} \min(|S_A|,|S_B|)
	\end{equation}
	with Lieb-Robinson speed $v=v(Z,\kappa):=C(d)(\kappa^{-(d+2)} e^\kappa) \Vert Z \Vert_{2 \kappa}$ and $C(d)$ only depending on the spatial dimension $d$.
\end{lemma}

The following result is proven in \cite{Abanin2017} for time-periodic operators (see Section 5.1 therein). For completeness, we present a proof of it in Appendix \ref{append:ProofTec}.
\begin{lemma}[See also \cite{Abanin2017}]\label{lem:LemmaSempreDiverso}
	Let $O$ be a local observable acting within $S_O$, $A \in \mathcal{O}_{\kappa,0}$ and $Z \in \mathcal{O}_{2 \kappa}$.  Then there exists a positive constant ${C}(|S_O|,d,\kappa)$ such that
	\begin{equation}
	\hspace{-0.09cm}\int_0^t \ud \tau \Vert[A(\omega \tau),e^{-\ii \tau Z} O e^{\ii \tau Z} ] \Vert_{\mathrm{op}} \leq C(|S_O|,d,\kappa) \langle\Vert  Z \Vert_{2 \kappa} \rangle^{d+1} \langle t \rangle^{d+2} \Vert O \Vert_{\mathrm{op}} \Vert A \Vert_{\kappa,0} \, .
	\end{equation}
\end{lemma}

In the following, we will define
\begin{equation}
\kappa_{\min} \;:=\; \frac{\kappa_{*}}{2} \;=\; \frac{\min\{\kappa, \rho\}}{2 \sqrt{2}}\,,
\end{equation}
with $\kappa_*$ defined as in \eqref{def:oggetti.finali}, so that the operators $Z_{\obs}$, $V_{\obs}$ and $\{G^{(n)}_{\obs}\}_{n=1}^{n^*}$ defined in Proposition \ref{nf.thm} satisfy
\begin{equation}\label{eq:delle.stime}
Z_{\obs} \in \mathcal{O}_{2\kappa_{\min}}\,, \quad V_{\obs} \in \mathcal{O}_{2 \kappa_{\min}, 0}\,, \quad G^{(n)}_{\obs} \in \mathcal{O}_{\kappa_n, \rho_n} \subset \mathcal{O}_{2\kappa_{\min}, \rho_*}\,.
\end{equation}
\begin{corollary}\label{lem:TaroccoBound}
	Let $O$ be a local observable supported on $S_O$ and $G \in \mathcal{O}_{2\kappa_{\min},0}$, with $\Vert G \Vert_{2\kappa_{\min},0} \leq 1$, then $\exists C(|S_O|,d,\kappa, \rho)>0$ such that $\forall t>0$
	\begin{equation}
	\Vert e^{-\ii G(\omega t)} O e^{\ii G(\omega t)} - O \Vert_{\mathrm{op}} \leq C(|S_O|,d, \kappa, \rho) \Vert O \Vert_{\mathrm{op}} {\Vert G(\omega t) \Vert_{\kappa}}\,.
	\end{equation}
\end{corollary}
\begin{proof} For any fixed $t_0:=t$, we write
	\[
	\begin{split}
	\Vert e^{-\ii G(\omega t_0)} O e^{\ii G(\omega t_0)} - O \Vert_{\mathrm{op}} &\leq \int_0^1 \ud \tau \Vert [ G(\omega t_0),e^{-\ii \tau G(\omega t_0)} O e^{\ii \tau G(\omega t_0)}] \Vert_{\mathrm{op}}\,
	\end{split}
	\]
	and we apply Lemma \ref{lem:LemmaSempreDiverso} with $Z=G(\omega t_0)$ and $t=1$. This yields
	\[
	\begin{split}
	\Vert e^{-\ii G(\omega t_0)} O e^{\ii G(\omega t_0)} - O \Vert_{\mathrm{op}} &\;\leq\; C(|S_O|,d,\kappa_{\min}) \langle \Vert G(\omega t) \Vert_{2 \kappa_{\min}} \rangle^{d+1} \|O\|_{\mathrm{op}} \|G(\omega t_0)\|_{\kappa_{\min}} \, ,
	\end{split}
	\]
	which is the thesis.
\end{proof}


\begin{lemma}\label{lem.spagnolette}
	Let $H_{\obs}(\omega t) = H_0 + Z_{\obs} + V_{\obs}( \omega t)$ and $\varepsilon$ as in Proposition \ref{teo:main}, and $H_{\mathrm{eff}}$ as in \eqref{H.eff}. For any local observable $O$, there exists a constant $C = C(O, d, \kappa, \rho, \|H_0\|_{\kappa}, \|V\|_{\kappa, \rho})$ such that, for any $t \in \mathbb{R}$ we have
	\begin{equation}\label{bagigio}
	\Vert {U}^*_{H_{\obs}}(t) O {U}_{H_{\obs}}(t) - e^{-\ii H_{\mathrm{eff}} t} O e^{\ii H_{\mathrm{eff}} t} \Vert_{\mathrm{op}} \; \leq \; C {e^{-{ \varepsilon^{-\tb}}} \langle t\rangle^{d+2}}\,.
	\end{equation}
\end{lemma}

\begin{proof}
	By Proposition \ref{teo:main}, provided $\varepsilon$ is small enough, one has
	\begin{equation}\label{dall.alto.e.dal.basso}
	\frac{1}{2} \|H_0\|_{\kappa} \leq \| H_{\mathrm{eff}}\|_{2 \kappa_{\min}} \leq 2 \| H_0\|_{\kappa}\,.
	\end{equation}
	Let furthermore $W(\tau)={U}_{H_{\obs}}(\tau)^{-1} {U}_{H_{\obs}}(t)$, then by Duhamel formula we have
	\[
	{U}_{H_{\obs}}^*(t) O {U}_{H_{\obs}}(t)-e^{-\ii H_{\mathrm{eff}} t} O e^{\ii H_{\mathrm{eff}} t} = \int_0^t W^*(\tau) [{V}_{\obs}(\omega \tau), e^{-\ii \tau H_{\mathrm{eff}}} O e^{\ii \tau H_{\mathrm{eff}}}] W(\tau) \ud\tau\,.
	\]
	Taking the operator norm on both sides, one can use Lemma \ref{lem:LemmaSempreDiverso} to get
	\begin{equation}\label{eq:non.diverge.piu}
	\begin{split}
	\left\Vert\int_0^t W^*(\tau) [{V}_{\obs}(\tau), e^{-\ii \tau H_{\mathrm{eff}}} O e^{\ii \tau H_{\mathrm{eff}}}] W(\tau) \ud\tau \right\Vert_{\mathrm{op}}
	\leq C(|S_O|,d,\kappa), \rho)\cdot\\ \cdot \langle \Vert H_{\mathrm{eff}} \Vert_{2 \kappa_{\min}} \rangle^{d+1} \langle t \rangle^{d+2} \Vert V_{\obs} \Vert_{\kappa_{\min},0} \Vert O \Vert_{\mathrm{op}}\,.
	\end{split}
	\end{equation}
	Then the thesis follows combining \eqref{eq:non.diverge.piu} with $\sup_{t \in \R}\|V_{\obs}(\omega t)\|_{\kappa_{\min}} \leq e^{-{ \varepsilon^{-\tb}}} \|V\|_{\kappa, \rho}$, due to Item (b.iii) of Theorem \ref{teo:main}.
\end{proof}

\begin{lemma}\label{lemma:nemo}
	For any local observable $O$, there exist two constants $C_1 = C_1(O, d, \kappa, \rho, \| V\|_{\kappa, \rho})$ and $C_2 =( O, d, r, \kappa, \rho, \{\|N^{(\alpha)}\|_{0}\}_{\alpha = 1}^r, \| V\|_{\kappa, \rho})$ such that
	\begin{itemize}
		\item[(i)]
		\begin{equation}\label{ragazza}
		\Vert U_H^*(t) O U_H(t) - {U}_{H_\obs}^*(t) O {U}_{H_\obs}(t) \Vert_{\mathrm{op}} \leq {C_1} \varepsilon^{\frac 1 2} \, \qquad \forall t \in \mathbb{R} 
		\end{equation}
		\item[(ii)]
		\begin{equation}\label{zombie}
		\Vert U_H^*(t) O U_H(t) - {U}^*_{H_\obs}(t) O {U}_{H_\obs}(t) \Vert_{\mathrm{op}} \leq C_2 \varepsilon^{\frac 1 2} |\omega t|_{\mathbb{T}^m} \, \qquad \forall t \in \mathbb{R} \, .
		\end{equation}
	\end{itemize}
\end{lemma}


\begin{proof}
	By Item (b.iv) of Proposition \ref{teo:main}, one has
	\[
	\begin{split}
	\Vert U_H^*(t) O U_H(t) - {U}_{H_\obs}^*(t) O {U}_{H_\obs}(t) \Vert_{\mathrm{op}} &= \Vert U_H^*(t) (O-Y_\obs^*(\omega t)OY_\obs(\omega t)) U_H(t) \Vert_{\mathrm{op}} \\ 
	&= \Vert O - Y_\obs^*(\omega t) O Y_\obs(\omega t) \Vert_{\mathrm{op}} \\
	& \leq \sum_{n={1}}^{n^*} \Vert (Y_\obs^{(n)})^*(\omega t) O Y^{(n)}_\obs(\omega t) - (Y_\obs^{(n-1)})^*(\omega t) O Y_\obs^{(n-1)}(\omega t) \Vert_{\mathrm{op}} \\
	&=\sum_{n={1}}^{n^*} \Vert (Y_\obs^{(n-1)})^*(\omega t) (e^{-\ii G^{(n)}_\obs(\omega t)} O e^{\ii G^{(n)}_\obs(\omega t)}-O) Y_\obs^{(n-1)}(\omega t) \Vert_{\mathrm{op}} \\
	&=\sum_{n={1}}^{n^*} \Vert e^{-\ii G^{(n)}_\obs(\omega t)} O e^{\ii G^{(n)}_\obs(\omega t)} - O \Vert_{\mathrm{op}}\,.
	\end{split}
	\]
	We are thus left with finding a good bound for $\Vert e^{-\ii G^{(n)}_\obs(\omega t)} O e^{\ii G^{(n)}_\obs(\omega t)} - O \Vert_{\mathrm{op}}$. To this end, we use Corollary \ref{lem:TaroccoBound} to get
	\begin{equation}\label{tullio}
	\Vert e^{-\ii G^{(n)}_\obs(\omega t)} O e^{\ii G^{(n)}_\obs(\omega t)} - O  \Vert_{\mathrm{op}} \leq C(|S_O|,d, \kappa, \rho) \Vert G^{(n)}_\obs(\omega t) \Vert_{\kappa_{\min}} \Vert O \Vert_{\mathrm{op}} \, .
	\end{equation}
	Now, using \eqref{Gn.per.n}  and recalling $\kappa_{\min} \leq \kappa_n \leq \kappa$, one has 
	\begin{equation}\label{san.martino}
	\| G^{(n)}_\obs(\omega t)\|_{\kappa_{\min}} \leq \Vert G_\obs^{(n)} \Vert_{\kappa_n, \rho_n} \leq \Vert V \Vert_{\kappa, \rho} \varepsilon^{\frac 1 2} e^{-n}\,,
	\end{equation}
	whence 
	\[
	\begin{split}
	\sum_{n=0}^{n^*-1} \Vert e^{-\ii G^{(n)}_\obs(\omega t)} O e^{\ii G^{(n)}_\obs(\omega t)}-O \Vert_{\mathrm{op}} &\leq C(|S_O|,d,\kappa, \rho) \Vert O \Vert_{\mathrm{op}} \Vert V \Vert_{\kappa, \rho} \varepsilon^{\frac 1 2} \sum_{n=0}^{n^*-1} e^{-n} \\
	&\leq  C(|S_O|,d,\kappa, \rho) \Vert O \Vert_{\mathrm{op}} \frac{{e}}{e-1} \Vert V \Vert_{\kappa, \rho} \varepsilon^{\frac 1 2} \, .
	\end{split}
	\]
	This proves \eqref{ragazza}.
	Concerning \eqref{zombie}, for any $t$, one has
	\[
		\begin{split}
	\Vert G^{(n)}_\obs(\omega t) - G^{(n)}_\obs(0) \Vert_{\kappa_n} &\leq  \sup_{\varphi \in \mathbb{T}^m}\max_{j=1,\dots,m} \| \partial_{\varphi_j} G^{(n)}_\obs(\varphi) \|_{\kappa_n} |\varphi|_{\mathbb{T}^m} \\
	&\leq \frac{1}{\rho_n} \|G^{(n)}_\obs \|_{\kappa_n, \rho_n} |\varphi|_{\mathbb{T}^m} \leq \frac{1}{\rho_*} \|G^{(n)}_\obs \|_{\kappa_n, \rho_n} |\varphi|_{\mathbb{T}^m}\,.
		\end{split}
	\]
	Thus, recalling that by \eqref{def:oggetti.finali} $\rho_*= \rho_*(\kappa, \rho, r, \{\|N^{(\alpha)}\|_0\}_\alpha)$, by \eqref{san.martino} one has, for a positive constant $C(r, \kappa,\rho,\Vert V \Vert_{\kappa,\rho},\Vert N^{(\alpha)} \Vert_{0})$,
	\[
	\|G^{(n)}_\obs(\omega t) \|_{\kappa_n}=\Vert G^{(n)}_\obs(\omega t) - G^{(n)}_\obs(0) \Vert_{\kappa_n}  \leq C(r, \kappa,\rho, \Vert V \Vert_{\kappa,\rho},\Vert N^{(\alpha)} \Vert_0) \varepsilon^{\frac12}e^{-n} |\omega t|_{\mathbb{T}^m} \, .
	\]
	which, combined with \eqref{tullio}, yields \eqref{zombie}.
\end{proof}

Given $\omega \in \R^m$, $\delta >0$ and $\theta \in \mathbb{T}^m$, we define the $\delta-$ergodization time as
\begin{equation}\label{def:ergo.time}
T_{\omega, \delta, \theta} := \inf\;\left \lbrace t \in \R_+\ |\ \forall \varphi \in \mathbb{T}^m \quad \Vert \varphi - (\theta + [0, t]\omega)\Vert_{\mathbb{T}^m} \leq \delta \right \rbrace\,.
\end{equation}
Note that $T_{\omega, \delta, \theta}$ is clearly independent of $\theta$, thus from now on we will omit the subscript $\theta$ and we will simply write $T_{\omega, \delta}$. The following result is due to \cite{Berti2003, bourgain.e.altri.amici}:
\begin{theorem}[Theorem 4.1 of \cite{Berti2003}]\label{teo:berti.biasco.bolle}
	$\forall m \in \N$ $\exists a_m>0$ such that $\forall \omega \in \R^m$ satisfying \eqref{eq:siamo.dei.cani}, $\forall \delta >0$ one has
	\begin{equation}
	T_{\omega,\delta} \leq \frac{1}{\gamma} \left(\frac{a_m}{\delta}\right)^{\tau} =: \widetilde{T}_{\omega,\delta}\,.
	\end{equation}
\end{theorem}
This means that, starting from any point $\theta \in \mathbb{T}^m$, after a time $\widetilde{T}_{\omega, \delta}$ the trajectory of the dynamical system $\dot \varphi = \omega$, $\varphi(0) = \theta$, has visited a neighborhood of size $\delta$ of any point $\varphi \in \mathbb{T}^m$. In particular, we will be interested in neighborhoods of $\varphi = 0$. 

Thus, Theorem \ref{teo:berti.biasco.bolle} guarantees that there exists at least a time $t_0 \in [0, \widetilde{T}_{\omega, \delta}]$ such that $\|\omega t_0  -0\|_{\mathbb{T}^m} < \delta$. Furthermore, choosing $\theta = \omega \widetilde{T}_{\omega, \delta}$ and $\varphi = 0$, there exists a time $t_1 \in [\widetilde{T}_{\omega, \delta}, 2 \widetilde{T}_{\omega, \delta}]$ such that $\|\omega t_1\|_{\mathbb{T}^m} < \delta$. Iterating this construction, one obtains that $\forall j \in \N$ there exists a time
\begin{equation}\label{eq:sono.proprio.loro}
t_j \in [j \widetilde{T}_{\omega, \delta}, (j+1) \widetilde{T}_{\omega, \delta}] \quad \text{s.t.} \quad \|\omega t_j \|_{\mathbb{T}^m} < \delta \quad \forall j\,.
\end{equation}

\begin{proof}[Proof of Theorem \ref{thm:EvLocObs}] We define $H_{\textrm{eff}}=J\cdot N + Z_{\mathrm{eff}}$ as in \eqref{H.eff} and we observe that, by Item (b.ii) of Proposition \ref{teo:main}, \eqref{vaccataprossimavoltastozitto} is satisfied. We now prove that Items (i) and (ii) hold.
	
	\noindent(i) 
	By  \eqref{bagigio} and \eqref{ragazza}, $\forall t \in \R$ one has
	\[
	\begin{split}
	\hspace{-7pt}\Vert U_H^*(t) O U_H(t) - e^{-\ii H_{\mathrm{eff}} t} O e^{\ii H_{\mathrm{eff}} t} \Vert_{\mathrm{op}} &\leq \Vert {U}_{H_{\obs}}^*(t) O {U}_{H_\obs}(t) -e^{-\ii H_{\mathrm{eff}} t} O e^{\ii H_{\mathrm{eff}} t} \Vert_{\mathrm{op}}\\
	&+\Vert U_H^*(t) O U_H(t)- {U}_{H_{\obs}}^*(t) O {U}_{H_{\obs}}(t)\Vert_{\mathrm{op}} \\
	& \leq C_1 t^{d+2} e^{-{ \varepsilon^{-\tb}}} + C_2 \varepsilon^{\frac 1 2}\,,
	\end{split}
	\]
	where $C_1 := C_1(O, d, \kappa, \rho, \|H_0\|_{\kappa}, \|V\|_{\kappa, \rho})>0$ and $C_2 := C_2(O, d, \kappa, \rho, \|V\|_{\kappa, \rho})>0$.\\
	(ii) Given $\delta >0$, let $\{t_j\}_{j \in \N}$ be defined as in \eqref{eq:sono.proprio.loro}. By  \eqref{bagigio} and \eqref{zombie}, one has
	\[
	\begin{split}
	\hspace{-7pt}\Vert U_H^*(t_j) O U_H(t_j) - e^{-\ii H_{\mathrm{eff}} t_j} O e^{\ii H_{\mathrm{eff}} t_j} \Vert_{\mathrm{op}} &\leq \Vert {U}_{H_\obs}^*(t_j) O {U}_{H_\obs}(t_j) -e^{-\ii H_{\mathrm{eff}} t_j} O e^{\ii H_{\mathrm{eff}} t_j} \Vert_{\mathrm{op}}\\
	&+\Vert U_H^*(t_j) O U_H(t_j)- {U}_{H_\obs}^*(t_j) O {U}_{H_\obs}(t_j)\Vert_{\mathrm{op}} \\
	& \leq C_1 t_j^{d+2} e^{-{ \varepsilon^{-\mathsf{b}}}} + C_3 \varepsilon^{\frac 1 2} \delta\,
	\end{split}
	\]
		with $C_3:=C_3(O, d, r, \kappa, \rho, \|V\|_{\kappa, \rho}, \{\|N^{(\alpha)}\|_0\})>0$. 	Then in particular ${j \widetilde{T}_{\omega, \delta}} \leq t_j \leq {(j+1)\widetilde{T}_{\omega, \delta}}$ with $\widetilde{T}_{\omega, \delta}= \gamma a_{m}^\tau \delta^{-\tau}$. Thus one has 
	\[
	\begin{split}
	\Vert U_{H}^*(t_j) O U_H(t_j) - e^{-\ii H_{\mathrm{eff}} t_j} O e^{\ii H_{\mathrm{eff}} t_j} \Vert_{\mathrm{op}} &\leq C_1 \left(\gamma a_m^\tau (j+1)\right)^{d+2} \delta^{-\tau(d+2)} e^{-{\varepsilon^{-\tb}}} + C_3 \varepsilon^{\frac 12}\delta\,.
	\end{split}
	\]
	Now take
	\[
	\delta= e^{-{\mathsf{f}} \varepsilon^{-\mathsf{b}}} \varepsilon^{-\frac{\mathsf{f}}{2}}\, 
	\]
	with $\mathsf{f}$ defined in \eqref{eq:EquazioneCarina}.
	One has
	$$
	\Vert U_{H}^*(t_j) O U_H(t_j) - e^{-\ii H_{\mathrm{eff}} t_j} O e^{\ii H_{\mathrm{eff}} t_j} \Vert_{\mathrm{op}} \leq \left(C_1 \left(\gamma a_m^\tau(j+1)\right)^{d+1} + C_3\right) \varepsilon^{\frac{1-\mathsf{f}}{2}} e^{-{\mathsf{f}} \varepsilon^{-\mathsf{b}}}\,,
	$$
	which gives the thesis.
\end{proof}

\subsection{Fast-Forcing Perturbations}

The proof of Theorem \ref{thm:SlowHeatingFF} goes exactly as in the case of small perturbations (see the proof of Theorem \ref{teo:SlowHeating}). Concerning the proof of Theorem \ref{thm:EvLocObs.ff}, one uses the same arguments that are used to prove Theorem \ref{thm:EvLocObs} in the case of small perturbations, except for the fact that Lemmas \ref{lem.spagnolette}, \ref{lemma:nemo} are substituted by the following:

\begin{lemma}\label{lem.nome.castissimo}
	Let $H_\obs(\lambda \omega t) = J \cdot N + Z_\obs + V_\obs(\lambda \omega t)$ be defined as in Proposition \ref{teo:main}. There exists $\lambda_0 = \lambda_0(\|H_0\|_{\kappa}, \|V\|_{\kappa, \rho}, r)>0$ such that $\forall \lambda < \lambda_0$ and for any local observable $O$, there exists a constant $C = C(O, d, \kappa, \rho, \|H_0\|_{\kappa}, \|V\|_{\kappa, \rho})$ such that, for any $t \in \mathbb{R}$ we have
	\begin{equation}\label{bagigio.ff}
	\Vert {U}^*_{H_\obs}(t) O {U}_{H_\obs}(t) - e^{-\ii H_{\mathrm{eff}} t} O e^{\ii H_{\mathrm{eff}} t} \Vert_{\mathrm{op}} \; \leq \; C {e^{-{ \lambda^\beta}} \langle t\rangle^{d+2}}\,.
	\end{equation}
\end{lemma}

\begin{lemma}
	For any local observable $O$, there exist two constants $C_1 := C_1(O, d, \kappa, \rho, \| V\|_{\kappa, \rho})>0$ and $C_2 :=(O, d, r, \kappa, \rho, \{\|N^{(\alpha)}\|_{0}\}_{\alpha = 1}^r, \| V\|_{\kappa, \rho})>0$ such that
	\begin{itemize}
		\item[(i)]
		\begin{equation}\label{ragazza.ff}
		\Vert U_H^*(t) O U_H(t) - {U}_{H_\obs}^*(t) O {U}_{H_\obs}(t) \Vert_{\mathrm{op}} \leq {C_1} \lambda^{-\frac{1}{8\tau_\omega}} \, \qquad \forall t \in \mathbb{R} 
		\end{equation}
		\item[(ii)]
		\begin{equation}\label{zombie.ff}
		\Vert U_H^*(t) O U_H(t) - {U}^*_{H_\obs}(t) O {U}_{H_\obs}(t) \Vert_{\mathrm{op}} \leq C_2 \lambda^{-\frac{1}{8 \tau_\omega}+1} |\omega t|_{\mathbb{T}^m} \, \qquad \forall t \in \mathbb{R} \, .
		\end{equation}
	\end{itemize}
\end{lemma}

Note the presence of a factor $\lambda$ in \eqref{zombie.ff} due to the fast forcing.

\begin{proof}[Proof of Theorem \ref{thm:EvLocObs.ff}] The proof of Item (i) goes as the one of Item (i) of Theorem \ref{thm:EvLocObs}. We now prove Item (ii).\\
	Given $\delta >0$, let $\{t_j\}_{j \in \N}$ be defined as in \eqref{eq:sono.proprio.loro}. By  \eqref{bagigio.ff} and \eqref{zombie.ff}, one has
	\[
	\begin{split}
	\hspace{-7pt}\Vert U_H^*(t_j) O U_H(t_j) - e^{-\ii H_{\mathrm{eff}} t_j} O e^{\ii H_{\mathrm{eff}} t_j} \Vert_{\mathrm{op}} &\leq \Vert {U}_{H_\obs}^*(t_j) O {U}_{H_\obs}(t_j) -e^{-\ii H_{\mathrm{eff}} t_j} O e^{\ii H_{\mathrm{eff}} t_j} \Vert_{\mathrm{op}}\\
	&+\Vert U_H^*(t_j) O U_H(t_j)- {U}_{H_\obs}^*(t_j) O {U}_{H_\obs}(t_j)\Vert_{\mathrm{op}} \\
	& \leq C_1 t_j^{d+2} e^{-{\lambda^\beta}} + C_2 \lambda^{-\frac{1}{8\tau_\omega}+1} \delta\,,
	\end{split}
	\]
	with $C_1:= C_1(O, d, \kappa, \rho, \|H_0\|_{\kappa}, \|V\|_{\kappa, \rho})>0$ and $C_2 := C_2(O, d, r, \kappa, \rho, \{\|N^{(\alpha)}\|_0\}_\alpha, \|V\|_{\kappa, \rho})>0$.
	Then in particular ${j \widetilde{T}_{\omega, \delta}} \leq t_j \leq {(j+1)\widetilde{T}_{\omega, \delta}}$ with $\widetilde{T}_{\omega, \delta}= \gamma a_{m}^{\tau_\omega} \delta^{-\tau_\omega}$. Thus one has 
	\[
	\begin{split}
	\Vert U_{H}^*(t_j) O U_H(t_j) - e^{-\ii H_{\mathrm{eff}} t_j} O e^{\ii H_{\mathrm{eff}} t_j} \Vert_{\mathrm{op}} &\leq C_1 \left(\gamma a_m^{\tau_\omega} (j+1)\right)^{d+2} \delta^{-\tau_\omega(d+2)} e^{-{ \lambda^\beta}} + C_2 \lambda^{1-\frac{1}{8\tau_\omega}}\delta\,.
	\end{split}
	\]
	Then the thesis follows taking
	\[
	\delta:=e^{-{\mathsf{g}} \lambda^\beta} \lambda^{\frac{(1-8\tau_\omega)\mathsf{g}}{8 \tau_\omega}}\,,
	\]
	with $\mathsf{g}$ defined as in \eqref{abbandono.di.dal.maso}.
\end{proof}

\appendix

\section{Proof of Lemmas \ref{cor.lignano.pineta} and \ref{lem:LemmaSempreDiverso}}\label{append:ProofTec}

\subsection{Proof of Lemma \ref{cor.lignano.pineta}}
We are going to control the {locality} properties of $e^{\ii G(\nu t)} A (\nu t) e^{-\ii G(\nu t)}$ for $A \in \mathcal{O}^\strong_{\kappa, \rho}$ and $G \in \mathcal{O}^\strong_{\kappa, \rho}$, using the commutator expansion \eqref{eq:ConjugationAd}.
An important part of this procedure is to control the locality property of the operators $\textrm{Ad}^q_{G} A$. Indeed, considering two operators $A$ and $B$ supported on $S_A$ and $S_B$ respectively, from \eqref{eq:LocalCommutator} one sees that if $S_A \cap S_B \neq \varnothing$, then the support of $[A,B]$ is $S_A \cup S_B$. Taking several commutators, in fact, enlarges the support of each local term of $\mathrm{Ad}^q_A B$. 
This phenomenon is what is responsible for the loss of regularity in the parameter $\kappa$ in Lemma \ref{cor.lignano.pineta}.
To control the smearing out of the supports of the $\mathrm{Ad}^q_A B$'s, we use the following Lemma 
\begin{lemma}\label{lemma.fs}$S_1,\dots,S_q \in \mathcal{P}_c(\Lambda)$ and let $f(S_1,\dots,S_q)$ be a positive function. Then
	\begin{equation}
	\begin{split}
	\sup_{x \in \Lambda}&\sum_{\substack{S_1,\dots,S_q \in \mathcal{P}_c(\Lambda)\\\text{s.t. } 
			x \in \bigcup_{j=1}^q S_j \\
			S_j \cap (\bigcup_{r=j+1}^q S_r)
			\neq \varnothing \, \forall j=1,\dots,q-1}} f(S_1,\dots,S_q) \\&\leq \max_{p \in \mathfrak{S}_q} \sup_{x_{p(1)} \in \Lambda} \sum_{\substack{S_{1} \, s.t. \\ x_{p(1)} \in S_1}} \dots \sup_{x_{p(q)} \in \Lambda} \sum_{\substack{S_q \, s.t. \\ x_{p(q)} \in S_q}} 2^{q-1} \left(|S_1|+\dots + |S_q| \right)^{q-1} f(S_1,\dots,S_q)
	\end{split}
	\end{equation}
	where $\mathfrak{S}_q$ denotes the set of permutation of $q$ elements.
\end{lemma}
\begin{proof}[Proof of Lemma \ref{lemma.fs}]
	To avoid  heavy notation, it is understood that each of the $S_j$ is such that $S_j \in \mathcal{P}_c(\Lambda)$. The proof goes by induction on the number of sets. Precisely, we prove the inductive estimate
	\[
	\sum_{\substack{S_j,\dots,S_{q}\\x \in \bigcup_{r=j}^q S_r \\
			S_r \cap \left(\bigcup_{l=r+1}^q S_l \right) \neq \varnothing \\
			\forall r=j, \dotsm q-1}} f(S_1,\dots,S_q) \; \leq \; 2^{q-j} \max_{p \in \mathfrak{S}_{[j,q]}} \sup_{x_{p(j)} \in \Lambda} \sum_{\substack{S_j\, s.t. \\ S_j \ni x_{p(j)}}} \dots \sup_{x_{p(q)} \in \Lambda} \sum_{\substack{S_q\, s.t. \\ S_q \ni x_{p(q)}}} (|S_j|+\dots+|S_q|)^{q-j} f(S_1,\dots,S_q) 
	\]
	where $\mathfrak{S}_{[j,q]}$ denotes the set of permutations of $\{j,\dots,q\}$. First we consider the case with only $S_{q-1}$ and $S_q$. We have to prove that
	\[
	\sup_{x \in \Lambda} \sum_{\substack{S_{q-1},S_q \text{ s.t.}  \\ x \in S_{q-1} \cup S_q \\
			S_{q-1} \cap S_q \neq \varnothing}} f(S_{q-1},S_q) \leq \max_{p \in \mathfrak{S}_{[q-1,q]}} \sup_{x_{p(q-1)} \in \Lambda} \sum_{\substack{S_{q-1} \, \text{s.t.} \\ x_{p(q-1)} \in S_{q-1}}} \sup_{x_{p(q)} \in \Lambda} \sum_{\substack{S_q \, \text{s.t.} \\ x_{p(q)} \in S_q}} 2 (|S_{q-1}|+|S_q|) f(S_{q-1},S_q) \, . 
	\]
	Starting from l.h.s., we notice that $x \in S_{q-1} \cup S_q$ means that either $x \in S_{q-1}$ or $x \in S_q$. Thus we can bound l.h.s. as
	\[
	\begin{split}
	\sup_{x \in \Lambda} \sum_{\substack{S_{q-1},S_q \text{ s.t.} \\ \, x \in S_{q-1} \cup S_q \, , \\
			S_{q-1} \cap S_q \neq \varnothing}} f(S_{q-1},S_q)& \leq \sup_{x \in \Lambda} \Bigg[\sum_{\substack{S_{q-1},S_q \text{ s.t.} \\ x \in S_{q-1}  \, ,\\ S_{q-1} \cap S_q \neq \varnothing}} f(S_{q-1},S_q)+ \sum_{\substack{S_q,S_{q-1} \text{ s.t.}\\  \, x \in S_q ,\\ S_{q-1} \cap S_q \neq \varnothing \, .}} f(S_{q-1},S_q)\Bigg]\;=:\; (\star)
	\end{split}
	\] 
	Fix $S_q$, to bound the sums it is convenient to define
	\[
	\begin{split}
	\mathscr{A}(S_q) \;&:=\; \big\{ S_{q-1} \, \big| \, \exists y \in S_q \cap S_{q-1} \big\} \subseteq \; \bigcup_{y \in S_q} \big\{ S_{q-1} \, \big| \, y \in S_{q-1} \big\} =\;\bigcup_{y \in S_q} \mathscr{A}_y(S_q)
	\end{split}
	\]
	which permits us to get
	\[
	\begin{split}
	(\star) \; &\leq \; \sup_{x \in \Lambda} \Bigg[ \sum_{\substack{S_q \text{ s.t.}  \\ \, x \in S_q}} \sum_{y \in S_q} \sum_{S_{q-1} \in \mathscr{A}_y(S_q)} f(S_{q-1},S_q)+ \sum_{\substack{S_{q-1}  \text{ s.t.}\\  \, x \in S_{q-1}}} \sum_{y \in S_{q-1}} \sum_{S_q \in \mathscr{A}_y(S_{q-1})} f(S_{q-1},S_q) \Bigg] \\
	&\leq \; \sup_{x \in \Lambda} \Bigg[ \sum_{\substack{S_q \text{ s.t.} \\ x \in S_q}} |S_q| \sup_{y \in S_q} \sum_{S_{q-1} \in \mathscr{A}_y(S_q)} f(S_{q-1},S_q)+ \sum_{\substack{S_{q-1} \text{ s.t.} \\ x \in S_{q-1}}} |S_{q-1}| \sup_{y \in S_{q-1}} \sum_{S_q \in \mathscr{A}_y(S_{q-1})} f(S_{q-1},S_q) \Bigg] \\
	&\leq \; \sup_{x \in \Lambda} \Bigg[ \sum_{\substack{S_q \text{ s.t.} \\  \, x \in S_q}} (|S_q|+|S_{q-1}|) \sup_{y \in S_q} \sum_{S_{q-1} \in \mathscr{A}_y(S_q)} f(S_{q-1},S_q)\\
	& \qquad +\sum_{\substack{S_{q-1} \text{ s.t.} \\  x \in S_{q-1}}} (|S_q|+|S_{q-1}|) \sup_{y \in S_{q-1}} \sum_{S_q \in \mathscr{A}_y(S_{q-1})} f(S_{q-1},S_q) \Bigg] \\
	&\leq\; \sup_{x \in \Lambda} \Bigg[ \Bigg( \sum_{\substack{S_q \text{ s.t.}\\  x \in S_q}} \sup_{y \in \Lambda} \sum_{S_{q-1} \in \mathscr{A}_y(S_{q-1})}+ \sum_{\substack{S_{q-1}\text{ s.t.} \\ x \in S_q}} \sup_{y \in \Lambda} \sum_{S_{q} \in \mathscr{A}_y(S_{q})} \Bigg) (|S_q|+|S_{q-1}|)  f(S_{q-1},S_q) \Bigg] \\
	&\leq \max_{p \in \mathfrak{S}_2} \sup_{x_{p(1)} \in \Lambda} \sum_{\substack{S_q \text{ s.t.}  \\ x_{p(1)} \in S_q}} \sup_{x_{p(2)} \in \Lambda} \sum_{\substack{S_{q-1} \text{ s.t.}\\ x_{p(2)} \in S_{q-1}}} 2 (|S_q|+|S_{q-1}|) f(S_{q-1},S_q)\,.
	\end{split}
	\]
	This proves the thesis for the function of $S_{q-1},S_q$. Now we proceed iteratively: we suppose that the thesis is true for $S_{j+1},\dots,S_q$ and we prove it is true also for $S_j,\dots,S_q$.
	
	\[
	\begin{split}
	\sum_{\substack{S_j,\dots,S_{q} \text{ s.t.}\\x \in \bigcup_{r=j}^q S_r, \\
			S_r \cap \left(\bigcup_{l=r+1}^q S_l \right) \neq \varnothing \\
			\forall r=j, \dotsm q-1}} f(S_1,\dots,S_q) & \leq \; \sup_{x \in \Lambda} \Bigg\{ \sum_{\substack{S_j \text{ s.t.} \\ x \in S_j}} \sum_{\substack{S_{j+1},\dots,S_q  \text{ s.t.}\\  S_r \cap (\bigcup_{l=r+1}^q S_l) \neq \varnothing, \\ \forall r=j+1,\dots, q-1 \\ S_j \cap \left(\bigcup_{l=j+1}^qS_l \right) \neq \varnothing}} f(S_j,\dots,S_q) \\
	&\qquad +\sum_{\substack{S_{j+1},\dots,S_q  \text{ s.t.}\\  x \in \bigcup_{r={j+1}}^q S_r, \\ S_r \cap (\bigcup_{l={r+1}}^q S_l) \neq \varnothing \\ \forall r=j+1,\dots,q-1}} \sum_{\substack{S_j \text{ s.t.} \\ S_j \cap (\bigcup_{r=j+1}^q S_r) \neq \varnothing}} f(S_j,\dots,S_q) \Bigg\} \\
	=: (\blacksquare)+(\clubsuit) \, .
	\end{split}
	\]
	We have now to estimate the two terms separately. For the first one, let us fix $S_j$ and consider
	\[
	\begin{split}
	\mathscr{A}(S_j) \;&:=\;\Big\{(S_{j+1},\dots,S_q) \, \Big| \, S_j \cap \Big(\bigcup_{r=j+1}^q S_r \Big) \neq \varnothing \Big\} \subseteq\bigcup_{y \in S_j} \Big\{\bigcup_{r={j+1}}^q S_r \, \Big| \, \bigcup_{r={j+1}}^q S_r \ni y \Big\} = \bigcup_{y \in S_j} \mathscr{A}_y(S_j) \, .
	\end{split}
	\]
	For the second one, fix $S_{j+1},\dots,S_q$ and consider
	\[
	\begin{split}
	\mathscr{A}'(S_{j+1},\dots,S_q)\;&=\; \Big\{S_j \,\Big|\, S_j \cap\Big(\bigcup_{r=j+1}^q S_r\Big) \neq \varnothing\Big\} \subseteq \bigcup_{y \in \bigcup_{r={j+1}}^q S_r} \left\{S_j \, | \, y \in S_j \right\} \;=\; \mathscr{A}'_y(S_{j+1},\dots,S_q) \, .
	\end{split}
	\]
	Thus,
	\[
	\begin{split}
	(\blacksquare) & \leq \sup_{x \in \Lambda} \sum_{\substack{S_j \text{ s.t.}\\   x \in S_j}}\, \sum_{y \in S_j} \sum_{\substack{S_{j+1},\dots,S_q \text{ s.t.} \\ y \in \bigcup_{r=j+1}^q S_r, \\ S_r \cap(\bigcup_{l={r+1}}^q S_l) \neq \varnothing, \\
			\forall r=j+2,\dots,q-1}} f(S_j,\dots,S_q) \leq\; \sup_{x \in \Lambda} \sum_{\substack{S_j \text{ s.t.} \\  x \in S_j}} |S_j| \sup_{y \in \Lambda} \sum_{\substack{S_{j+1},\dots,S_q \text{ s.t.}\\  y \in \bigcup_{r=j+1}^q S_r, \\ S_r \cap(\bigcup_{l={r+1}}^q S_l) \neq \varnothing \\
			\forall r=j+2,\dots,q-1}} f(S_j,\dots,S_q) \\
	&\!\!\!\!\!\!\!\!\!\!\overset{\text{inductive hyp}}{\leq} \sup_{x \in \Lambda} \sum_{\substack{S_j \text{ s.t.} \\ x \in S_j}} |S_j| 2^{q-j-1} \max_{p \in \mathfrak{S}_{[j+1, q]}} \sup_{x_{p(j+1)} \in \Lambda} \sum_{\substack{S_{j+1} \text{ s.t.} \\  x_{p(j+1)} \in S_{j+1}}} \dots \\
	& \qquad \qquad \qquad \dots \sup_{x_{p(q)} \in \Lambda} \sum_{\substack{S_q \text{ s.t.}\\ x_{p(q)} \in S_q}} (|S_{j+1}|+ \dots + |S_q|)^{q-j-1} f(S_j,\dots,S_q) \\
	&\leq 2^{q-j-1} \sup_{x \in \Lambda} \sum_{\substack{S_j \text{ s.t.} \\ x \in  S_j}} \max_{p \in \mathfrak{S}_{[j+1,q}}\sup_{x_{p(j+1)} \in \Lambda} \sum_{\substack{S_{j+1} \text{s.t. } x_{p(j+1)} \in S_{j+1}}} \dots \\&\qquad \dots \qquad \sup_{x_{p(q)}\in \Lambda} \sum_{\substack{S_q\text{ s.t.} \\  x_{p(q)} \in S_q}}(|S_j|+\dots+|S_q|)^{q-j} f(S_j,\dots,S_q) \, \\
	&\leq 2^{q-j-1} \max_{p \in \mathfrak{S}_{[j,q]}} \sup_{x_{p(j)} \in \Lambda} \sum_{\substack{S_j   \text{ s.t.}\\ x_{p(j)} \in S_j }} \dots \sup_{x_{p(q)} \in \Lambda} \sum_{\substack{S_q \text{ s.t.}\\ x_{p(q)} \in S_q }} (|S_j|+\dots+|S_q|)^{q-j} f(S_j,\dots,S_q)\,.
	\end{split}
	\]
	Concerning the second piece,
	\[
	\begin{split}
	(\clubsuit) & \leq \sup_{x \in \Lambda} \sum_{\substack{S_{j+1},\dots,S_q \text{ s.t.}\\ x \in S_{j+1} \cup \dots \cup S_q , \\
			S_r \cap (\bigcup_{l=r+1}^q S_l) \neq \varnothing \\
			\forall r={j+1},\dots,q-1}} \sum_{y \in S_{j+1} \cup \dots \cup S_q} \sum_{\substack{S_j \text{ s.t.} \\ y \in S_j}} f(S_j,\dots,S_q) \\
	& \leq \sup_{x \in \Lambda} \sum_{\substack{S_{j+1},\dots,S_q  \\ x \in S_{j+1} \cup \dots \cup S_q \\
			S_r \cap (\bigcup_{l=r+1}^q S_l) \neq \varnothing \\
			\forall r={j+1},\dots,q-1}} (|S_{j+1}|+\dots+|S_q|) \sup_{y \in \Lambda}\sum_{\substack{S_j \text{ s.t.} \\  S_j \ni y}} f(S_j,\dots,S_q) \\
	&{\leq} \max_{p \in \mathfrak{S}_{[j+1,q]}} \sup_{x_{p(j+1)} \in \Lambda} \hspace{-7pt} \sum_{\substack{S_{j+1} \text{ s.t.} \\ x_{p(j+1)} \in S_{j+1}}}\hspace{-20pt} \dots \sup_{x_{p(q)} \in \Lambda} \hspace{-5pt} \sum_{\substack{S_q \text{ s.t.} \\ x_{p(q)} \in S_q }}2^{q-j-1}(|S_{j+1}|+ \dots + |S_q|)^{q-j-1} \sup_{y \in \Lambda} \sum_{\substack{S_j \text{ s.t.} \\ y \in S_j}}f(S_j,\dots,S_q) 
	\end{split}
	\]
	\[
	\begin{split}
	&\leq \max_{p \in \mathfrak{S}_{[j,q]}} 2^{q-j-1} \sup_{x_{p(j)} \in \Lambda} \sum_{\substack{S_j  \text {s.t.} \\ x_{p(j)} \in S_j }} \dots \sup_{x_{p(q)} \in \Lambda}\sum_{\substack{S_q  \text{ s.t.}\\ x_{p(q)} \in S_q  }} 2^{q-j-1} (|S_{j}|+\dots+|S_q|)^{q-j} f(S_j,\dots,S_q) \, .
	\end{split}
	\]
	At this point, combining the estimates for $(\clubsuit)$ and $(\blacksquare)$ we get 
	\[
	(\blacksquare)+(\clubsuit) \leq \max_{p \in \mathfrak{S}_{[j,q]}} \sup_{x_{p(j)} \in \Lambda} \sum_{\substack{S_{j}\, s.t. \\ S_j \ni  x_{p(j)} }} \dots \sup_{x_{p(q)} \in \Lambda} \sum_{\substack{S_q\, s.t. \\ S_q \ni x_{p(q)} }} 2^{q-j} \left(|S_j|+\dots + |S_q| \right)^{q-j} f(S_j,\dots,S_q)
	\]
	proving the inductive thesis and thus the technical Lemma.
\end{proof}

To understand quantitatively how to control the norm of $\mathrm{Ad}_A^q B$, let us prove first the case of a single commutator.

\begin{lemma}\label{lem.2k}
	Let $\kappa, \rho >0$. Then $\forall \sigma>0$ and $\forall A, B \in \mathcal{O}^\strong_{\kappa + \sigma, \rho}$ one has $[A, B] \in \mathcal{O}^\strong_{\kappa, \rho}$, with
	\begin{equation}\label{comm.estimate}
	\Vert [A, B]\Vert_{\kappa, \rho} \leq \frac{4 e^{-(\kappa+1)}}{\sigma} \| A\|_{\kappa + \sigma, \rho} \| B\|_{\kappa + \sigma, \rho}\,.
	\end{equation}
\end{lemma}
\begin{proof} Using the definition of local parts of a commutator given by \eqref{eq:LocalCommutator} and the fact that $A,B \in \mathcal{O}_{\kappa+\sigma,\rho}$, one has
	\begin{align*}
	\Vert [A, B] \Vert_{\kappa, \rho} &\;=\; \sup_{x \in \Lambda} \sum_{S \ni x} \sum_{l \in \mathbb{Z}^m} e^{\kappa |S|} e^{\rho|l|} \Vert (\widehat{[A, B]_S})_{l} \Vert_{\mathrm{op}}\\
	&\leq  2 \sup_{x \in \Lambda} \sum_{S \ni x} \sum_{S_1, S_2 : S_1 \cup S_2 = S\atop S_1 \cap S_2 \neq \emptyset}  \sum_{k \in \mathbb{Z}^m} \sum_{l' \in \mathbb{Z}^m} e^{\kappa |S|} e^{\rho|l-l'|} \Vert(\widehat{A_{S_1}})_{l-l'}\Vert_{\mathrm{op}} \Vert(\widehat{B_{S_2}})_{l'}\Vert_{\mathrm{op}}\\
	&\leq  2 e^{-\kappa} \sup_{x \in \Lambda} \sum_{S_1, S_2 : S_1 \cup S_2 \ni x \atop S_1 \cap S_2 \neq \emptyset} e^{\kappa \left( |S_1| + |S_2|\right)} F_{S_1, S_2}(A,B)\,,
	\end{align*}
	with
	\begin{align*}
	F_{S_1, S_2}(A, B) &:= \sum_{l\in \mathbb{Z}^m} \sum_{l' \in \mathbb{Z}^m} e^{\rho|l-l'|} e^{\rho|l'|} \Vert(\widehat{A_{S_1}})_{l - l'}\Vert_{\mathrm{op}} \Vert(\widehat{B_{S_2}})_{l'}\Vert_{\mathrm{op}} \\
	& = \sum_{l' \in \mathbb{Z}^m}  e^{\rho|l'|} \Vert(\widehat{B_{S_2}})_{l'}\Vert_{\mathrm{op}} \sum_{l \in \mathbb{Z}^m} e^{\rho|l|} \Vert(\widehat{A_{S_1}})_{l}\Vert_{\mathrm{op}}\,.
	\end{align*}
	By Lemma \ref{lemma.fs} with $q = 1$, one has
	\begin{align*}
	\Vert [A, B] \Vert_{\kappa,\rho} & \leq 4 e^{-\kappa} \max_{p \in \mathfrak{S}_2} \sup_{x_{p(1)} \in \Lambda} \sum_{S_1 \ni x_{p(1)}} \sup_{x_{p(2)} \in \Lambda} \sum_{S_2 \ni x_{p(2)}} e^{\kappa (|S_1| + |S_2|)} (|S_1| + |S_2|) F_{S_1, S_2}(A, B)\\
	&\leq 4 e^{-(\kappa + 1)} \sigma^{-1} \max_{p \in \mathfrak{S}_2} \sup_{x_{p(1)} \in \Lambda} \sum_{S_1 \ni x_{p(1)}} \sup_{x_{p(2)} \in \Lambda} \sum_{S_2 \ni x_{p(2)}} e^{(\kappa + \sigma) (|S_1| + |S_2|)} F_{S_1, S_2}(A, B)\\
	& = 4 e^{-(\kappa + 1)} \sigma^{-1}  \Vert A\Vert_{\kappa + \sigma, \rho} \Vert B\Vert_{\kappa + \sigma, \rho}\,,
	\end{align*}
	where we have used the fact that $	\sup_{x \in \R} x e^{-\sigma x} \leq \frac{1}{e}$.
	This proves \eqref{comm.estimate}.
	Using that $A, B \in \mathcal{O}^\strong_{\kappa + \sigma, \rho}$ and that $S'\nsubseteq S$, by \eqref{eq:LocalCommutator}
	we deduce that
	$$
	\begin{aligned}
	{[[A, B]_S, N^{(\alpha)}_{S'}]} &= \hspace{-10pt}\sum_{S_1, S_2 : S_1 \cup S_2 = S\atop S_1 \cap S_2 \neq \emptyset} \hspace{-10pt} [[A_{S_1}, B_{S_2}], N^{(\alpha)}_{S'}]\\
	&= -\hspace{-10pt} \sum_{S_1, S_2 : S_1 \cup S_2 = S\atop S_1 \cap S_2 \neq \emptyset} \hspace{-10pt} [[B_{S_2}, N^{(\alpha)}_{S'}], A_{S_1}]  - \hspace{-10pt}\sum_{S_1, S_2 : S_1 \cup S_2 = S\atop S_1 \cap S_2 \neq \emptyset}\hspace{-10pt} [[N^{(\alpha)}_{S'}, A_{S_1}], B_{S_2}] \, .
	\end{aligned}
	$$
	Now we note that, if $S' \nsubseteq S,$ $\forall S_1, S_2$ such that $S = S_1 \cup S_2$ then it must be that both $S' \nsubseteq S_1$, and $S' \nsubseteq S_2$. But then, by definition of strongly local operator
	$[B_{S_2}, N^{(\alpha)}_{S'}] = 0$ and $[N^{(\alpha)}_{S'}, A_{S_1}] = 0$, 
	which proves $	[[A, B]_S, N^{(\alpha)}_{S'}] = 0$, namely $[A,B] \in \mathcal{O}^\strong_{\kappa, \rho}$.
\end{proof}

\begin{proof}[Proof of Lemma \ref{cor.lignano.pineta}]
	We use the expansion $e^A B e^{-A}=\sum_{q\geq 0} \frac{1}{q!} \mathrm{Ad}_A^qB$ and we prove that $\mathrm{Ad}^q_A B \in \mathcal{O}^\strong_{\kappa + \sigma, \rho}$ with
	\begin{equation}\label{ad.A.B.q}
	\| \mathrm{Ad}_A^q B \|_{\kappa, \rho} \leq \left(\frac q e\right)^q \frac{(4 e^{-\kappa})^q}{\sigma^q} \| A\|^q_{\kappa + \sigma, \rho} \| B\|_{\kappa + \sigma, \rho}\,.
	\end{equation}
	We start with observing that
	\begin{equation}\label{gli.acchi}
	\begin{split}
	\mathrm{Ad}^q_{A(\varphi)} B(\varphi) &= \sum_{S \in \mathcal{P}_c(\Lambda)} (\mathrm{Ad}^q_{A(\varphi)}B(\varphi))_S\,, \\ 
	(\mathrm{Ad}^q_{A(\varphi)} B(\varphi))_S &= \sum_{\substack{S_1,\dots,S_{q+1} \, s.t. \\ \bigcup_{j=1}^{q+1} S_j = S \\ S_r \cap (\bigcup_{r+1}^{q+1} S_\ell) \neq \varnothing \\ \forall r=1, \dots, q}} [{A_{S_1}}(\varphi),[{A_{S_2}}(\varphi),[\dots,[{A_{S_q}}(\varphi),{B_{S_{q+1}}}(\varphi)] \dots]\,.
	\end{split}
	\end{equation}
	Arguing as in Lemma \ref{lem.2k}, one obtains
	\begin{align*}
	\Vert \textnormal{Ad}_A^q(B)\Vert_{\kappa, p} & \leq (2 e^{-\kappa})^{q} \sup_{x \in \Lambda} \sum_{\substack{S_1 \cup \dots \cup S_{q+1} \ni x \\
			S_r \cap (\bigcup_{s={r+1}}^{q+1} S_s) \neq \varnothing \\ \forall r=1,\dots,q}} e^{(|S_1| + \dots + |S_{q+1}|)\kappa}  F_{S_1, \dots, S_{q+1}} (A, B)\,,
	\end{align*}
	where we have defined
	$$
	\begin{aligned}
	F_{S_1, \dots, S_{q+1}} (A, B) &:= \sum_{l_1, \dots, l_{q+1} \in\mathbb{Z}^m} e^{\rho|l_1|} \Big(\prod_{j=1}^{q} \big\Vert (\widehat{A_{S_j}})_{l_j - l_{j+1}}\big\Vert_{\mathrm{op}} \Big) \big\Vert (\widehat{B_{S_{q+1}}})_{l_{q+1}}\big\Vert_{\mathrm{op}}\\
	&\leq \Big(\prod_{j = 1}^{q} \sum_{l_j \in \mathbb{Z}^m} e^{\rho|l_j|} \big\Vert (\widehat{A_{S_j}})_{l_j}\big\Vert_{\mathrm{op}}\Big) \sum_{l_{q+1} \in \mathbb{Z}^m} e^{\rho| l_{q+1}|} \big\Vert(\widehat{B_{S_{q+1}}})_{l_{q+1}}\big\Vert_{\mathrm{op}}\,.
	\end{aligned}
	$$
	Then by Lemma \ref{lemma.fs} one obtains
	\begin{multline*}
	\Vert \textnormal{Ad}^q_A B \Vert_{\kappa,\rho} \leq  (4 e^{-\kappa})^q \max_{p \in \mathfrak{S}_{q+1}} \sup_{x_{p(1)} \in \Lambda} \sum_{S_1 	\ni  x_{p(1)}}  \dots \sup_{x_{p(q+1)} \in \Lambda} \sum_{S_{q+1} \ni x_{p(q+1)}} e^{\kappa(|S_1| + \dots |S_{q+1}|)} \cdot \\ \cdot (|S_1| + \dots + |S_{q+1}|)^{q} F_{S_1, \dots, S_{q+1}}(A, B)\,,
	\end{multline*}
	which by Cauchy estimates implies \eqref{ad.A.B.q}.
	Arguing analogously as in the proof of Lemma \ref{lem.2k}, one also proves $ \mathrm{Ad}_A^q B \in \mathcal{O}^\strong_{\kappa, \rho}$.
	Summing over $q$, \eqref{ad.A.B.q} gives
	\begin{align*}
	\Big\Vert \sum_{q \geq 1} \frac{1}{q!}\textnormal{Ad}^q_A B \Big\Vert_{\kappa,\rho} &\leq \sum_{q \geq 1} \frac{1}{q!} \left(\frac{q}{e} \right)^q \left(\frac{4 e^{-\kappa} \|A\|_{\kappa + \sigma,\rho}}{\sigma}\right)^q \|B\|_{\kappa + \sigma, \rho}\\
	&= \frac{4 e^{-\kappa} \|A\|_{\kappa + \sigma, \rho}}{\sigma} \sum_{q \geq 0} \frac{1}{(q+1)!} \left(\frac{q+1}{e} \right)^{q+1} \left(\frac{4 e^{-\kappa} \|A\|_{\kappa + \sigma, \rho}}{\sigma}\right)^q \|B\|_{\kappa + \sigma, \rho}\,.
	\end{align*}
	Since, by Stirling formula $q! \geq	\left(\frac{q}{e}\right)^q \sqrt{2\pi q}$, one has
	$
	\frac{1}{(q+1)!} \left(\frac{q+1}{e}\right)^{q+1} \leq 1\,,
	$
	and the series converges due to the condition $4 e^{-\kappa} \|A\|_{\kappa + \sigma, \rho}<\sigma$,	this proves \eqref{anna}. Equation \eqref{ha.ragione} is proved in a similar way.
\end{proof}

\subsection{Proof of Lemma \ref{lem:LemmaSempreDiverso}}

	First, we use $A(\omega t)=\sum_{S \in \mathcal{P}_c(\Lambda)} A_S(\omega t)$ and triangular inequality to have
	\[
	\Vert [A(\omega t), e^{-\ii \tau Z} O e^{\ii \tau Z}] \Vert_{\mathrm{op}} \leq \sum_{x \in \Lambda} \sum_{\substack{S \in \mathcal{P}_c(\Lambda)\\ {x \in S}}} \Vert [A_S(\omega t), e^{-\ii \tau Z} O e^{\ii \tau Z}] \Vert_{\mathrm{op}} =: (\star) \, .
	\]
	We now define $Q_{S_O}$ as the smallest ball that contains $S_O$; we call $r_Q$ its radius. Then we construct $B_{S_O}$ which is the ball of radius $r_{Q}+v\tau$ and the same center as $Q_{S_O}$. Then by Lemma \ref{lieb.rob} we get for any fixed $\tau>0$
	\[
	(\star) \leq 2\sum_{x \in  B_{S_O}} \sum_{\substack{S \in \mathcal{P}_c(\Lambda) \\ \text{s.t. } x \in S}} \Vert A_S(\omega \tau) \Vert_{\mathrm{op}} \Vert O \Vert_{\mathrm{op}} + \sum_{x \in \Lambda \setminus B_{S_O}} \sum_{\substack{S \in \mathcal{P}_c(\Lambda) \\ \text{s.t. } x \in S}} \Vert A_S(\omega \tau) \Vert_{\mathrm{op}} \Vert O \Vert_{\mathrm{op}} |S_O| e^{- \kappa (d(S,S_O)-v \tau))} = :(*)\,.
	\]
	 For any $x \in S$:
	\[
	\begin{split}
	d(S,S_O)-v\tau & \geq d(x,S_O)-v\tau-|S| \geq d(x, Q_{S_O}) - v \tau - |S| = d(x,B_{S_O})-|S|
	\end{split}
	\]
	which means
	\[
	e^{-\kappa(d(S,S_O)-v\tau)} \leq e^{-\kappa d(x,B_{S_O})} e^{\kappa |S|}
	\]
	and then
	\[
	(*) \leq 2\sum_{x \in  B_{S_O}} \sum_{\substack{S \in \mathcal{P}_c(\Lambda) \\ \text{s.t. } x \in S}} \Vert A_S(\omega\tau) \Vert_{\mathrm{op}} \Vert O \Vert_{\mathrm{op}}+ \sum_{x \in \Lambda \setminus B_{S_O}} \sum_{\substack{S \in \mathcal{P}_c(\Lambda) \\ \text{s.t. } x \in S}} \Vert A_S (\omega\tau) \Vert_{\mathrm{op}} \Vert O \Vert_{\mathrm{op}} |S_O| e^{\kappa |S|} e^{-\kappa d(x,B_{S_O})} \, .
	\]
	For any $\ell \in \mathbb{N}$, let $C_\ell  = \{x \in \Lambda \,  | \, \ell < d(x,B_{S_O}) \leq \ell+1 \}$. Thus,
	\[
	\begin{split}
	\sum_{x \in \Lambda \setminus B_{S_O}} \sum_{\substack{S \in \mathcal{P}_c(\Lambda) \\ \text{s.t. } x \in S}} &\Vert A_S(\omega\tau) \Vert_{\mathrm{op}} \Vert O \Vert_{\mathrm{op}} |S_O| e^{\kappa |S|} e^{-\kappa d(x,B_{S_O})}  = \sum_{\ell \geq 0} \sum_{x \in C_\ell} \sum_{\substack{ S \in \mathcal{P}_c(\Lambda) \\ \text{s.t. } x \in S}} \Vert A_S(\omega\tau) \Vert_{\mathrm{op}} \Vert O \Vert_{\mathrm{op}} |S_O| e^{\kappa|S|} e^{-\kappa(\ell-1)} \\
	&\leq \sum_{\ell \geq 0} |S_O| (|S_O|+v\tau+\ell+1)^d e^{-\kappa (\ell-1)} \sup_{x \in \Lambda} \sum_{\substack{S \in \mathcal{P}_c(\Lambda) \\ \text{s.t. } x \in S}} \Vert A_S(\omega\tau) \Vert_{\mathrm{op}} \Vert O \Vert_{\mathrm{op}} e^{\kappa |S|} \\
	& \leq |S_O| (S_O+v\tau)^d e^{-1} \sum_{\ell\geq 0} \left( 1+\frac{\ell+1}{|S_O|+v\tau} \right)^d e^{-\kappa \ell} \sup_{x \in \Lambda} \sum_{\substack{S \in \mathcal{P}_c(\Lambda) \\ \text{s.t. } x \in S}} \Vert A_S(\omega\tau) \Vert_{\mathrm{op}} \Vert O \Vert_{\mathrm{op}} e^{\kappa |S|}\,.
	\end{split}
	\]
	Since the series in $\ell$ is convergent, we define $C(d,\kappa):= e \sum_{\ell \geq 0} (2+\ell)^d e^{-\kappa \ell}$ and we obtain
	\[
	\sum_{x \in \Lambda \setminus B_{S_O}} \sum_{\substack{S \in \mathcal{P}_c(\Lambda) \\ \text{s.t. } x \in S}} \Vert A_S(\omega\tau) \Vert_{\mathrm{op}} \Vert O \Vert_{op} |S_O| e^{\kappa |S|} e^{-\kappa d(x,B_{S_O})} \leq |S_O| (|S_O|+v \tau)^d C(d,\kappa) \Vert O \Vert_{\mathrm{op}} \Vert A(\omega\tau) \Vert_{\kappa} \, .
	\]
	Then, for fixed $\tau$, since $C(d,\kappa)\geq 1,$ one has
	\[
	\begin{aligned}
	\Vert [A(\omega\tau), e^{-\ii \tau Z} O e^{\ii \tau Z} ] \Vert_{\mathrm{op}} &\leq 2{(|S_O+v\tau)^d}\Vert O \Vert_{\mathrm{op}} \Vert A(\omega\tau) \Vert_{0} + |S_O| (|S_O|+v \tau)^d C(d,\kappa) \Vert O \Vert_{\mathrm{op}} \Vert A(\omega\tau) \Vert_{\kappa}\\
	& \leq 3|S_O| (|S_O|+v \tau)^d C(d,\kappa) \Vert O \Vert_{\mathrm{op}} \Vert A(\omega\tau) \Vert_{\kappa}.
	\end{aligned}
	\]
	Recalling that $\sup_{\tau \in \mathbb{R}} \Vert A(\omega\tau)\Vert_\kappa \leq \Vert A \Vert_{\kappa,0}$, one has
	\[
	\int_0^t \ud \tau \Vert[A(\omega\tau),e^{-\ii \tau Z} O e^{\ii \tau Z} ] \Vert_{\mathrm{op}} \leq 3|S_O|\Vert O \Vert_{\mathrm{op}} \Vert A \Vert_{\kappa,0} \frac{C(d,\kappa)}{(d+1) v(d,\kappa,\Vert Z \Vert_{2 \kappa})} \left[(|S_O|+v(d,\kappa,\Vert Z \Vert_{2 \kappa}) t)^{d+1}-|S_O|^{d+1}  \right]\,,
	\]
	Through elementary computations one sees that
	\[
		\frac{(|S_O|+y)^{d+1}-|S_O|^{d+1}}{y} \leq 2^{d+1} (d+1) |S_O|^{d+1} \langle y\rangle^{d+1} \, \qquad \forall y >0 \, .
	\]
	Putting $y=v(Z,\kappa)t$, and recalling that $v(Z,\kappa)=C(d)\kappa^{-{(d+2)}}e^{\kappa}\Vert Z \Vert_{2 \kappa}$, one gets
	\[
		\begin{split}
			\int_0^t \ud \tau \Vert[A,e^{-\ii \tau Z} O e^{\ii \tau Z} ] \Vert_{\mathrm{op}} &\leq 3 \cdot 2^{d+1} |S_0|^{d+2} C(d,\kappa) \langle tv(\kappa,Z) \rangle^{d+1} t \Vert O \Vert_{\mathrm{op}} \Vert A \Vert_{\kappa} \\
			&\leq C(d,\kappa,|S_O|)\langle \Vert Z \Vert_{2\kappa} \rangle^{d+1} \langle t \rangle^{d+1} t \Vert O \Vert_{\mathrm{op}} \Vert A \Vert_{\kappa} \, ,
		\end{split}
	\]
	with $C(|S_O|,d,\kappa)=3\cdot 2^{d+1} |S_O|^{d+2}C(d,\kappa) \big(C(d) \kappa^{-{(d+2)}} e^{\kappa}\big)^{d+1}$.

\footnotesize

\paragraph{\footnotesize Data availability statement.} This manuscript has no associated data.

\paragraph{\footnotesize Conflicts of interests/Competing interests.} The authors have no competing interests
to declare that are relevant to the content of this article.


\begin{thebibliography}{10}

\bibitem{Abanin2017}
{\sc D.~Abanin, W.~De~Roeck, W.~W. Ho, and F.~Huveneers}, {\em {A Rigorous
  Theory of Many-Body Prethermalization for Periodically Driven and Closed
  Quantum Systems}}, Communications in Mathematical Physics, 354 (2017),
  pp.~809--827.

\bibitem{Baldi_Berti}
{\sc P.~Baldi and M.~Berti}, {\em Forced vibrations of a nonhomogeneous
  string}, SIAM J. Math. Anal., 40 (2008), pp.~382--412.

\bibitem{Dario2}
{\sc D.~Bambusi and A.~Giorgilli}, {\em Exponential stability of states close
  to resonance in infinite-dimensional Hamiltonian systems}, Journal of
  Statistical Physics, 71 (1993), pp.~569--606.

\bibitem{Dario1}
{\sc D.~Bambusi, A.~Giorgilli, S.~Paleari, and T.~Penati}, {\em Normal form and
  energy conservation of high frequency subsystems without nonresonance
  conditions}, Istituto Lombardo - Accademia di Scienze e Lettere - Rendiconti
  di Scienze,  (2013) 147. https://doi.org/10.4081/scie.2013.172.

\bibitem{BGMR_growth}
{\sc D.~Bambusi, B.~Gr\'{e}bert, A.~Maspero, and D.~Robert}, {\em Growth of
  {S}obolev norms for abstract linear {S}chr\"{o}dinger equations}, J. Eur.
  Math. Soc. (JEMS), 23 (2021), pp.~557--583.

\bibitem{QNbari}
{\sc D.~Bambusi and B.~Langella}, {\em Globally integrable quantum
	systems and their perturbations}, arXiv:2403.18670. Accepted for publication on ``{S}ingularities, {A}symptotics and {L}imiting {M}odels" {IN}d{AM} {S}pringer {Volume}, 2024 pp. 64-103.

\bibitem{QN}
{\sc D.~Bambusi and B.~Langella}, {\em Growth of {S}obolev norms in quasi
	integrable quantum systems}, arXiv:2202.04505, 2022.
	
\bibitem{NekCInfinito}
{\sc D.~Bambusi and B.~Langella}, {\em A $C^\infty$ Nekhoroshev theorem}, Mathematics in Engineering, 3(2) (2021), pp.~1--17.

\bibitem{BambusiMaspero2016}
{\sc D.~Bambusi and A.~Maspero}, {\em Birkhoff coordinates for the {T}oda
  lattice in the limit of infinitely many particles with an application to
  {FPU}}, Journal of Functional Analysis, 270 (2016), pp.~1818--1887.

\bibitem{Bambusi2006}
{\sc D.~Bambusi and A.~Ponno}, {\em On metastability in {FPU}}, Communications
  in Mathematical Physics, 264 (2006), pp.~539--561.

\bibitem{Benettin_Fro_Giorgilli}
{\sc G.~Benettin, J.~Fr\"{o}hlich, and A.~Giorgilli}, {\em A {N}ekhoroshev-type
  theorem for {H}amiltonian systems with infinitely many degrees of freedom},
  Communications in Mathematical Physics, 119 (1988), pp.~95--108.

\bibitem{Benettin_Gallavotti}
{\sc G.~Benettin and G.~Gallavotti}, {\em Stability of motions near resonances
  in quasi-integrable {H}amiltonian systems}, Journal of Statistical Physics, 44 (1986),
  pp.~293--338.

\bibitem{Benettin2023-vw}
{\sc G.~Benettin and A.~Ponno}, {\em {FPU} model and {T}oda model: A survey, a
  view}, Springer INdAM series, Springer Nature
  Singapore, Singapore, 2023, pp.~21--44.

\bibitem{Berges2004}
{\sc J.~Berges, S.~Bors{\'{a}}nyi, and C.~Wetterich}, {\em Prethermalization},
  Physical Review Letters, 93 142002 (2004).

\bibitem{Berti2003}
{\sc M.~Berti, L.~Biasco, and P.~Bolle}, {\em Drift in phase space: a new
  variational mechanism with optimal diffusion time}, Journal de
  Math{\'{e}}matiques Pures et Appliqu{\'{e}}es, 82 (2003), pp.~613--664.

\bibitem{Bertini2015}
{\sc B.~Bertini, F.~H. Essler, S.~Groha, and N.~J. Robinson}, {\em
  Prethermalization and thermalization in models with weak integrability
  breaking}, Physical Review Letters, 115 180601 (2015).

\bibitem{bourgain.e.altri.amici}
{\sc J.~Bourgain, F.~Golse, and B.~Wennberg}, {\em On the distribution of free
  path lengths for the periodic {L}orentz gas}, Communications in Mathematical
  Physics, 190 (1998), pp.~491--508.

\bibitem{Boyers2020-exp}
{\sc E.~Boyers, P.~J. Crowley, A.~Chandran, and A.~O. Sushkov}, {\em Exploring
  2d synthetic quantum hall physics with a quasiperiodically driven qubit},
  Physical Review Letters, 125 160505 (2020).

\bibitem{Carati_Maiocchi}
{\sc A.~Carati and A.~M. Maiocchi}, {\em Exponentially long stability times for
  a nonlinear lattice in the thermodynamic limit}, Communications in
  Mathematical Physics, 314 (2012), pp.~129--161.

\bibitem{Collura2022}
{\sc M.~Collura, A.~D. Luca, D.~Rossini, and A.~Lerose}, {\em Discrete
  time-crystalline response stabilized by domain-wall confinement}, Physical
  Review X, 12 031037 (2022).

\bibitem{Corsi_Genovese}
{\sc L.~Corsi and G.~Genovese}, {\em Periodic driving at high frequencies of an
  impurity in the isotropic {XY} chain}, Communications in Mathematical
  Physics, 354 (2017), pp.~1173--1203.

\bibitem{DAlessio2014}
{\sc L.~D'Alessio and M.~Rigol}, {\em Long-time behavior of isolated
  periodically driven interacting lattice systems}, Physical Review X, 4 041048
  (2014).

\bibitem{DeRoeck-Verreet}
{\sc W.~De~Roeck and V.~Verreet}, {\em {Very slow heating for weakly driven
  quantum many-body systems}}, arXiv:1911.01998, 2019.

\bibitem{Eckhardt2022}
{\sc C.~J. Eckhardt, G.~Passetti, M.~Othman, C.~Karrasch, F.~Cavaliere, M.~A.
  Sentef, and D.~M. Kennes}, {\em Quantum {F}loquet engineering with an exactly
  solvable tight-binding chain in a cavity}, Communications Physics 5, 122 (2022).

\bibitem{Else2020}
{\sc D.~V. Else, W.~W. Ho, and P.~T. Dumitrescu}, {\em Long-lived interacting
  phases of matter protected by multiple time-translation symmetries in
  quasiperiodically driven systems}, Physical Review X, 10 021032 (2020).

\bibitem{Ferguson1982}
{\sc W.~Ferguson, H.~Flaschka, and D.~McLaughlin}, {\em Nonlinear normal modes
  for the {T}oda chain}, Journal of Computational Physics, 45 (1982),
  pp.~157--209.

\bibitem{FPU-Original}
{\sc E.~Fermi, J.~Pasta, and S.~Ulam}, {\em Studies of nonlinear problems}, Los-Alamos Internal Report, Document
LA-1940 (1955). In \emph{Enrico Fermi Collected Papers}, Vol. II, The University of Chicago Press, Chicago, and Accademia Nazionale dei Lincei, Roma, 1965, pp. 977–988.

\bibitem{Franzoi}
{\sc L.~Franzoi}, {\em Reducibility for a linear wave equation with {S}obolev
  smooth fast driven potential}, arXiv:2301.08009,  2023.

\bibitem{Franzoi_Maspero}
{\sc L.~Franzoi and A.~Maspero}, {\em Reducibility for a fast-driven linear
  {K}lein-{G}ordon equation}, Ann. Mat. Pura Appl. (4), 198 (2019),
  pp.~1407--1439.

\bibitem{Gallone2022PRL}
{\sc M.~Gallone, M.~Marian, A.~Ponno, and S.~Ruffo}, {\em Burgers turbulence in
  the {Fermi-Pasta-Ulam-Tsingou} chain}, Physical Review Letters, 129 114101 (2022).

\bibitem{Gallone2021}
{\sc M.~Gallone, A.~Ponno, and B.~Rink}, {\em Korteweg{\textendash}de Vries and
  Fermi{\textendash}Pasta{\textendash}Ulam{\textendash}Tsingou: asymptotic
  integrability of quasi unidirectional waves}, Journal of Physics A:
  Mathematical and Theoretical, 54 305701 (2021) .
  
\bibitem{GioPisa}
{\sc A.~Giorgilli}, {\em Notes on exponential stability of Hamiltonian systems}. In: Dynamical Systems, Part I: Hamiltonian systems and Celestial Mechanics. Pubblicazioni del Centro di Ricerca Matematica Ennio De Giorgi, Pisa (2003), pp. 87–198.

\bibitem{Grava2020}
{\sc T.~Grava, A.~Maspero, G.~Mazzuca, and A.~Ponno}, {\em Adiabatic invariants
  for the {FPUT} and {T}oda chain in the thermodynamic limit}, Communications
  in Mathematical Physics, 380 (2020), pp.~811--851.

\bibitem{Gring2012}
{\sc M.~Gring, M.~Kuhnert, T.~Langen, T.~Kitagawa, B.~Rauer, M.~Schreitl,
  I.~Mazets, D.~A. Smith, E.~Demler, and J.~Schmiedmayer}, {\em Relaxation and
  prethermalization in an isolated quantum system}, Science, 337 (2012),
  pp.~1318--1322.

\bibitem{Guanghui}
{\sc G.~He, B.~Ye, R.~Gong, Z.~Liu, K.~W. Murch, N.~Y. Yao, and C.~Zu}, {\em
  Quasi-{F}loquet prethermalization in a disordered dipolar spin ensemble in
  diamond}, Phys. Rev. Lett. 131, 130401 (2023).

\bibitem{Ho-DeRoeck-2020}
{\sc W.~W. Ho and W.~De~Roeck}, {\em {A Rigorous Theory of Prethermalization
  without Temperature}}, arXiv:2011.14583, 2020.

\bibitem{Howell2019}
{\sc O.~Howell, P.~Weinberg, D.~Sels, A.~Polkovnikov, and M.~Bukov}, {\em
  Asymptotic prethermalization in periodically driven classical spin chains},
  Physical Review Letters, 122 010602 (2019).

\bibitem{Huveneers2020}
{\sc F.~Huveneers and J.~Lukkarinen}, {\em Prethermalization in a classical
  phonon field: Slow relaxation of the number of phonons}, Physical Review
  Research, 2 022034(R) (2020).

\bibitem{Kollar2011}
{\sc M.~Kollar, F.~A. Wolf, and M.~Eckstein}, {\em Generalized {G}ibbs ensemble
  prediction of prethermalization plateaus and their relation to nonthermal
  steady states in integrable systems}, Physical Review B, 84 054304 (2011).

\bibitem{Lapierre2020}
{\sc B.~Lapierre, K.~Choo, A.~Tiwari, C.~Tauber, T.~Neupert, and R.~Chitra},
  {\em Fine structure of heating in a quasiperiodically driven critical quantum
  system}, Physical Review Research, 2 033461  (2020).

\bibitem{Lazarides2014}
{\sc A.~Lazarides, A.~Das, and R.~Moessner}, {\em Equilibrium states of generic
  quantum systems subject to periodic driving}, Physical Review E, 90 012110 (2014).

\bibitem{Lieb1972}
{\sc E.~H. Lieb and D.~W. Robinson}, {\em The finite group velocity of quantum
  spin systems}, Communications in Mathematical Physics, 28 (1972),
  pp.~251--257.

\bibitem{Lindner2017}
{\sc N.~H. Lindner, E.~Berg, and M.~S. Rudner}, {\em Universal chiral
  quasisteady states in periodically driven many-body systems}, Physical Review
  X, 7 011018 (2017).

\bibitem{Long2021}
{\sc D.~M. Long, P.~J. Crowley, and A.~Chandran}, {\em Nonadiabatic topological
  energy pumps with quasiperiodic driving}, Physical Review Letters, 126 106805
  (2021).

\bibitem{Long2022PRB}
{\sc D.~M. Long, P.~J.~D. Crowley, and A.~Chandran}, {\em Many-body
  localization with quasiperiodic driving}, Physical Review B, 105 144204  (2022).


\bibitem{Lochak}
{\sc P.~Loshak}, {\em Canonical perturbation theory: an approach based on joint
  approximations}, Uspekhi Mat. Nauk, 47 (1992), pp.~59--140.

\bibitem{Maiocchi2014}
{\sc A.~Maiocchi, D.~Bambusi, and A.~Carati}, {\em An averaging theorem for
  {FPU} in the thermodynamic limit}, Journal of Statistical Physics, 155
  (2014), pp.~300--322.

\bibitem{Mallayya2019}
{\sc K.~Mallayya, M.~Rigol, and W.~De~Roeck}, {\em {Prethermalization and
  Thermalization in Isolated Quantum Systems}}, Physical Review X, 9 021027 (2019).

\bibitem{Malz2021}
{\sc D.~Malz and A.~Smith}, {\em Topological two-dimensional {F}loquet lattice
  on a single superconducting qubit}, Physical Review Letters, 126 163602 (2021).

\bibitem{Martin2017PRX}
{\sc I.~Martin, G.~Refael, and B.~Halperin}, {\em Topological frequency
  conversion in strongly driven quantum systems}, Physical Review X, 7 041008 (2017).

\bibitem{Martin2022}
{\sc T.~Martin, I.~Martin, and K.~Agarwal}, {\em Effect of quasiperiodic and
  random noise on many-body dynamical decoupling protocols}, Physical Review B,
  106 134306 (2022).

\bibitem{Moeckel2008}
{\sc M.~Moeckel and S.~Kehrein}, {\em Interaction quench in the {H}ubbard
  model}, Physical Review Letters, 100 175702 (2008).

\bibitem{Monti_Jauslin}
{\sc F.~Monti and H.~R. Jauslin}, {\em Quantum {N}ekhoroshev theorem for
  quasi-periodic {F}loquet {H}amiltonians}, Rev. Math. Phys., 10 (1998),
  pp.~393--428.

\bibitem{Mori2016}
{\sc T.~Mori, T.~Kuwahara, and K.~Saito}, {\em {Rigorous Bound on Energy
  Absorption and Generic Relaxation in Periodically Driven Quantum Systems}},
  Physical Review Letters, 116 120401 (2016).

\bibitem{Nek77}
{\sc N.~N. Nehoro\v{s}ev}, {\em An exponential estimate of the time of
  stability of nearly integrable {H}amiltonian systems}, Uspehi Mat. Nauk, 32
  (1977), pp.~5--66, 287.

\bibitem{Nek79}
{\sc N.~N. Nehoro\v{s}ev}, {\em An exponential
  estimate of the time of stability of nearly integrable {H}amiltonian systems.
  {II}}, Trudy Sem. Petrovsk.,  (1979), pp.~5--50.

\bibitem{Ponte2015}
{\sc P.~Ponte, A.~Chandran, Z.~Papi{\'{c}}, and D.~A. Abanin}, {\em
  Periodically driven ergodic and many-body localized quantum systems}, Annals
  of Physics, 353 (2015), pp.~196--204.

\bibitem{PonPon}
{\sc J.~P\"{o}schel}, {\em Small divisors with spatial structure in
  infinite-dimensional {H}amiltonian systems}, Communications in Mathematical
  Physics, 127 (1990), pp.~351--393.

\bibitem{Poschel_nek}
{\sc J.~P\"{o}schel}, {\em Nekhoroshev estimates for quasi-convex Hamiltonian
  systems}, Mathematische Zeitschrift, 213 (1993), pp.~187--216.

\bibitem{Potter2016}
{\sc A.~C. Potter, T.~Morimoto, and A.~Vishwanath}, {\em Classification of
  interacting topological {F}loquet phases in one dimension}, Physical Review
  X, 6 041001 (2016).

\bibitem{Qi2021}
{\sc Z.~Qi, G.~Refael, and Y.~Peng}, {\em Universal nonadiabatic energy pumping
  in a quasiperiodically driven extended system}, Physical Review B, 104
 224301  (2021).

\bibitem{RubioAbadal2020}
{\sc A.~Rubio-Abadal, M.~Ippoliti, S.~Hollerith, D.~Wei, J.~Rui, S.~Sondhi,
  V.~Khemani, C.~Gross, and I.~Bloch}, {\em Floquet prethermalization in a
  {Bose-Hubbard} system}, Physical Review X, 10 021044 (2020).

\bibitem{Wayne1}
{\sc C.~E. Wayne}, {\em The {KAM} theory of systems with short range
  interactions. {I}, {II}}, Communications in Mathematical Physics, 96 (1984),
  pp.~311--329, 331--344.

\bibitem{Wayne3}
{\sc C.~E. Wayne}, {\em Bounds on the
  trajectories of a system of weakly coupled rotators}, Communications in
  Mathematical Physics, 104 (1986), pp.~21--36.

\bibitem{Ye2021}
{\sc B.~Ye, F.~Machado, and N.~Y. Yao}, {\em {F}loquet phases of matter via
  classical prethermalization}, Physical Review Letters, 127 140603 (2021).

\bibitem{Zabusky1965}
{\sc N.~J. Zabusky and M.~D. Kruskal}, {\em Interaction of ``solitons'' in a
  collisionless plasma and the recurrence of initial states}, Physical Review
  Letters, 15 (1965), pp.~240--243.

\bibitem{Zhang2022}
{\sc L.~Zhang and X.-J. Liu}, {\em Unconventional {F}loquet topological phases
  from quantum engineering of band-inversion surfaces}, {PRX} Quantum, 3 040312 
  (2022).

\bibitem{Zhang2022e}
{\sc X.~Zhang, W.~Jiang, J.~Deng, K.~Wang, J.~Chen, P.~Zhang, W.~Ren, H.~Dong,
  S.~Xu, Y.~Gao, F.~Jin, X.~Zhu, Q.~Guo, H.~Li, C.~Song, A.~V. Gorshkov,
  T.~Iadecola, F.~Liu, Z.-X. Gong, Z.~Wang, D.-L. Deng, and H.~Wang}, {\em
  Digital quantum simulation of {F}loquet symmetry-protected topological
  phases}, Nature, 607 (2022), pp.~468--473.

\bibitem{Zhao2022bis}
{\sc H.~Zhao, J.~Knolle, R.~Moessner, and F.~Mintert}, {\em Suppression of
  interband heating for random driving}, Physical Review Letters, 129 120605  (2022).

\bibitem{Zhao2021}
{\sc H.~Zhao, F.~Mintert, R.~Moessner, and J.~Knolle}, {\em Random multipolar
  driving: Tunably slow heating through spectral engineering}, Physical Review
  Letters, 126 040601 (2021).

\end{thebibliography}

\end{document}